%% file: main_archive.tex
\documentclass{article}

\usepackage{arxiv}

\usepackage[utf8]{inputenc} 
\usepackage[T1]{fontenc}    
\usepackage{hyperref}       
\usepackage{url}            
\usepackage{booktabs}       
\usepackage{amsfonts}       
\usepackage{amsmath}
\usepackage{amssymb}
\usepackage{amsthm}
\usepackage{nicefrac}       
\usepackage{microtype}      
\usepackage{lipsum}		
\usepackage{graphicx}
\usepackage{natbib}
\usepackage{doi}
\usepackage{xcolor}
\usepackage{texmacro_maths}
\usepackage{subcaption}
\usepackage{float}

\usepackage{thmtools}    

\declaretheorem[name=Theorem,numberwithin=section]{theorem}
\declaretheorem[name=Lemma,numberwithin=section]{lemma}

\declaretheorem[name=Definition,numberwithin=section]{definition}

\newtheorem{example}[theorem]{Example}

\providecommand{\Description}[2][]{}

\newcommand{\md}{ \mathrm{MD} }

\newcommand{\stem}{ \mathrm{Stem} }
\newcommand{\ds}{\text{Down-Stem}}
\newcommand{\ex}{ \mathrm{ex} }

\DeclareMathOperator{\shell}{1-shell}

\newcommand{\rel}[1]{\stackrel{#1}{\sim}}

\title{Reducing Sensor Requirements by Relaxing the Network Metric Dimension}

\author{ \href{https://orcid.org/0009-0005-2812-865X}{\includegraphics[scale=0.06]{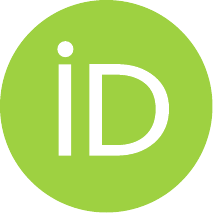}\hspace{1mm}Paula Mürmann*} \\
	EPFL, Switzerland\\
	\texttt{paula.murmann@epfl.ch} \\
	\And
	\href{https://orcid.org/0009-0003-9758-0573}{\includegraphics[scale=0.06]{orcid.pdf}\hspace{1mm}Robin Jaccard*}\\
	EPFL, Switzerland\\
	\texttt{robin.jaccard@epfl.ch} \\
    \And
	\href{https://orcid.org/0000-0001-6613-0615}{\includegraphics[scale=0.06]{orcid.pdf}\hspace{1mm}Maximilien Dreveton*}\\
	EPFL, Switzerland\\
	\texttt{maximilien.dreveton@epfl.ch} \\
    \And
	\href{https://orcid.org/0009-0005-3877-6992}{\includegraphics[scale=0.06]{orcid.pdf}\hspace{1mm}Aryan Alavi Razavi Ravari}\\
	Sharif University of Technology, Tehran, Iran\\
	\texttt{aryan.alaviarar@gmail.com} \\
    \And
	\href{https://orcid.org/0009-0001-8176-6854}{\includegraphics[scale=0.06]{orcid.pdf}\hspace{1mm}Patrick Thiran}\\
	EPFL, Switzerland\\
	\texttt{patrick.thiran@epfl.ch} \\
}

\renewcommand{\shorttitle}{Reducing Sensor Requirements by Relaxing the Network Metric Dimension}

\hypersetup{
pdftitle={Reducing Sensor Requirements by Relaxing the Network Metric Dimension},
}

\begin{document}
\maketitle
\def\thefootnote{*}\footnotetext{Authors contributed equally to this research.}\def\thefootnote{\arabic{footnote}}

\begin{abstract}
	Source localization in graphs involves identifying the origin of a phenomenon or event, such as an epidemic outbreak or a misinformation source, by leveraging structural graph properties. One key concept in this context is the metric dimension, which quantifies the minimum number of strategically placed sensors needed to uniquely identify all vertices based on their distances. While powerful, the traditional metric dimension imposes a stringent requirement that every vertex must be uniquely identified, often necessitating a large number of sensors. 

In this work, we relax the metric dimension and allow vertices at a graph distance less than $k$ to share identical distance profiles relative to the sensors. This relaxation reduces the number of sensors needed while maintaining sufficient resolution for practical applications like source localization and network monitoring. We provide two main theoretical contributions: an analysis of the $k$-relaxed metric dimension in deterministic trees, revealing the interplay between structural properties and sensor placement, and an extension to random trees generated by branching processes, offering insights into stochastic settings. 

We also conduct numerical experiments across a variety of graph types, including random trees, random geometric graphs, and real-world networks. For graphs with loops, we use a greedy algorithm to obtain an upper-bound on the relaxed metric dimension. The results show that the relaxed metric dimension is significantly smaller than the traditional metric dimension. Furthermore, the number of vertices indistinguishable from any given target vertex always remains small. Finally, we propose and evaluate a two-step localization strategy that balances the trade-off between resolution and the number of sensors required. This strategy identifies an optimal relaxation level that minimizes the total number of sensors across both steps, providing a practical and efficient approach to source localization.
\end{abstract}


\input{text.tex}

\bibliographystyle{plain}
\bibliography{main.bib} 

\appendix
\input{appendix}

\end{document}

%% file: text.tex
\section{Introduction}

 Source localization in graphs is a critical problem with applications in diverse fields such as epidemiology, social network analysis, and cyber-security. It focuses on identifying the origin of a phenomenon or event, such as the initial vertex of an epidemic~\cite{spinelli2018many}, the originator of misinformation~\cite{shah2010rumors}, the source of a cyber-attack~\cite{10.1145/3485447.3512184}, the position of a robot moving on a graph~\cite{khuller1996landmarks}. Accurate source localization enables timely interventions, such as containment of disease, countering misinformation, or mitigating cyber-threats. It can also be used in deanonymization attacks of diffusion spreading protocols in Bitcoin P2P networks~\cite{fanti2017deanonymization}. However, the problem is challenging because of the complexity of graph structures, the dynamics of propagation, and the limited or noisy observation data. Addressing these challenges requires efficient methods that can exploit the underlying properties of graphs to achieve accurate localization.

 One such property is the \new{metric dimension} of a graph, which measures the minimum number of strategically chosen vertices, called sensors or landmarks, needed to uniquely identify all other vertices based on their distances to the sensors. This concept is directly applicable to source localization, as distances from a suspected source to these sensors uniquely determine the source. Graphs with a smaller metric dimension require fewer sensors, reducing the observation cost. By leveraging the metric dimension, systems can be designed to optimize the placement of sensors or observers, enabling accurate and cost-effective localization in various networks, from biological systems to communication infrastructures. 

 The metric dimension has been extensively studied across theoretical and applied domains in graph theory. Early research focused on characterizing the metric dimension for specific graph families, such as paths, grids, trees~\cite{slater_leaves_1975,khuller1996landmarks,Harary_1976,MD_for_tree}, and Cartesian products of graphs~\cite{caceres2007metric}, revealing how structural properties influence the number of required sensors. Over time, greedy algorithms have been developed to approximate the metric dimension for general graphs, addressing the NP-hard nature of the problem~\cite{khuller1996landmarks}. Recent advancements have explored the metric dimension on random graphs, such as \Erdos-\Renyi random graphs~\cite{bollobas2013metric}, random geometric graphs~\cite{lichev2023localization}, and random trees~\cite{komjathy_Odor_GW_MD,mitsche2015limiting}. 

However, the definition of metric dimension imposes the strong requirement that \textit{every} vertex must be uniquely identified based on its distances to the sensors. This stringent condition often leads to the need for a disproportionately high number of sensors. Moreover, in many real-world applications, perfect source identification may not be necessary, and approximate or partial identification can suffice. For instance, in source localization, identifying a small set of candidate vertices, rather than pinpointing the exact origin of a diffusion, may still enable effective interventions. Similarly, in network monitoring, distinguishing between groups of vertices with similar properties or roles can often achieve the desired outcomes. 
One classical example of this is the trade-off between passive and active surveillance in infectious disease detection \cite{Murray_Cohen_Encyclopedia_PubH, Tan_act_vs_pass_epi_surv}. Aiming to detect a pathogen outbreak as quickly as possible, public health organizations would preferably do targeted and direct testing of a population at risk. However this is often too resource intensive, which leads to the use of passive epidemic surveillance data. In contrast, passively acquired data may not provide as accurate a picture on current public health, but it is cheaper and more easily obtained.

 To reduce the number of sensors, we relax the metric dimension by allowing vertices within a graph distance less than $k$ from each other to share identical distance profiles with respect to the sensors. We call \new{$k$-relaxed metric dimension} the minimum number of sensors needed to satisfy this condition. This relaxation acknowledges that in many practical scenarios, distinguishing between closely located nodes is either unnecessary or infeasible because of data noise or computational constraints. By grouping vertices that are close to each other in the graph, we reduce the number of required sensors while still preserving sufficient resolution for tasks such as source localization or network monitoring. 

 Our work makes two main theoretical contributions. First, we study the relaxed metric dimension in deterministic trees, characterizing how the structural properties of trees influence the number and placement of sensors required under the relaxed conditions. Second, we extend this study to random trees generated by branching processes, providing insights into the behavior of the relaxed metric dimension in stochastic settings. By combining the deterministic and stochastic perspectives, we establish a comprehensive understanding of how deterministic and random tree topologies affect the trade-off between resolution and the number of sensors needed in the relaxed framework. 

 To complement the theoretical findings, we conduct numerical experiments to explore how the $k$-relaxed metric dimension compares to the traditional metric dimension across different synthetic and real graphs. Our numerical results indicate that the relaxed metric dimension is significantly smaller than the non-relaxed metric dimension for many graph classes, including random trees, random geometric graphs, and real-world networks. Interestingly, while the number of vertices that are not uniquely identified by the sensors can be large, the size of the largest equivalence class of vertices (\textit{i.e.}, the subset of vertices having identical distances to the sensors) remains small. In the context of source localization, this means that although a small number of sensors may no longer uniquely identify the source, they confine it to a small set of candidate vertices. 

 This observation opens the door to a two-step localization strategy. In the first step, sensors are placed to distinguish vertices at distances strictly larger than $k$. This results in a set of candidate vertices for the target. In the second step, additional sensors are deployed to uniquely identify the target within this candidate set. The critical question becomes: \textit{how much can this two-step approach reduce the total number of sensors compared to a one-step strategy where sensors uniquely identify all vertices outright?} Our experiments on both synthetic and real graphs reveal a clear trade-off. If the relaxation is too small, the first step requires too many sensors. Conversely, if the relaxation is too large, the first step fails to distinguish enough vertex pairs, and the second step becomes excessively demanding. Between these extremes, we identify an optimal value of the relaxation parameter $k$ such that the two-step approach optimally balances the number of sensors needed in each phase. 

 Finally, while the proofs of our main results are lengthy and sometimes tedious, their underlying intuition is both simple and enlightening. To conclude the introduction, we provide an overview of the key theoretical results with a high-level summary of the proof techniques employed.

\paragraph{Overview of the results on deterministic trees}
The metric dimension of a tree is often large because distinguishing the leaves is difficult. Consider the tree $T$ given in Figure~\ref{fig:Tree_Stem_Ex1}. Notice that the two leaves $v_1$ and $v_2$ are equidistant from any other vertex $v' \notin \{v_1,v_2\}$. Consequently, a sensor must be placed at either $v_1$ or $v_2$ in order to distinguish $v_1$ from $v_2$. More generally, for each exterior major vertex~$v$--defined as a vertex of degree at least~$3$ and adjacent to one or more leaf paths\footnote{A leaf path is a path of degree two vertices to a degree one vertex. We refer to Section~\ref{subsection:majorVertices} for examples.}--sensors must be placed at all but one of the leaves adjacent to $v$. 
This is the intuition leading to the following result from~\cite{slater_leaves_1975}. Consider a tree $T$. If $T$ is a line graph (a path), then its metric dimension is $1$, and is achieved by placing one sensor at the extremity of the path; otherwise, its metric dimension $\md(T)$ is given by 
\begin{equation}
\label{eq:formula_md_tree}
    \md(T) = \sigma(T) - \ex(T), 
\end{equation} 
where $\sigma(T)$ is the number of leaves in $T$ and $\ex(T)$ is the number of exterior major vertices in~$T$.

\begin{figure}[!ht]
    \centering
    \begin{subfigure}[b]{0.30\textwidth}
        \centering
        \includegraphics[width=\textwidth]{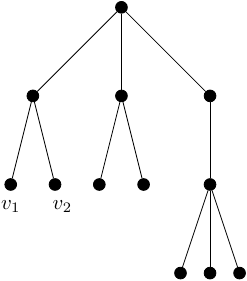}
        \caption{Before stemming}
        \label{fig:Tree_Stem_Ex1}
    \end{subfigure}
    \hfil
    \begin{subfigure}[b]{0.30\textwidth}
        \centering
        \includegraphics[width=\textwidth]{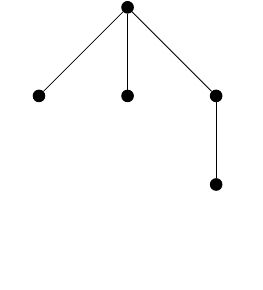}
        \caption{After stemming}
        \label{fig:Tree_Stem_Ex2}
    \end{subfigure}
    \caption{Example of a tree before and after stemming. 
    }
    \Description[Example of a tree before and after stemming]{The left figure shows a tree $T$, and the right figure show the stem $\stem(T)$ of this tree, where the leaves have been removed.} 
    \label{fig:Tree_Stem_Ex}
\end{figure}

Now, suppose we relax the metric dimension so that only pairs of vertices that are at a distance strictly larger than 1 need to be distinguished. Since two leaves in a tree are always at a distance larger than 1, this relaxation does not reduce the number of sensors required to distinguish the leaves, and is therefore helpless. This observation leads to a result on the relationship of odd to even distance relaxations:
\begin{align*}
 \md_{2r}(T) \weq \md_{2r+1}(T),
\end{align*}
where $\md_{k}(T)$ denotes the $k$-relaxed metric dimension of $T$. We refer to Theorem~\ref{thm:odd_dist} for a formal statement. 

Next, suppose we relax the requirement further so that only pairs of vertices at a distance strictly larger than 2 need to be distinguished. Consider again the tree $T$ from Figure~\ref{fig:Tree_Stem_Ex1}. The leaves $v_1$ and~$v_2$, being at distance 2 from each other, do no longer need to be distinguished. The same applies to any pair of leaves attached to the same exterior major vertex. As a result, we can delete (stem) all the leaves of $T$, resulting in the smaller tree $T'$ shown in Figure~\ref{fig:Tree_Stem_Ex2}. We call \new{stemming} the process of deleting leaves, and in Theorem~\ref{thm:k_metric_dim_tree}, we prove that
\begin{align*}
 \md_{2r}(T) \weq \md(\stem_r(T)),
\end{align*}
where $\stem_r(T)$ is the tree obtained by applying the stemming process iteratively $r$ times, starting from $T$. Because $\stem_r(T)$ is also a tree, its metric dimension can be computed using formula~\eqref{eq:formula_md_tree}. 

 \paragraph{Overview of the main results on random trees}

 A \new{critical Galton-Watson tree} is a random tree generated by a branching process where the expected number of offspring per individual is exactly~1. This ensures that the tree neither grows explosively nor dies out too rapidly, striking a balance that produces structures of interest in many applications. The critical Galton-Watson tree is particularly useful in modeling processes such as population dynamics, genealogies, and network structures in cases where growth is self-sustaining but not super-linear. 

 Conditioned on having $n$ vertices, a critical Galton-Watson tree with a Poisson offspring distribution has the same distribution as a uniform random tree (which is a tree selected uniformly at random among the $n^{n-2}$ labeled trees with $n$ vertices). Moreover, a conditioned Galton–Watson tree with an offspring distribution~$\xi'$ verifying $\E[\xi'] \ne 1$ is, in most cases, equivalent to a conditioned Galton–Watson tree with another offspring distribution~$\xi$ satisfying $\E [\xi] = 1$ (we refer to~\cite[Section~4]{janson2012simply} to make this equivalence precise).  This makes the conditioned Galton-Watson tree a natural and analytically convenient framework for studying properties of random trees.

The metric dimension of Galton-Watson trees can be expressed as the difference between the number of leaves and the number of exterior major vertices, using Equation~\eqref{eq:formula_md_tree}. Thus, deriving an asymptotic expression for the metric dimension of such trees reduces to answering two key questions: how many leaves and how many exterior major vertices does a Galton-Watson tree have? At first glance, these two questions may seem unrelated, appearing to require distinct analytical approaches. However, this is not the case. Let $v$ be a vertex chosen uniformly at random from a rooted random tree $T$. The \textit{fringe tree} $T_v$ is the subtree of $T$ rooted at $v$ that includes all its descendants. Because both $T$ and $v$ are random, the fringe tree $T_v$ is itself a random tree.\footnote{We describe here the \textit{quenched} version, where we first fix a realization of the random tree $T$ and then choose $v$ uniformly at random from its vertices. We refer the reader to~\cite[Remark~1.2]{janson_fringe_tree_2016} for additional details.} 

Interestingly, the proportion of vertices that are leaves is equal to the probability that the fringe tree consists of a single vertex. Similarly, the proportion of vertices that are exterior major vertices can be related to the probability that the root of the fringe tree has degree at least two and at least one of its subtrees is a line graph to a leaf.  More generally, many seemingly distinct properties of a random tree can in fact be analyzed through the distribution of fringe trees. This distribution provides valuable insights into the typical shapes and sizes of subtrees, making it a powerful tool for understanding properties of random trees.  We refer the reader to \cite{aldous_fringe_tree_1991} for general results on fringe tree distributions and focus here on the specific case of conditioned Galton-Watson trees. For a critical Galton-Watson tree conditioned to have $n$ vertices, the fringe tree distribution converges in distribution to an \textit{unconditioned} Galton-Watson tree with the \textit{same} offspring distribution. A proof of this convergence with the additional assumption that the offspring distribution has a bounded second moment can be found in~\cite[Lemma~9]{aldous_fringe_tree_1991}. For a more general proof without this assumption, we refer to \cite[Theorem~1.3]{janson_fringe_tree_2016}. 

 Using fringe tree distributions, \cite{komjathy_Odor_GW_MD} computes the number of leaves and the number of exterior major vertices of random trees and establishes a law of large number for the metric dimension of a broad class of random trees. 
 
 Obtaining an analogous result for the $2r$-relaxed metric dimension requires us to analyze the stem of fringe trees rather than the fringe trees themselves, introducing additional complexity. These efforts culminate in one of our main results, Theorem~\ref{thm:rlaxed_md_random_trees}, which provides a detailed characterization of the relaxed metric dimension in random trees. 

\paragraph{Main notations} Throughout the paper, $G = (V,E)$ denotes a (unweighted, undirected) graph with $n = |V|$ vertices, and $T$ denotes a tree. For two vertices $u$ and $v$ belonging to a graph $G$, we denote $d_G(u,v)$ (and simply $d(u,v)$ when no confusion is possible) the graph distance between $u$ and $v$. the $n$-by-1 vector whose entries are all equal to 1 is denoted by $\mathbf{1}_n$. Finally, we use the convention that $0^0=1$. 
For convenience, we provide in Appendix~\ref{appendix:table_notations} a table of notations.

\paragraph{Code availability} The code to reproduce the simulations is available at \url{https://github.com/mdreveton/metric-dimension-relaxed/} 

\paragraph{Paper Organization} The paper is organized as follows. Section~\ref{section:definitions} introduces the key definitions and notations. In Section~\ref{section:theory}, we present the theoretical results concerning the relaxed metric dimension of arbitrary trees and Galton-Watson trees. Numerical results are provided in Section~\ref{section:experiments}. The proof techniques required to establish the theoretical results are discussed in Section~\ref{section:proof_overview}. Finally, Section~\ref{section:conclusion} concludes the paper.

\section{Definitions and Related Contents}
\label{section:definitions}

\subsection{Relaxed Metric Dimensions}

Let $G = (V,E)$ be an undirected, connected graph. Given two vertices $u, v \in V$, we denote by $d(u,v)$ the number of edges on a shortest path from $u$ to $v$. For an ordered subset of vertices $S = ( s_1, \cdots, s_{|S|} ) \subset V$, we denote by 
\[ 
\Phi_G(u, S) \weq \left( d(u,s_1), \cdots, d(u,s_{|S|} ) \right) 
\]
the vector whose entries are the lengths of the shortest paths between a vertex $u$ and all the vertices in $S$. 
We call $\Phi_G(u, S)$ the \new{identification vector} of $u$ with respect to~$S$ in $G$. See Figure~\ref{fig:ex_rs_a} for an example.
We say that vertices $u$ and $v$ are \new{distinguished} by $S$ if $\Phi_G(u, S) \ne \Phi_G(v, S)$.

\begin{definition}
    A subset $S \subseteq V$ is called a \new{resolving set} if for all pairs of vertices $u \neq v \in V$, $u$ and $v$ are distinguished by $S$, \textit{i.e.}, for all $ u, v  \in V$, we have 
    \begin{equation*}
        \Phi_G(u, S) = \Phi_G(v, S) \implies u = v.
    \end{equation*}    
    The \new{metric dimension} is defined as the the minimum cardinality of a resolving set, denoted by $\md(G)$.
\end{definition}

We consider a more lenient variant of the metric dimension where vertices that are close to one another do not need to be distinguished. 

\begin{definition}
    A subset $S \subseteq V$ is called a \new{$k$-relaxed resolving set} if for all $ u, v  \in V$, we have 
    \begin{equation*}
     \Phi_G(u, S) = \Phi_G(v, S) \implies d(u,v) \leq k.
    \end{equation*}
    Moreover, the \new{$k$-relaxed metric dimension} is defined as the the minimum cardinality of a $k$-relaxed resolving set, denoted by $\md_k(G)$.
\end{definition}
Observe that when $k=0$, we recover the definition of a resolving set (as $d(u,v) \le 0$ is equivalent to $u=v$). 
In particular, a resolving set is a 0-relaxed resolving set, and $\md(G) = \md_0(G)$. We show in Figure~\ref{fig:ex_rs} an example of a $0$-relaxed resolving set and of a $2$-relaxed resolving set of minimum cardinality on a toy graph. 

\begin{figure}[!ht]
    \centering
    \begin{subfigure}{0.40\textwidth}
    \centering
        \includegraphics[width=0.7\linewidth]{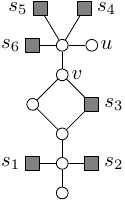}
        \caption{0-relaxed resolving set $S = \{s_1, ...,s_6\}$\newline
         and the identification vector for example of vertex $u$ $\Phi_G(u,S) = (6,6, 3, 2, 2, 2 )$.}
    \label{fig:ex_rs_a}
    \end{subfigure}
    \hfil
    \begin{subfigure}{0.40\textwidth}
    \centering
        \includegraphics[width=0.7\linewidth]{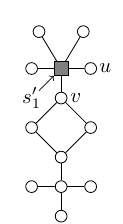}
        \caption{2-relaxed resolving set $S_{k=2} = \{s_1'\}$ \newline
         and the identification vector for example of vertex $u$ $\Phi_G(u,S_{k=2}) = (1)$.}
    \label{fig:ex_rs_b}
    \end{subfigure}
    \caption{Example of a graph with a 0-relaxed resolving set and 2-relaxed resolving set of minimum cardinality. (The resolving set is denoted by gray squares.)}
    \Description[0-relaxed resolving set and 2-relaxed resolving set of a graph.]{Example of a graph with a 0-relaxed resolving set and 2-relaxed resolving set of minimum cardinality. }
    \label{fig:ex_rs}
\end{figure}

\subsection{Major Vertices and Leaf Paths}
\label{subsection:majorVertices}

We now recall the notion of exterior major vertex. We call \new{major vertex} of $G$ a vertex of degree at least 3 in a graph $G$. We define a \new{leaf path} $G_L$ as the path between a leaf (a vertex of degree 1) and its closest major vertex.  A major vertex $v$ of $G$ is an \new{exterior major vertex} if it has at least one adjacent leaf path. See  Figure \ref{fig:MX_example} for an example. Additionally, we define $\sigma_k(G)$ as the sum of the leaves in $\stem_k(G)$, and $\ex_k(G)$ as the number of exterior major vertices in $\stem_k(G)$.

\begin{figure}[!ht]
    \centering
    \includegraphics[width=0.5\linewidth]{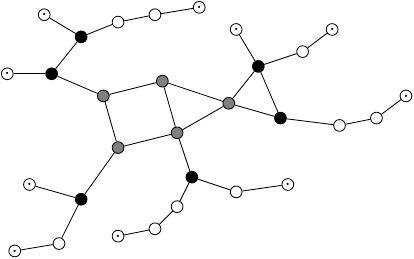}
    \caption{Example of a graph $G$. The leaves of $G$ are denoted by a dot $\odot$. The major vertices are the vertices highlighted in {\color{gray}{gray}} and \textbf{black}. Moreover, vertices in \textbf{black} are exterior major vertices (vertices in {\color{gray}{gray}} are the major vertices which are not exterior major vertices).}
    \Description[Example graph showing leaves, major vertices and exterior major vertices]{A general graph with marked leaves (degree 1 vertices), major vertices (vertices with degree larger or equal to 3) and a subset of major vertices that are also exterior (they also have an adjacent leaf path).}
    \label{fig:MX_example}
\end{figure}

\subsection{Stemming a Graph}

Trees typically have a large metric dimension because all leaves must be uniquely distinguished. As discussed in the previous section, this requires placing all but one leaf per exterior major vertex into the resolving set. However, when considering a $2r$-relaxed metric dimension, it is no longer necessary to distinguish leaves that share a common ancestor within $r$ hops. This observation can be formalized as a simple iterative pruning process, which we call \new{stemming}. 

\begin{definition}[$r$-Stem]
 We define $\stem(G)$, the \new{stem} of a graph $G$, as the unique subgraph of $G$ induced by all vertices of degree strictly larger than $1$. We further denote by $\stem_r(G)$ the $r$-fold application of the stemming operation on $G$. By convention, $\stem_0(G) = G$. 
\end{definition}

In other words, the stem of a graph $G$ is the subgraph $\stem(G) \subseteq G$ in which all the vertices of degree~$0$ and $1$ are removed. In particular, stemming a tree is equivalent to removing all its leaves.

\section{Relaxed Metric Dimensions of Trees}
\label{section:theory}

We provide in Section~\ref{subsection:deterministicTrees} the expression for the relaxed metric dimension of arbitrary trees. Next, we derive, in Section~\ref{subsection:randomTrees}, the asymptotic expression for the relaxed metric dimension for a class of random trees.

\subsection{Relaxed Metric Dimensions of Arbitrary Trees}
\label{subsection:deterministicTrees}

 We first show that the $2r$ and $2r+1$ relaxed metric dimensions of a tree are equal. Decreasing the relaxation parameter $k$ increases the relaxed metric dimension because it requires resolving more vertices. Hence, the inequality $\md_{2r+1}(G) \le \md_{2r}(G)$ holds for any graph $G$. To establish the reverse inequality in the case of trees, consider a set of sensors $S \subset V$ and a vertex $u \in V \setminus S$. When transitioning from a $(2r+1)$-relaxation to a $(2r)$-relaxation, $u$ must now be distinguished from all vertices at a distance $2r+1$ from itself. However, in a tree, any vertex $v\in S$ inherently distinguishes between two vertices at odd distances. Therefore, we can reduce the relaxation parameter from $2r+1$ to $2r$ without adding more vertices to the resolving set. This leads to the following theorem, whose proof is provided in Section~\ref{sec:prf_odd_dist}.

\begin{theorem} 
\label{thm:odd_dist}
Let $T$ be a tree and $r$ a non-negative integer such that $2r+1 < \diam(T)$. We have $\md_{2r}(T) = \md_{2r+1}(T)$. 
\end{theorem}

 Let us now turn to the key finding that links the relaxed metric dimension on trees to the ordinary metric dimension of the stemmed tree. As discussed in the introduction, the ordinary metric dimension of a tree can be described in terms of the leaves of the tree and its exterior major vertices, namely, if $T$ is not a line graph,
 \begin{align}\label{eq:md_tree}
  \md(T) \weq \sigma(T) - \ex(T).
 \end{align}
 (Note that $\md(T) = 1$ iff $T$ is a line graph). 
 We wish to express the $2r$-relaxed metric dimension of a tree $T$ in terms of the ordinary metric dimension of the $k$-Stem of $T$. Hence, we introduce the notations $\sigma_r(T)$ and $\ex_r(T)$ as the number of leaves and exterior major vertices in $\stem_r(T)$, respectively. 

\begin{theorem}\label{thm:k_metric_dim_tree} 
 Let $T$ be a tree and $r$ a nonnegative integer such that $2r < \diam(T)$. Then 
 \begin{align*}
   \md_{2r}(T) \weq \begin{cases}
       1  &\text{if } \stem_r(T) \text{ is a line graph,}\\
       \sigma_r(T) - \ex_r(T) & \text{otherwise}.
   \end{cases}
 \end{align*}
\end{theorem}

We describe the proof techniques used to establish this theorem in Section~\ref{sec:prf_k_metric_dim_tree}. 
Theorem~\ref{thm:k_metric_dim_tree} can be equivalently stated by $\md_{2r}(T) = \md(\stem_r(T))$, and then using Formula~\eqref{eq:md_tree} to obtain $\md(\stem_r(T)) = 1$ if $\stem_r(T)$ is a line graph and $\md(\stem_r(T)) = \sigma_r(T) - \ex_r(T)$ otherwise.

Combining Theorems~\ref{thm:odd_dist} and~\ref{thm:k_metric_dim_tree} provides an expression of the $k$-relaxed metric dimension of trees, for any~$k < \diam(T)$, while the case $k \ge \diam(T)$ trivially gives $\md_k(T) = 0$. 

Importantly, the relaxed metric dimension can be orders of magnitude smaller than the metric dimension, which yields a drastic reduction in number of requires sensors. In the particular case of a full $m$-ary tree, the following Example~\ref{exa:full_m_ary_tree} shows that a $2r$-relaxation of the metric dimension reduces the number of sensors by a factor $m^r$. 

\begin{example}
 \label{exa:full_m_ary_tree}
 Let $T$ be a full $m$-ary tree of height $h$. Then, for any nonnegative integer $r$ such that $2r < h$, we have $\sigma_r(T) = m^{h-r}$ and $\ex_r(T) = m^{h-r-1}$. Thus, the ratio $\frac{\md_0(T)}{\md_{2r}(T)}$ simplifies to $\frac{m^h - m^{h-1}}{m^{h-r} - m^{h-r-1}} = m^r$, and we obtain 
 $$\md_{2r}(T) \weq \frac{ \md_{0}(T) }{ m^r }.$$ 
\end{example}

\subsection{Metric Dimension of Random Trees}\label{subsection:randomTrees}

 In Section~\ref{subsection:deterministicTrees}, we provided an expression of the relaxed metric dimension for \textit{any} tree. To explore more broadly how many sensors can be saved by using a relaxed metric dimension, we examine its asymptotic behavior on critical Galton-Watson trees. Let us first recall the definition of a Galton-Watson tree. 

\begin{definition}
 \label{def:GW_trees}Let $\xi$ be a probability distribution over the set of nonnegative integers. 
 A \new{Galton-Watson tree} with offspring distribution $\xi$ is a random rooted tree constructed recursively as follows:
 \begin{itemize}
   \item the root has a random number of offsprings distributed according to $\xi$;
   \item each offspring of the root independently has a random number of offspring distributed according to $\xi$;
   \item this branching process continues independently for all subsequent generations.
 \end{itemize}
\end{definition}
In the following, we focus on the case $\E [ \xi ] = 1$, referred to as the \new{critical} Galton-Watson trees, and we condition the trees to have $n$ vertices. We refer to the introduction for a motivation of these choices. The following theorem provides an asymptotic for the relaxed metric dimension of critical Galton-Watson trees. 

\begin{theorem}
 \label{thm:rlaxed_md_random_trees}
 Consider a sequence $(\mathcal{T}^{GW}_n)_{n \in \mathbb{N}}$ of critical Galton-Watson trees conditioned on having $n$~vertices with offspring distribution $\xi$ satisfying $\E\left[\xi\right] = 1$ and $\E\left[\xi^2\right] < \infty$. Then, the $(2r)$-relaxed metric dimension of $\mathcal{T}^{GW}_n$ is given by: 
 \begin{align*}
   \frac{\md_{2r} \left( \cT^{GW}_n \right) }{n} \xrightarrow{p} c_r(\xi), 
 \end{align*}
 as $n \to \infty$, where $c_r(\xi)$ is a constant depending only on the distribution $\xi$.
\end{theorem}
An overview of the proof of Theorem~\ref{thm:rlaxed_md_random_trees} is given in Section~\ref{sec:prf_randomTrees}. Our proof provides explicit (albeit complicated) expressions for the coefficients $c_r(\xi)$, which depend solely on the offspring distribution $\xi$. As an example, consider the case where the offspring distribution is Poisson. This is an important example, where the Galton-Watson tree is distributed as a uniform random tree.
 
\begin{example}
 Let $\xi$ be the Poisson distribution with mean $1$, \textit{i.e.,} $\mathbb{P}(\xi = i) = e^{-1}/ i!$. 
 Denote
 \begin{align*}
   d_r \weq 
   \begin{cases}
           0 & \text{if } r=0, \\
           e^{d_{r-1}-1} & \text{otherwise,}
   \end{cases} \quad \text{ and } \quad 
   \ell_r \weq 
   \begin{cases}
    1/e & \text{if } r=0, \\
    d_r(e^{\ell_{r-1}} -1) & \text{otherwise.}
   \end{cases}
   \end{align*}
 Denote also $e_r = 1- e^{-s_r} -(s_r - \ell_r)$ and $s_r = \ell_r / (1 - e^{d_{r} -1} ) $.  An interpretation of the coefficients $d_r$, $\ell_r$, $e_r$ and $s_r$ is provided in Section~\ref{sec:prf_randomTrees}. The values $c_r(\mathrm{Poisson}(1))$ are given by 
   \begin{equation*}
     c_r(\xi) \weq s_r +e^{-s_r} - 1.
   \end{equation*}
    Table~\ref{tab:MD_Poisson} provides the numeric values of the coefficients $c_r$ for various values of $r$. Observe that the value of $c_r$ drops significantly from $r=0$ to $r=1$ (by a factor 3), and again from $r=1$ to $r=2$, but a bit less (by a factor 2). This hints again at significant savings in sensor resources for small values of the relaxation parameter~$r$.

\begin{table}[!ht]
\centering
\begin{tabular}{@{}ccccccccccc@{}}
\toprule
\( r \)          & 0      & 1      & 2      & 3      & 4      & 5      & 6      & 7      & 8      & 9      \\ \midrule
$c_r(\xi)$ & 0.1408 & 0.0544 & 0.0294 & 0.0185 & 0.0128 & 0.0094 & 0.0072 & 0.0057 & 0.0046 & 0.0038 \\ \bottomrule
\end{tabular}
\caption{Limit of $\md_{2r}/n$ for a Galton-Watson tree with \textit{Poisson} offspring distribution and mean $1$}
\label{tab:MD_Poisson}
\end{table}
\end{example}

\section{Numerical Experiments}
\label{section:experiments}

\subsection{Algorithm and Metrics}

\subsubsection{Approximating the relaxed metric dimension}

 We saw in Section~\ref{section:theory} that the relaxed metric dimension can be computed exactly and in a constructive manner on trees. However, computing the metric dimension in an arbitrary graph is NP-hard~\cite{khuller1996landmarks}, and the same can be expected for the $k$-relaxed version. Thus, to perform experiments on arbitrary graphs, we use a greedy algorithm that computes an approximate value of the relaxed metric dimension. We provide in Appendix~\ref{appendix:algo} the pseudo-code of the greedy algorithm (Algorithm~\ref{algo}) used to compute the relaxed metric dimension in all our experiments. 

  For any integer $k$, we denote by $\hS_k(G)$ the output of Algorithm~\ref{algo} applied on a graph $G=(V,E)$. Because Algorithm~\ref{algo} returns a $k$-relaxed resolving set, we have $|\hS_k(G)| \ge \md_k(G)$. Moreover, Theorem~\ref{thm:validity_algo} in Appendix shows that $|\hS_k(G)| / \md_k(G) \in O(\log |V|)$. 
  
  \subsubsection{Metrics used}
  \label{sec:metrics_used}
  To evaluate the effect of the relaxation parameter, we consider three different metrics. 
  
  (i) The first metric is the proportion of vertices belonging to the $k$-relaxed resolving set $\hS_k(G)$ found by Algorithm~\ref{algo}, as a function of $k$. It indicates the savings (in number of sensors) gained by resolving only pairs of vertices at distance larger than $k$, instead of resolving all the vertices ($k=0$).
  
  (ii) The second metric is the proportion of non-resolved vertices as a function of $k$. The larger this number, the lower the chance to uniquely detect the source by the sensors. 
  
  (iii)  To define the third and last metric, let us introduce some notations. For two vertices $u,v \in V$ and any set $S \subseteq V$, we denote $u \rel{S} v$ if $\Phi_{G}(u,S) = \Phi_{G}(v,S)$. Observe that $\rel{S}$ is an equivalence relationship over $V$, and the equivalent classes are the set of vertices that have the same identification vector. In other words, if the source is located at $u$, then the set
 \begin{align*}
  [u]_S \weq \{ v \in V \colon \Phi_{G}(v,S) = \Phi_{G}(u,S) \} 
 \end{align*}
 is the set of potential \textit{candidate vertices} for the source location. Therefore, the last metric we consider is the size $\alpha$ of the largest set of candidate vertices with respect to an arbitrary location of the source, namely
\begin{align*}
 \alpha 
 \weq \max_{u \in V} \left| \left[ u \right]_S  \right| 
 \weq \max_{u \in V} \left| \left\{ v \in V \colon \Phi_{G}(v,S) = \Phi_{G}(u,S) \right\} \right|. 
\end{align*}
This quantity $\alpha$ is the cardinality of the largest set of vertices having the same identification vector. 

\subsection{Synthetic Datasets}

 We provide results for standard random graph models: two types of random trees (\Barabasi-Albert and Galton-Watson models) and two types of random graphs (configuration model and random geometric graphs). For Galton-Watson trees, we choose the offspring distribution to be Poisson with mean 3. For the configuration model, we choose the degree sequence to be iid distributed from $2+X$ where $X$ denotes a Zipf distribution with exponent 2.5 and maximum value $n-2$.\footnote{This choice ensures that the minimum degree is larger or equal than 3, and hence the sampled graph is almost surely connected.} Finally, for random geometric graphs, the points are uniformly distributed on the $2$-d square $[0,1]^2$ and two vertices at distance less than $r = 1.5 r_c$ with $r_c = \sqrt{ \log(n) / (n \pi ) }$ are connected. Note that because $r > r_c$, the graphs sampled are almost surely connected. 
 
 Table~\ref{table:synthetic_graphs_statistics} lists some statistics of the random graphs considered. Among the standard statistics (such as average degree and diameter), we also report the size of the 1-shell. Recall that the $2$-core of a graph $G$ is the subgraph of $G$ where every vertex in the subgraph has a degree of at least $2$, and the $1$-shell is the set of vertices that do not belong to the $2$-core. In particular, the vertices of the $1$-shell are exactly the vertices that are removed by the iterative stemming operation. 
 
 \begin{table}[!ht]
\centering
\begin{tabular}{ c c c c c c c  }
\hline \toprule
 Network & $|V|$ & $|E|$ & $\bd$ & $D_{\max}$ & $\bar{D}$ & $|\shell|$ \\
 \midrule
\Barabasi-Albert & 1000 & 999 (0) & 2.0(0) & 20.3 (1.8) & 8.40 (0.60) & 1000 (0) \\
Galton-Watson & 1000 & 999 (0) & 2.0 (0) & 23.0 (3.0) & 15.2 (1.5) & 1000 (0) \\
configuration model & 1000 & 1908 (53) & 3.8 (0.11) & 9.11 (0.60) & 4.86 (0.33) & 0.7 (0.82) \\
random geometric graph & 1000 & 7314 (106) & 14.6 (0.21) & 24.5 (0.69) & 9.56 (0.12) & 0.25 (0.50) \\
\bottomrule
\end{tabular}
 \caption{Number of vertices $|V|$, number of edges $|E|$, average degree $\bd$, diameter $D_{\max}$, average shortest-path length $\bar{D}$, and size $|\shell|$ of the 1-shell of the random graph models when the number of vertices is set to 1000. In parenthesis: standard deviations.}
\label{table:synthetic_graphs_statistics}
\end{table}

Because the results for different models are similar, we describe them in the main text only for \Barabasi-Albert and random geometric graphs. We postpone to Appendix~\ref{appendix:additional_numerical_results} additional numerical results for the other classes of networks. 

 \subsubsection{\Barabasi-Albert trees}\label{subsection:experiments_BA}  
Figure~\ref{fig:evolution_md_BA} illustrates the evolution of the relaxed metric dimension in \Barabasi-Albert trees. As predicted by Theorem~\ref{thm:odd_dist}, the values for $k=2r$ and $k=2r+1$ are identical. This serves as an initial indication that, even though it is a greedy algorithm outputting an approximate solution, Algorithm~\ref{algo} reliably identifies resolving sets with the smallest cardinality on trees. Additionally, we observe a sharp decrease in the metric dimension when it is relaxed from $k=0$ to $k=2$. Although the proportion of non-distinguished vertices increases significantly (rising from 0\% to 60\%), the size of the largest equivalence class remains relatively small, with fewer than 20 vertices out of 1000 for $k=2$.

\begin{figure}[!ht]
 \centering
 \begin{subfigure}{0.32\textwidth}
  \includegraphics[width=\linewidth]{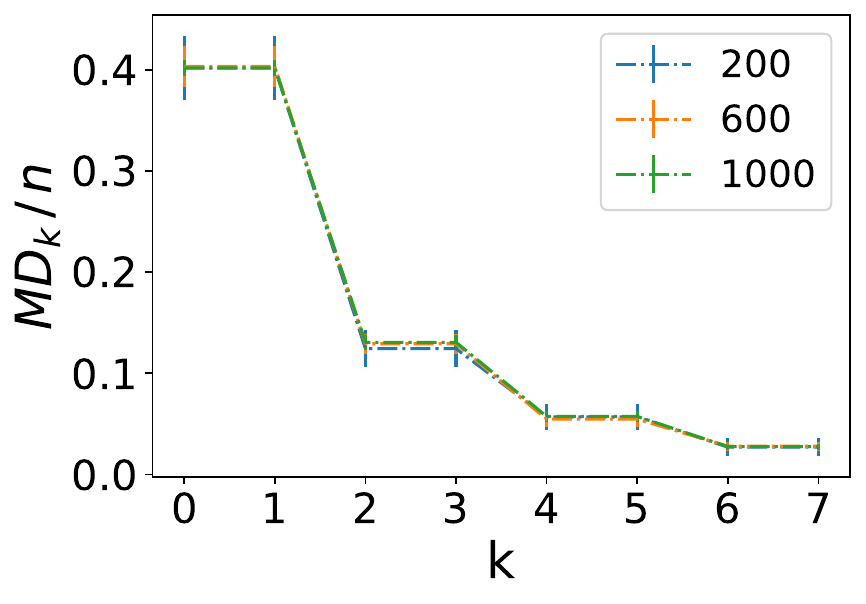}
  \caption{$\md_k / n$}
 \end{subfigure}
 \begin{subfigure}{0.32\textwidth}
  \includegraphics[width=\linewidth]{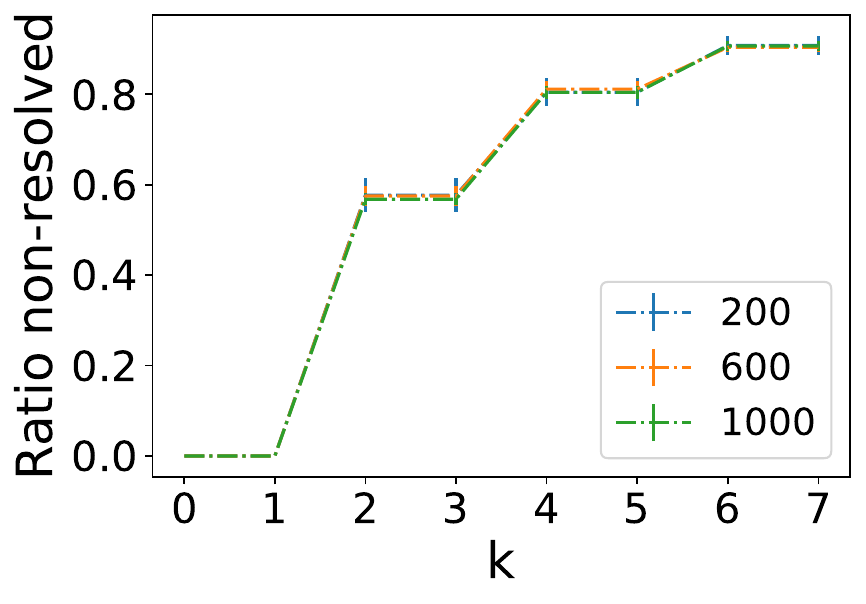}
  \caption{Ratio of non-resolved vertices}
 \end{subfigure}
 \begin{subfigure}{0.32\textwidth}
   \includegraphics[width=\linewidth]{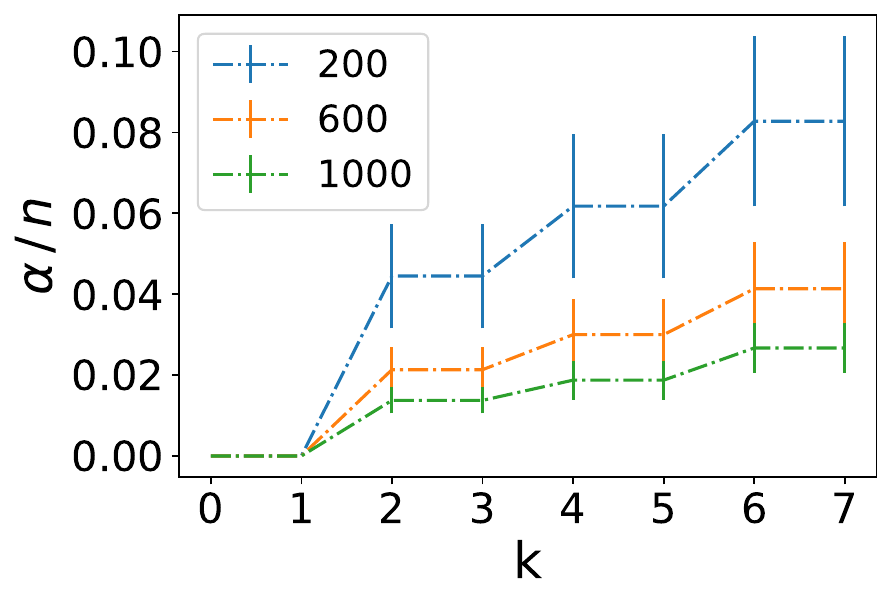}
  \caption{ $\alpha / n$} 
 \end{subfigure}
 \caption{\Barabasi-Albert random trees. Results are averaged over $20$ realizations, and error bars show the standard deviation.}
 \label{fig:evolution_md_BA}
\end{figure}

To gain a deeper insight, Figure~\ref{fig:visualization_BA} illustrates relaxed resolving sets generated by Algorithm~\ref{algo} for an instance of a \Barabasi-Albert tree, limited to 100 vertices for clarity. Even though Algorithm~\ref{algo} is a greedy algorithm, it consistently identifies resolving sets with the smallest possible cardinality—an observation that held true across all our simulations on random trees. Furthermore, when the metric dimension is relaxed, the equivalence classes of the non-resolved vertices predominantly consist of leaves connected to the same vertex. This last observation explains the small size of the largest equivalence class. 

\begin{figure}[!ht]
    \centering
    \begin{subfigure}{0.32\textwidth}
        \includegraphics[width=\linewidth]{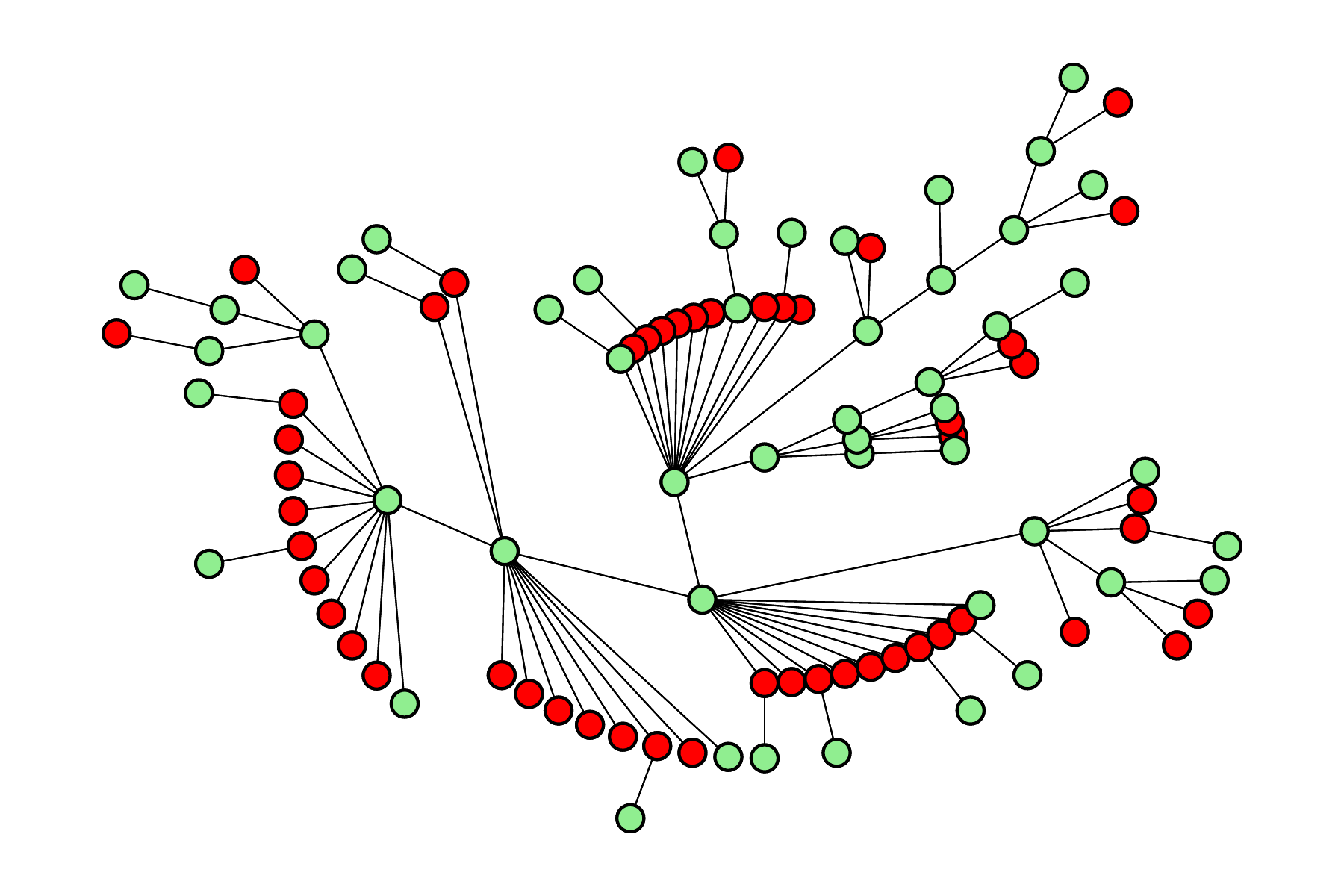}
        \caption{$k=0$}
    \end{subfigure}
    \begin{subfigure}{0.32\textwidth}
        \includegraphics[width=\linewidth]{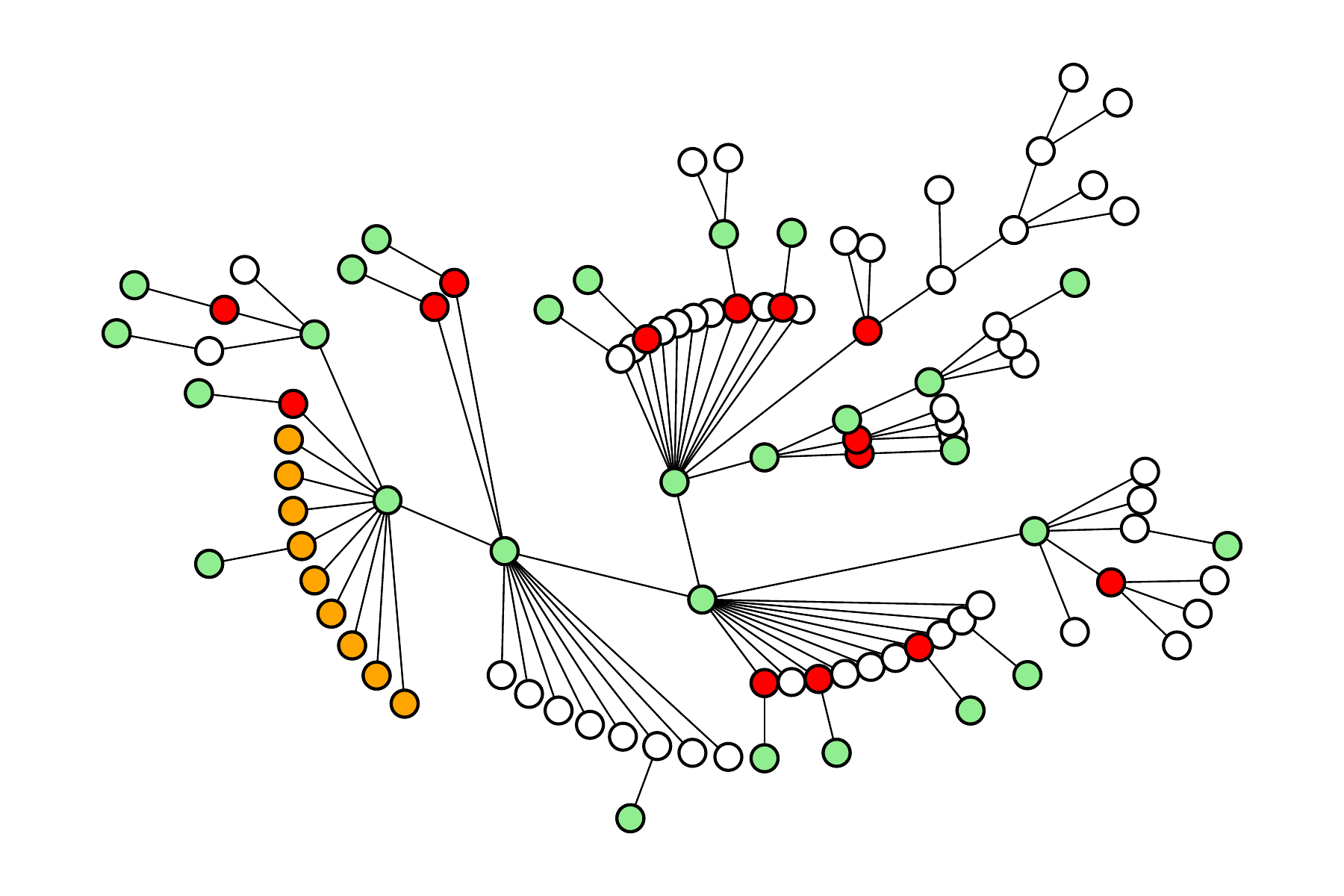}
        \caption{$k=2$}
    \end{subfigure}
    \begin{subfigure}{0.32\textwidth}
        \includegraphics[width=\linewidth]{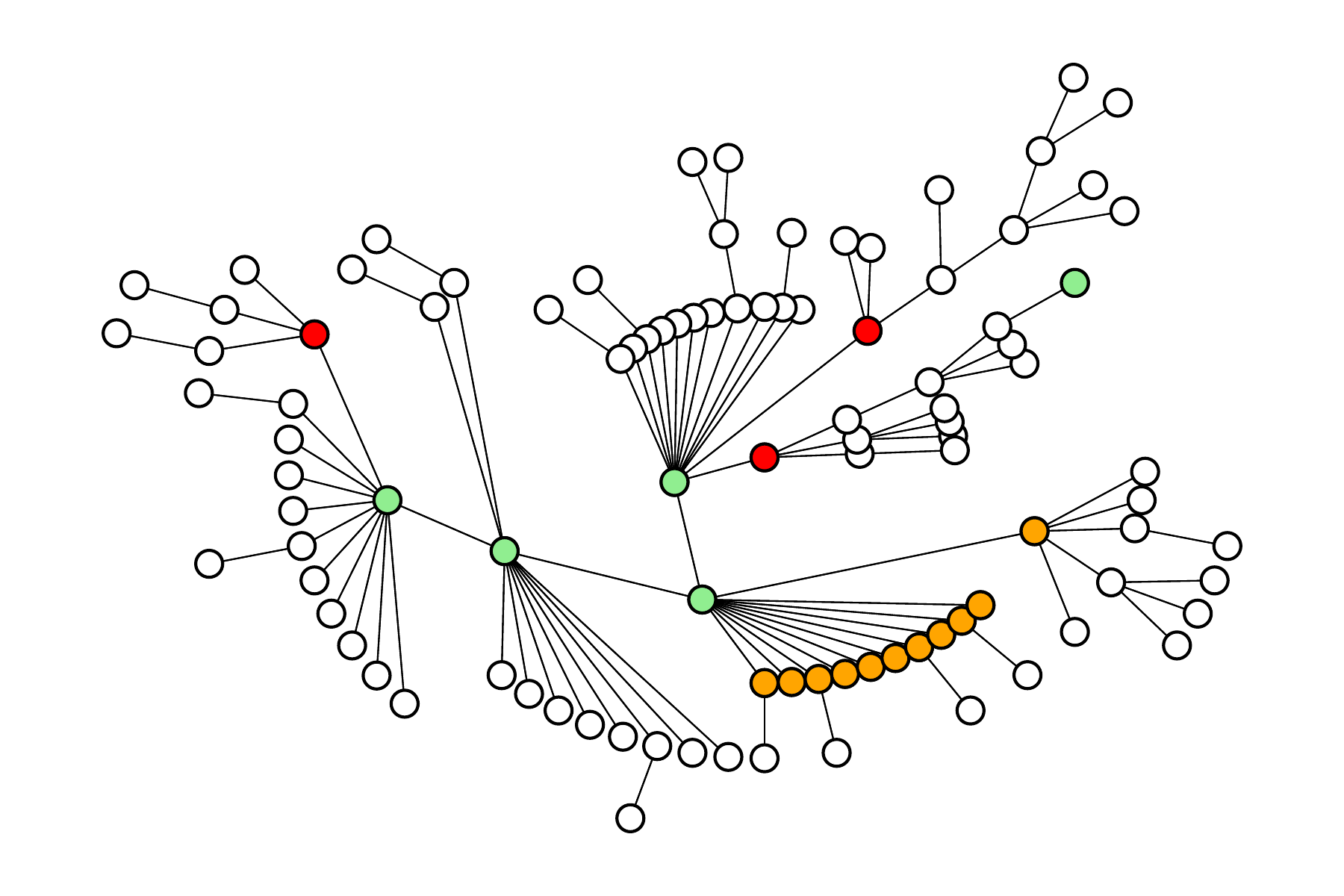}
        \caption{$k=4$}
    \end{subfigure}
    \caption{Illustration of resolving set and non-resolved vertices on a \Barabasi-Albert random tree for various values of the relaxation parameter $k$. {\color{red}Red}: vertices belonging to the $k$-relaxed resolving set found by Algorithm~\ref{algo}. {\color{green}Green}: vertices with a unique identification vector. {\color{orange}Orange}: vertices belonging to the largest equivalent class of non-resolved vertices. }
    \label{fig:visualization_BA}
\end{figure}

\subsubsection{Geometric random graph}
Figure~\ref{fig:evolution_md_RGG} shows the evolution of the relaxed metric dimension in random geometric graphs. Although these graphs differ significantly from \Barabasi-Albert random trees, we observe a similar behavior in the relaxed metric dimension. Specifically, the metric dimension decreases sharply when relaxed from $k=0$ to $k=1$, accompanied by a substantial increase in the proportion of non-distinguished vertices (rising from 0\% to 60\%). However, the size of the largest equivalence class remains quite small.

\begin{figure}[!ht]
 \centering
 \begin{subfigure}{0.32\textwidth}
  \includegraphics[width=\linewidth]{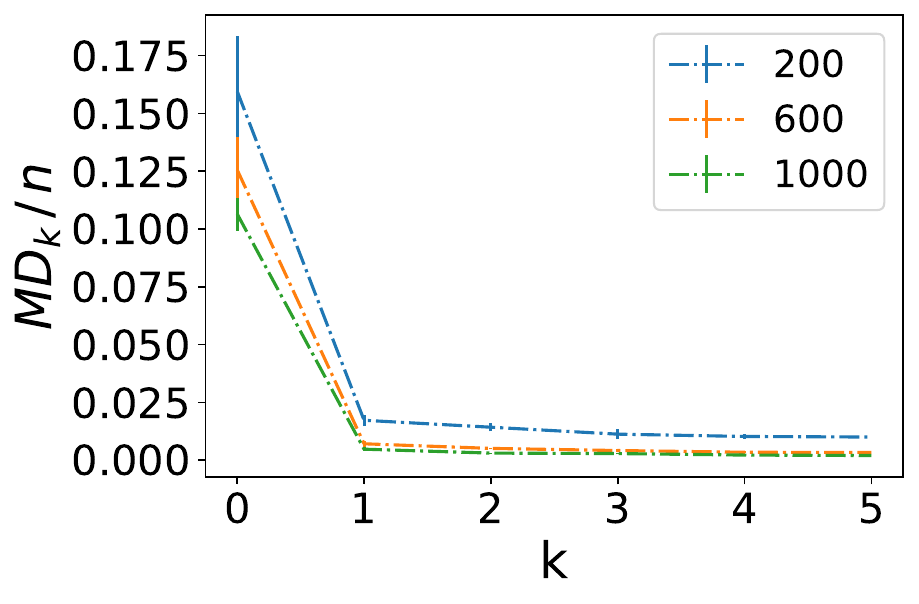}
  \caption{$\md_k / n$}
 \end{subfigure}
 \begin{subfigure}{0.32\textwidth}
  \includegraphics[width=\linewidth]{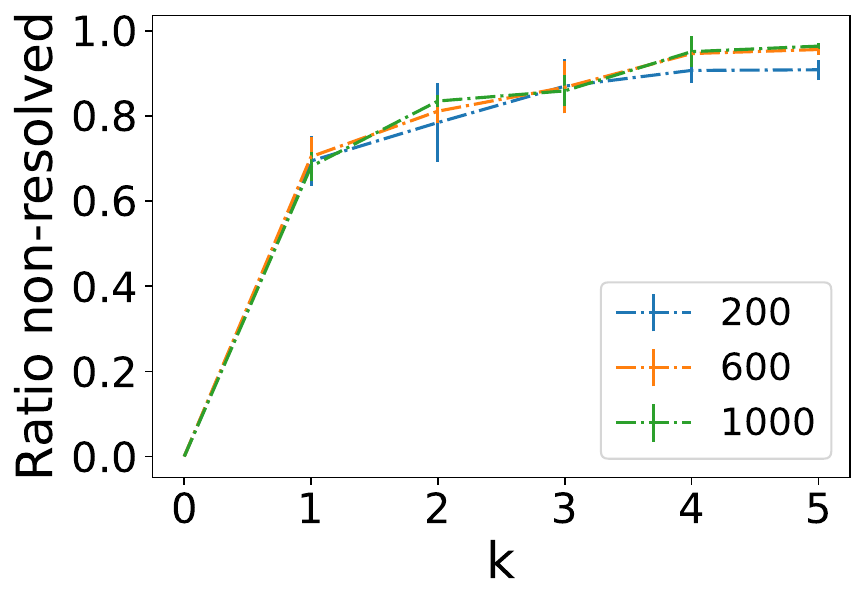}
  \caption{Ratio of non-resolved vertices}
 \end{subfigure}
 \begin{subfigure}{0.32\textwidth}
  \includegraphics[width=\linewidth]{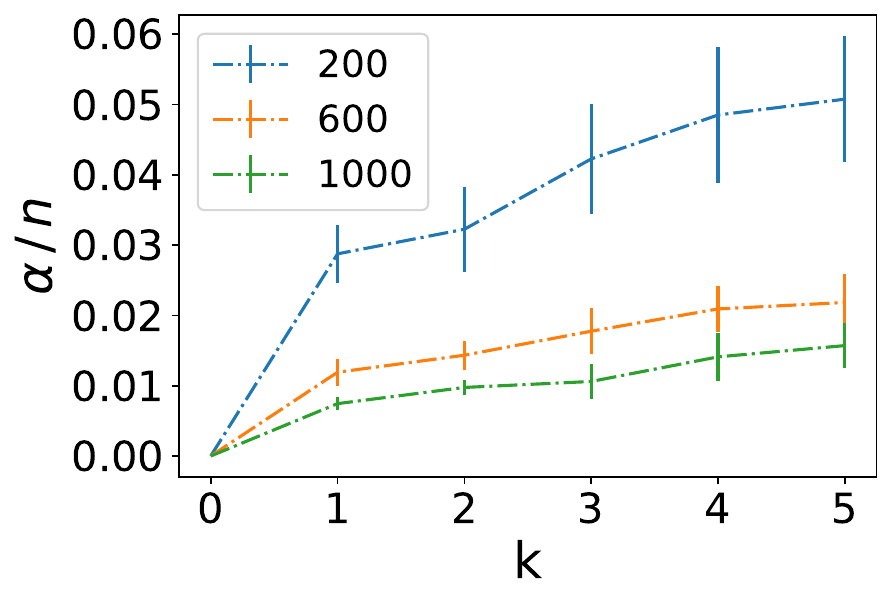}
  \caption{ $\alpha / n$ }
 \end{subfigure}
 \caption{Random geometric graphs in $2$-dimensional Euclidean space with radius $r = 1.5 \sqrt{\frac{\log n}{n \pi }}$. Results are averaged over $20$ realizations, and error bars show the standard deviation.}
 \label{fig:evolution_md_RGG}
\end{figure}

We draw in Figure~\ref{fig:visualization_RGG} relaxed resolving sets obtained by Algorithm~\ref{algo} for an instance of a random geometric graph (limited to 100 vertices for illustrative purposes). We observe that the equivalence classes of non-resolved vertices consist of vertices that are spatially close to one another in the metric space.

\begin{figure}[!ht]
    \centering
    \begin{subfigure}{0.32\textwidth}
        \includegraphics[width=\linewidth]{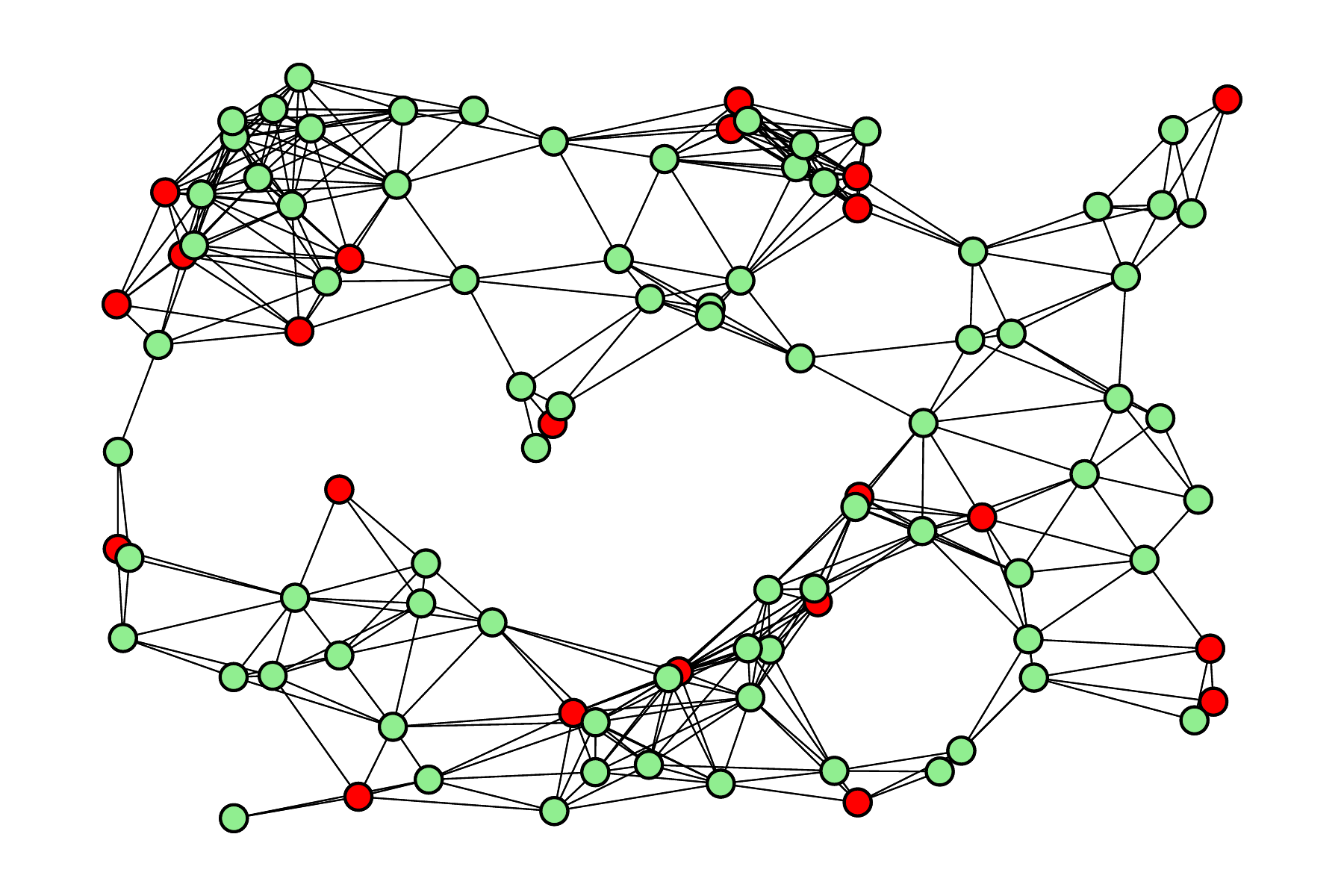}
        \caption{0-relaxed resolving set}
    \end{subfigure}
    \begin{subfigure}{0.32\textwidth}
        \includegraphics[width=\linewidth]{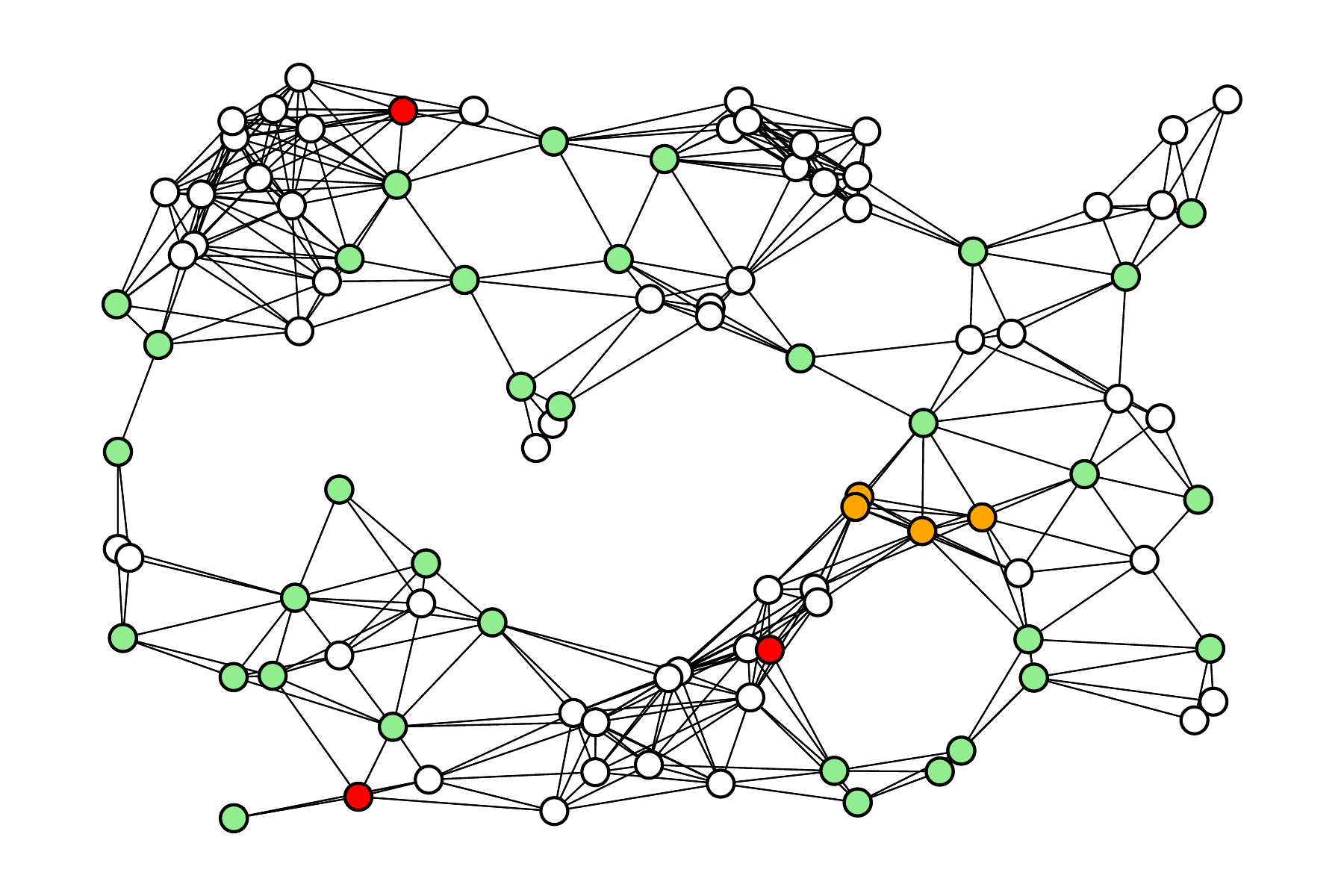}
        \caption{2-relaxed resolving set}
    \end{subfigure}
    \begin{subfigure}{0.32\textwidth}
        \includegraphics[width=\linewidth]{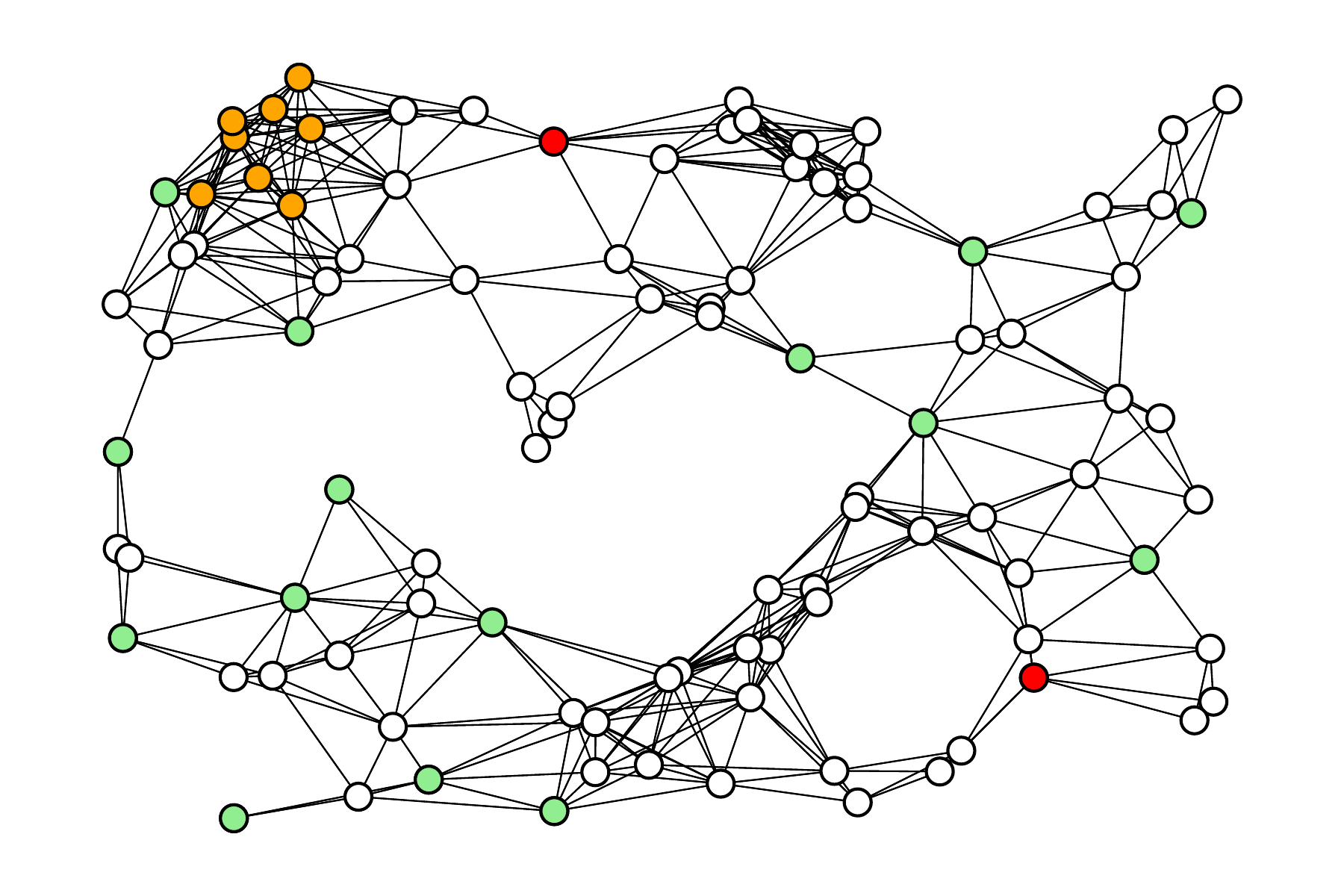}
        \caption{4-relaxed resolving set}
    \end{subfigure}
    \caption{Resolving sets obtained by Algorithm~\ref{algo} on a random geometric graph. Vertices are positioned according to their coordinate in the underlying metric space $[0,1]^2$. {\color{red}Red}: vertices belonging to the relaxed resolving set found by Algorithm~\ref{algo}. {\color{green}Green}: vertices with a unique identification vector. {\color{orange}Orange}: vertices belonging to the largest equivalent class of non-resolved vertices.}
    \label{fig:visualization_RGG}
\end{figure}

\subsection{Real Datasets}

We analyze three real-world networks to investigate their structural properties. The first two networks, \textit{Copenhagen-calls} and \textit{Copenhagen-friends}, represent social relationships among university students participating in the Copenhagen Networks Study \cite{real_graphs_copenhagen}. While \textit{Copenhagen-calls} captures interactions through phone calls, \textit{Copenhagen-friends} reflects Facebook friendships. The second network, \textit{Yeast} protein interactions,\footnote{Data available at \url{konect.cc/networks/moreno_propro/}} models metabolic interactions in yeast, where each vertex corresponds to a protein, and edges represent interactions between proteins. Lastly, the \textit{Co-authorships} network\footnote{Data available at \url{http://konect.cc/networks/dimacs10-netscience/}} captures collaborations among researchers in network science, with vertices representing authors and edges indicating co-authorships. For graphs that are originally disconnected, we restrict our analysis to their largest connected components. Table~\ref{table:real_graphs_statistics} summarizes the key characteristics of these networks.

\begin{table}[h!]
\centering
\begin{tabular}{ c c c c c c c  }
\hline \toprule
Network &   $|V|$  &  $|E|$   & $\bar{d}$ & $D_{\max}$ & $\bar{D}$ & $|S_1|$ \\
 \midrule
Copenhagen-calls & 347  & 477 & 2.75 & 22 & 7.40 & 141 \\
Copenhagen-friends & 800  & 6418 & 16.05 & 7 & 2.98 & 20 \\
Yeast & 1458 & 1948 & 2.67 & 19 & 6.81 & 864 \\
Co-authorships & 379 & 914   & 4.82 & 17 & 6.04 & 27 \\
\bottomrule
\end{tabular}
 \caption{Number of vertices $|V|$, number of edges $|E|$, average degree $\bd$, diameter $D_{\max}$, average shortest-path length $\bar{D}$, and size of the 1-\textit{shell} $|S_1|$ of the real graphs considered.}
\label{table:real_graphs_statistics}
\end{table}

\begin{figure}[!ht]
 \centering
 \begin{subfigure}{0.32\textwidth}
  \includegraphics[width=\linewidth]{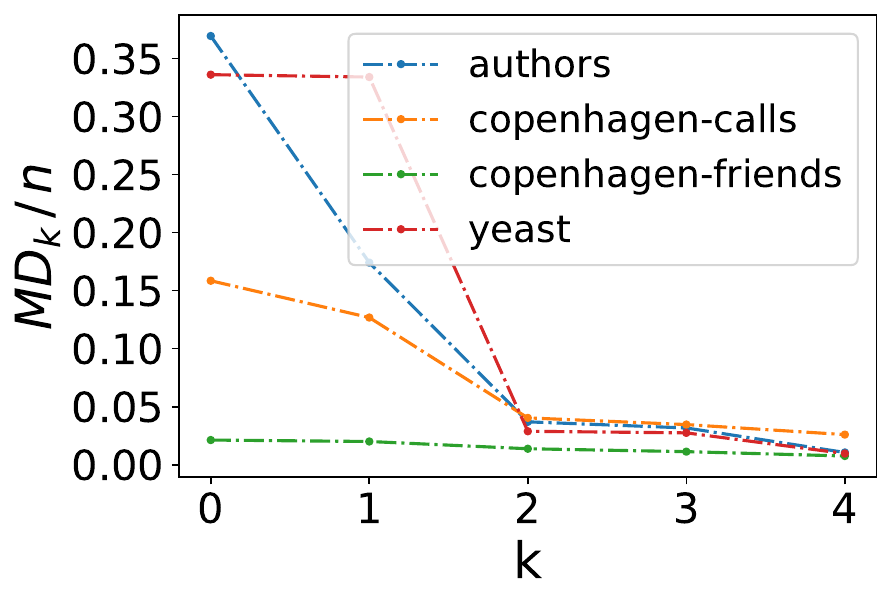}
  \caption{$\frac{ \md_k } { n }$}
 \end{subfigure}
 \begin{subfigure}{0.32\textwidth}
  \includegraphics[width=\linewidth]{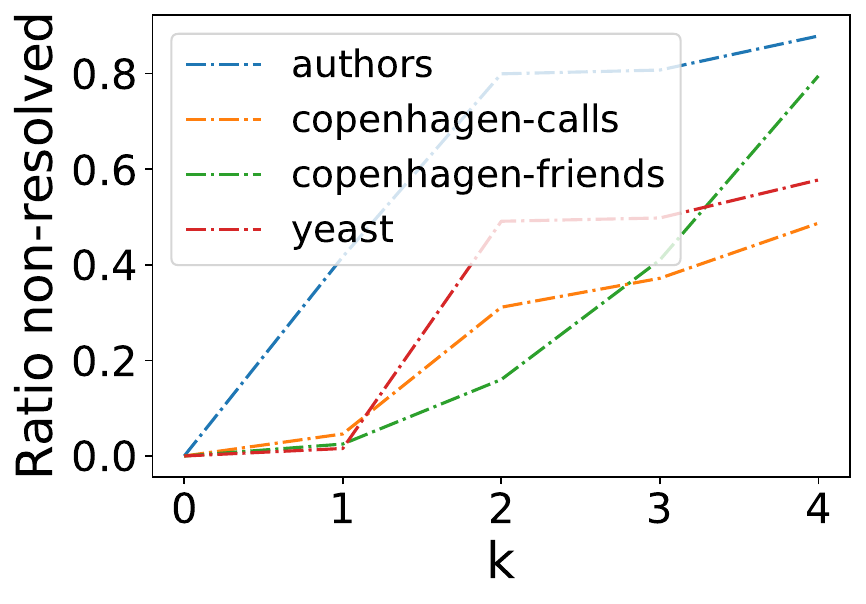}
  \caption{Ratio non-resolved vertices}
 \end{subfigure}
 \begin{subfigure}{0.32\textwidth}
  \includegraphics[width=\linewidth]{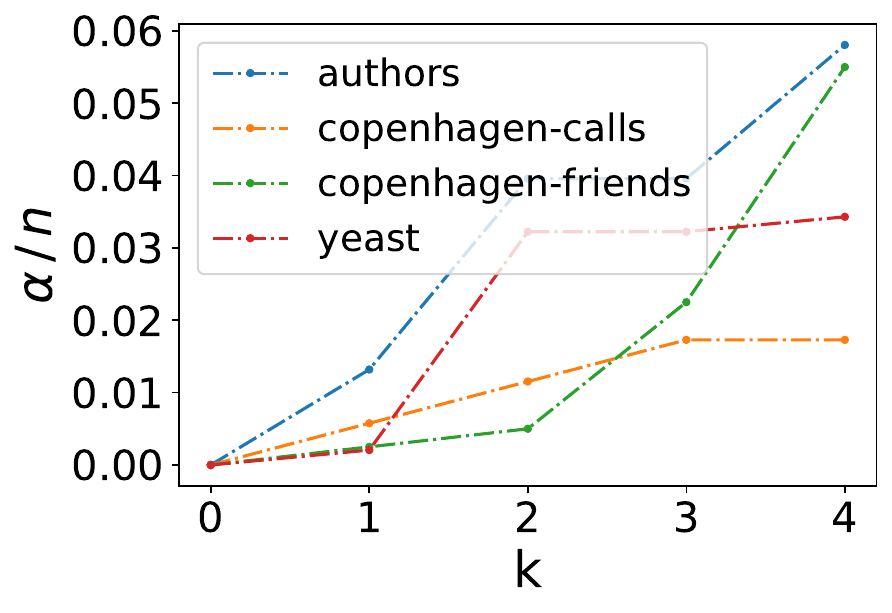}
  \caption{ $\alpha / n$}
 \end{subfigure}
 \caption{Relaxed metric dimension of real graphs.}
 \label{fig:evolution_md_realGraphs}
\end{figure}

Figure~\ref{fig:evolution_md_realGraphs} illustrates how the relaxed metric dimension evolves across these real-world networks. For most of them, introducing relaxation results in a substantial reduction in the number of required sensors, with the exception of the \textit{Copenhagen-friends} network, where the non-relaxed metric dimension is already low. Additionally, while an even modest relaxation rapidly increases the number of non-resolved vertices, the size~$\alpha$ of the largest equivalence class always remains small. These findings align closely with the patterns observed earlier for random graphs, further validating the generality of these behaviors.

\begin{figure}[!ht]
    \centering
    \begin{subfigure}{0.32\textwidth}
        \includegraphics[width=\linewidth]{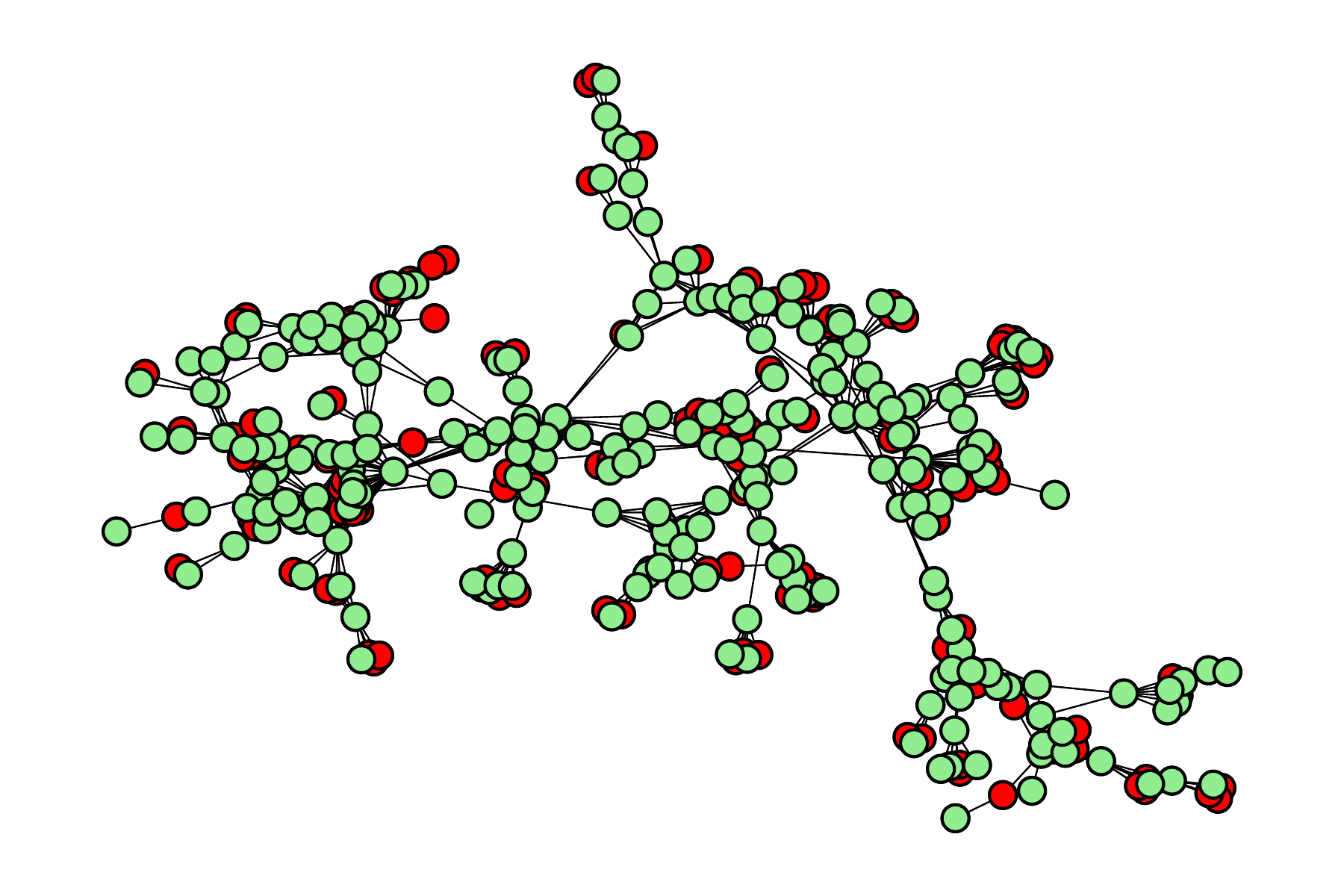}
        \caption{0-relaxed resolving set}
    \end{subfigure}
    \begin{subfigure}{0.32\textwidth}
        \includegraphics[width=\linewidth]{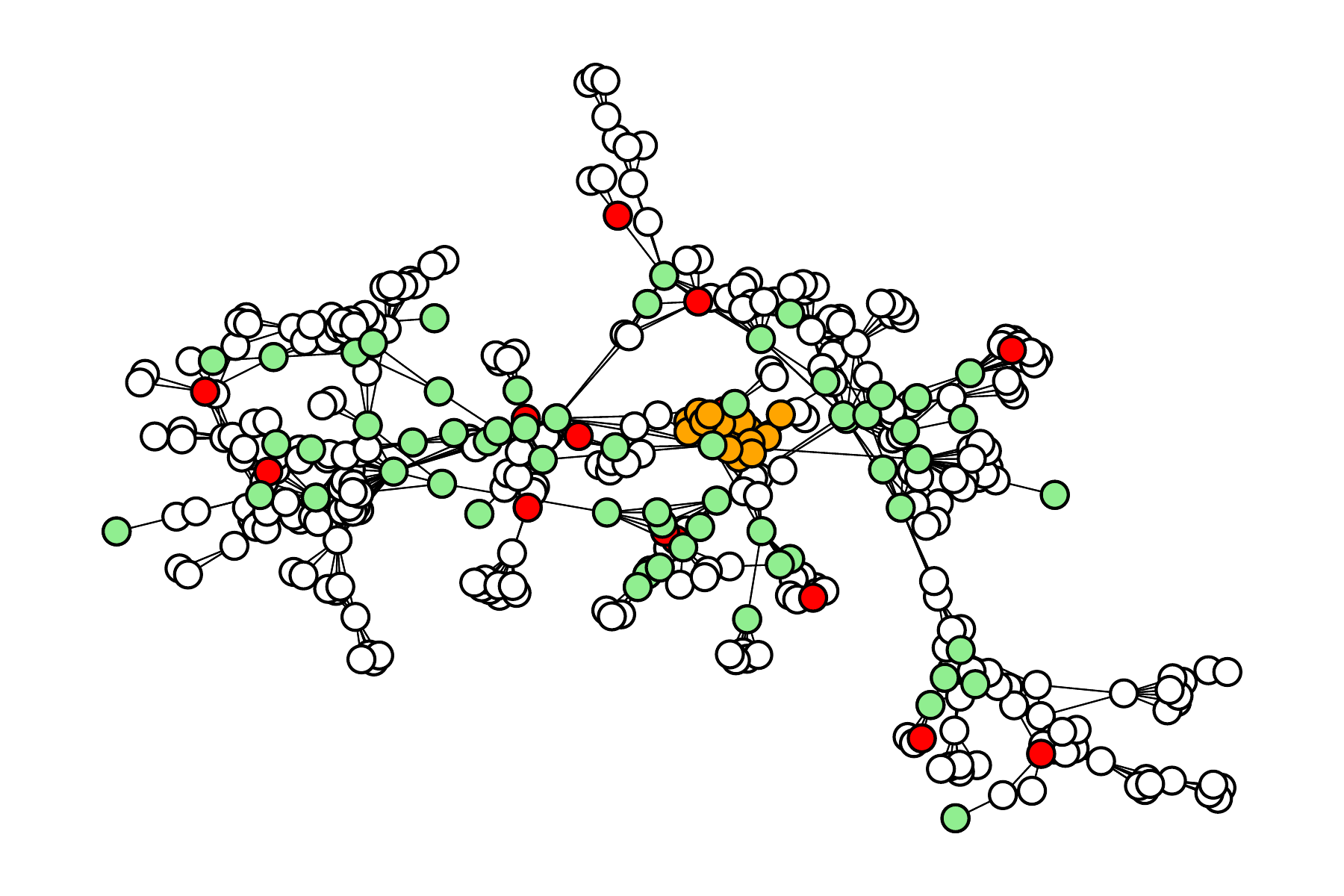}
        \caption{2-relaxed resolving set}
    \end{subfigure}
    \begin{subfigure}{0.32\textwidth}
        \includegraphics[width=\linewidth]{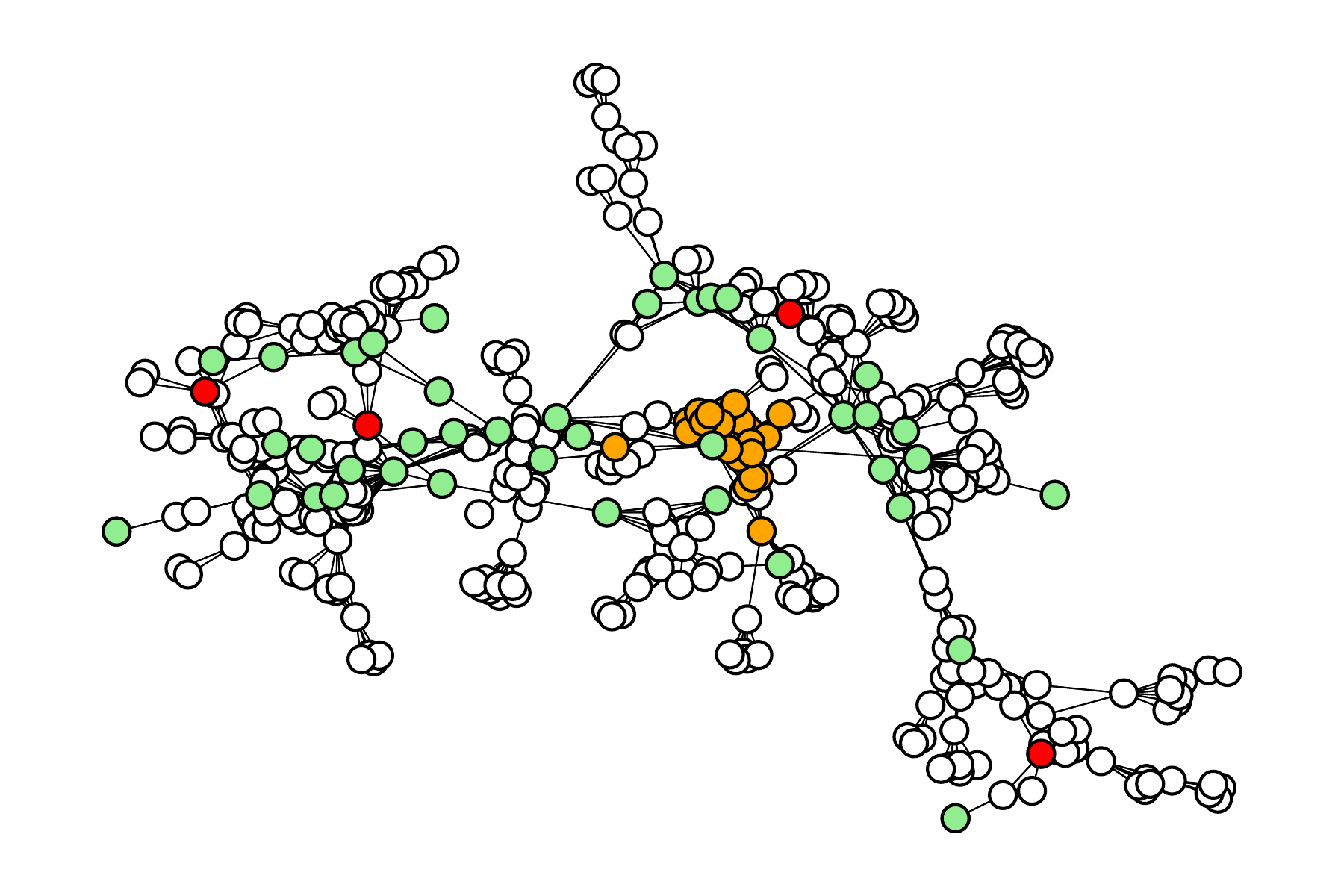}
        \caption{4-relaxed resolving set}
    \end{subfigure}
    \caption{Resolving sets obtained by Algorithm~\ref{algo} on \textit{co-authorship} graph. {\color{red}Red}: vertices belonging to the relaxed resolving set found by Algorithm~\ref{algo}. {\color{green}Green}: vertices with a unique identification vector. {\color{orange}Orange}: vertices belonging to the largest equivalent class of non-resolved vertices.}
    \label{fig:visualization_coauthorship}
\end{figure}

\section{Two-step Localization}
\label{section:twoStepsLocaalization}

 We motivated in the introduction the relaxation of the metric dimension as a trade-off between the savings in the number of sensors at the price of a reduced precision in the localization of source. The previous section shows that in practice, the number of potential suspects (i.e., the size~$\alpha$ of the largest equivalence class defined in Section~\ref{sec:metrics_used}) remains small for a wide range of the relaxation parameter $k$. In this section, we further explore this observation by proposing a two-step, passive-active strategy to locate a target and its benefits over the single-step, passive strategy. This two-step strategy is only applicable if the source localization does not change and the source signal is persistent over time, and is needed only if the exact localization of the source is required.

  Suppose that our goal is to find a target $u \in V$ on a graph $G = (V,E)$ in two steps. The first step is passive and simply consists in choosing a set $\hS_{k}^1$ of sensors, such that $\hS_{k}^1$ is a $k$-relaxed resolving set of the graph $G$ (in our experiments, we obtain $\hS_{k}^1$ by Algorithm~\ref{algo}). If the target~$u$ has a unique identification vector with respect to this set $S$, then we succeed at finding its location. But, as we saw in the previous section, if $k > 0$, a quite large fraction of the vertices may remain non uniquely identified by these fixed sensors. Therefore, the target $u$ will likely not be identified in this first step, and will need to be identified among the potential remaining candidate vertices. Those are the vertices having the same identification vector with respect to $\hS_{k}^1$ as $u$, that is, 
  \begin{align*}
   [u]_{ \hS_{k}^1 } \weq \left\{ v \in V \colon \Phi\left( v,\hS_{k}^1 \right) = \Phi \left( u,\hS_{k}^1 \right) \right\}.
  \end{align*} 
  The second step is an active step, which adds new sensors $\hS_{k,u}^{2}$ to resolve all the vertices belonging to this set of potential candidate vertices. As the goal is to minimize the number of sensors, we compute~$\hS_{k,u}^{2}$ in a greedy manner, similar to the computation of the relaxed metric dimension by Algorithm~\ref{algo}. We report the smallest number of sensors needed in the worst-case scenario, namely 
\begin{align*}
  q^*_k(G) \weq \max_{ u \in V } \left| \hS_{k}^{1} \cup \hS_{k,u}^{2} \right|. 
\end{align*}

Notice that $q^*_k(G) =  \left| \hS_{k}^{1} \right| + \max_{ u \in V } \left| \hS_{k,u}^{2} \right|$. Hence, when $k=0$ we do not need any additional sensor, and thus $q^*_k(G) = \md_0(G)$. Moreover, when $k \ge D_{\max}(G)$ where $D_{\max}(G)$ is the diameter of $G$, we have $\hS_{k}^{1} = \emptyset$. Because $[u]_{\emptyset} = V$ for any vertex $u$, the set $\hS_{k,u}^{2}$ must resolve all vertices, and hence $q^*_k(G) = \md_0(G)$ for $k \ge D_{\max}(G)$. 
Between these two extremes, $q^*$ is expected to decrease with $k$ until reaching a minimum value, after which it increases again. This behavior is empirically observed across the random graph models in Figure~\ref{fig:twoStepGame_syntheticGraphs}, as well as for real-world graphs in Figure~\ref{fig:twoStepGame_realGraphs}. The savings in terms of sensor resources can be drastic, leading to very small values of $q^*_k(G)$, if $k$ is selected in the good range of intermediate values avoiding between the first, passive phase and the second, active phase.

\begin{figure}[!ht]
 \centering
 \begin{subfigure}{0.32\textwidth}
  \includegraphics[width=\linewidth]{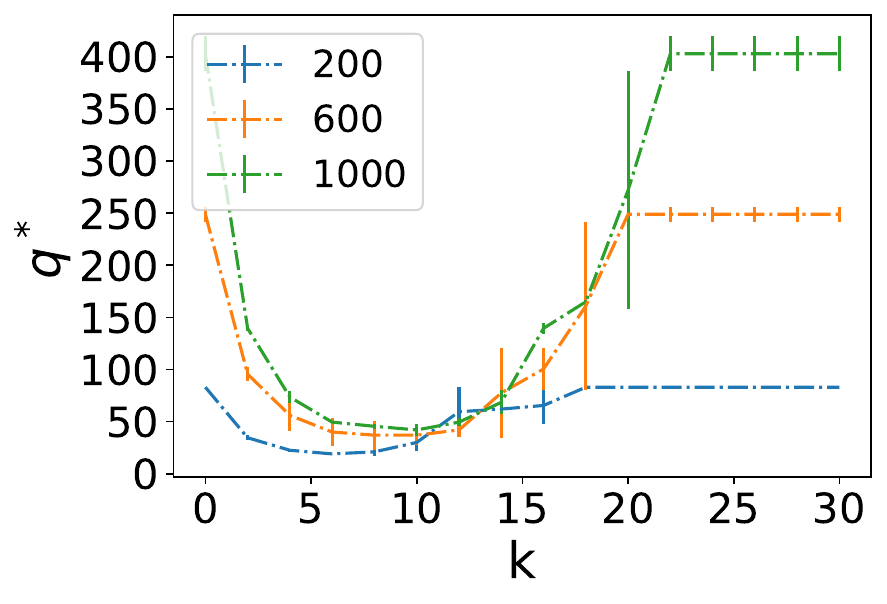}
  \caption{\Barabasi-Albert}
 \end{subfigure}
 \hfil
 \begin{subfigure}{0.32\textwidth}
  \includegraphics[width=\linewidth]{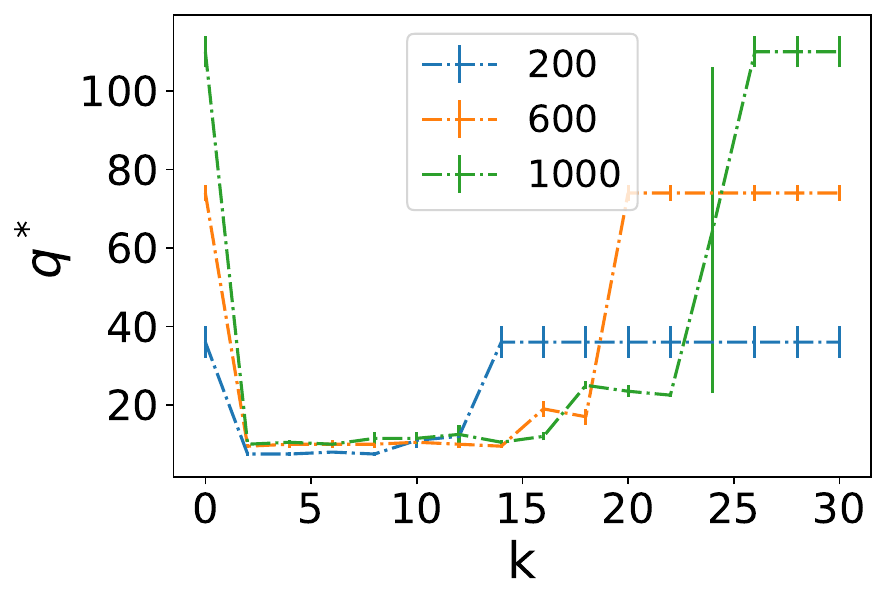}
  \caption{Random geometric graph}
 \end{subfigure}
 \caption{Value of $q^*_k(G)$ as a function of $k$ on \Barabasi-Albert and random geometric graphs. Error bars show the standard deviation over 10 runs.}
 \label{fig:twoStepGame_syntheticGraphs}
\end{figure}

\begin{figure}[!ht]
 \centering
 \begin{subfigure}{0.32\textwidth}
  \includegraphics[width=\linewidth]{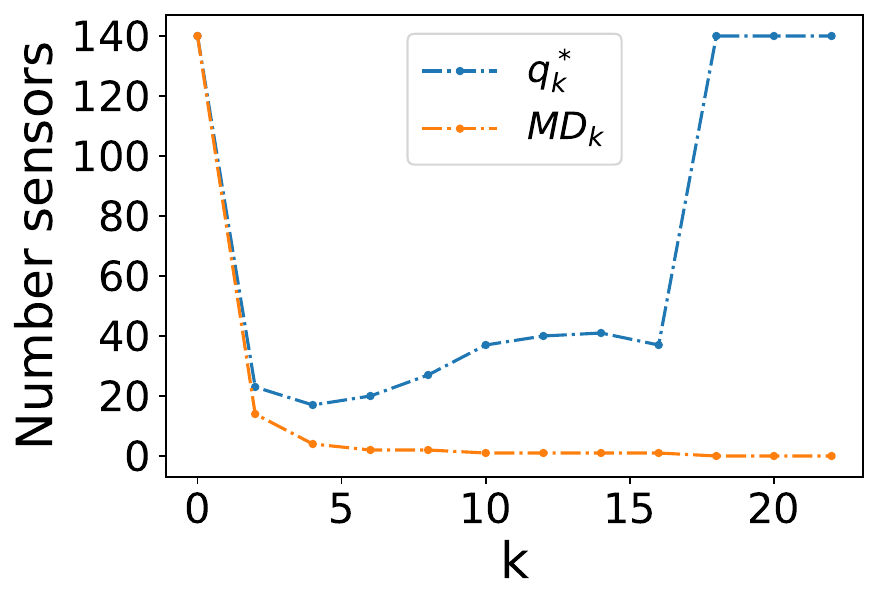}
  \caption{\textit{co-authors}}
 \end{subfigure}
 \hfill
 \begin{subfigure}{0.32\textwidth}
  \includegraphics[width=\linewidth]{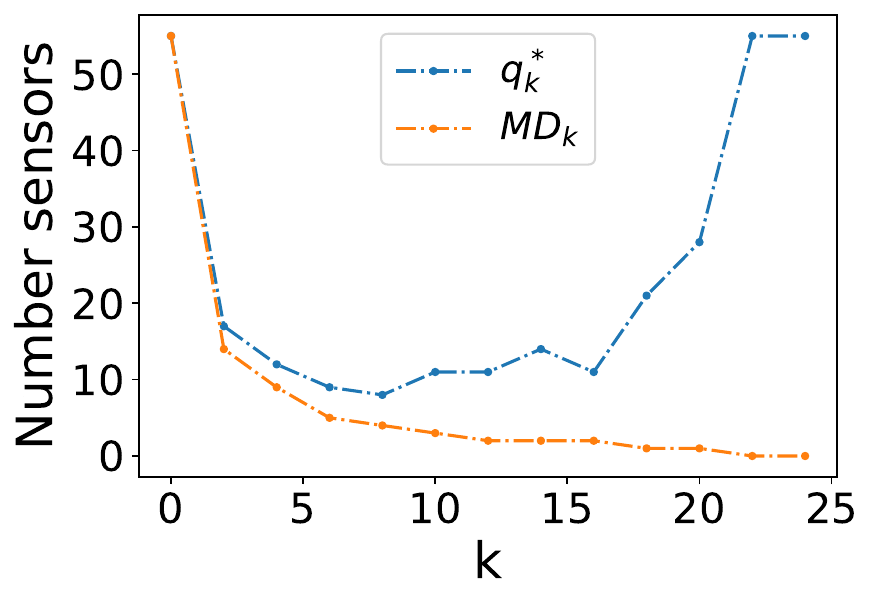}
  \caption{\textit{Copenhagen-calls}}
 \end{subfigure}
 \hfill
 \begin{subfigure}{0.32\textwidth}
  \includegraphics[width=\linewidth]{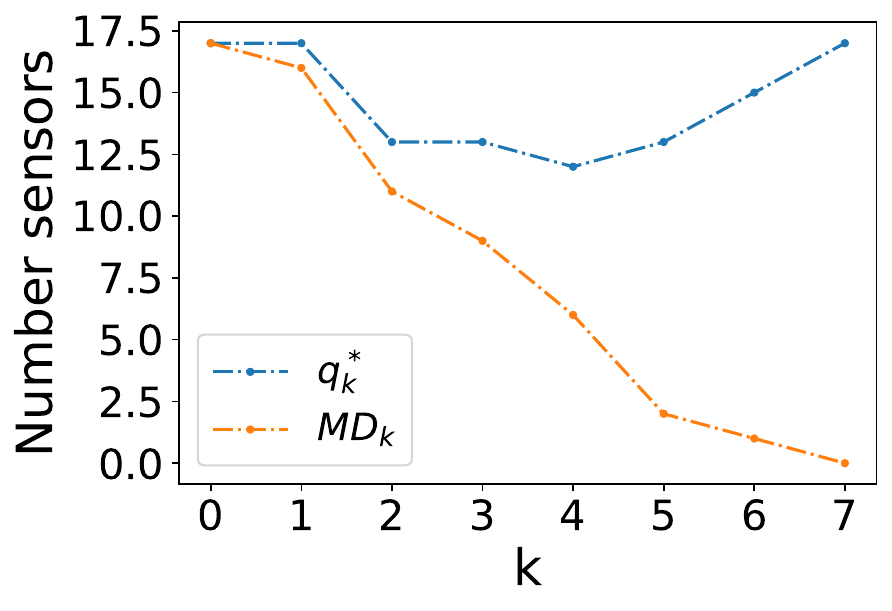}
  \caption{\textit{Copenhagen-friends}}
 \end{subfigure}
    \caption{Value of $q^*_k(G)$ (blue) and of the relaxed metric dimension $\md_k$ (orange) as a function of $k$ on three real graphs. }
    \label{fig:twoStepGame_realGraphs}
\end{figure}

\section{Overview of the Proofs}
\label{section:proof_overview}

\subsection{Proofs for Arbitrary Trees}

\subsubsection{Proof of Theorem~\ref{thm:odd_dist}}
\label{sec:prf_odd_dist} This section provides the proof of Theorem~\ref{thm:odd_dist}.

\begin{proof}[Proof of Theorem~\ref{thm:odd_dist}]
(i) Clearly, $\md_{2r}(T) \geq \md_{2r+1}(T)$. Indeed, let $R$ be any $i$-relaxed resolving set of a graph. Then $R$ is also an $j$-relaxed resolving set for $j \geq i$, because $R$ distinguishes all vertices at distance larger than $i$, which include all vertices at at distance larger than $j \geq i$.\\

(ii) Conversely, we show that $\md_{2r}(T) \leq \md_{2r+1}(T)$ by proving that on a tree, any $(2r+1)$-relaxed resolving set is also a $2r$-relaxed resolving set. To $(2r+1)$-resolve a vertex $u$, we need a resolving set that distinguishes $u$ at least from all vertices $w$ with $d(u,w) > 2r+1$. Now, consider a $2r$-resolving set, which additionally needs to distinguish $u$ from all vertices exactly at a distance $2r+1$. We will now show that any vertex $v$ in the $(2r+1)$-resolving set also resolves $u$ from any vertex $w$ located at an odd distance $d(u,w)$.\\
Let $u$ and $w$ be two vertices at an odd distance from each other, and let $v$ be an arbitrary vertex in the tree. Let $c$ be the vertex on the shortest path between $u$ and $v$, which is the closest to $v$ among all vertices on the path $u-w$ (this vertex $c$ is unique since the graph is a tree. We note that $c = u$ or $c = w$ is possible). We can then split the shortest path $v-u$ from $v$ to $u$ into two segments $v-c$ and $c-u$. Similarly, we can decompose the shortest path $v-w$ between $v$ and $w$ into the two segments $v-c$ and $c-w$. Now, since $d(u,w)$ is odd, $c$ cannot lie precisely at the midpoint of the path $u-w$ and therefore $d(c,u) \neq d(c,w)$, which in turn implies that  $d(v,c) + d(c,u) = d(v,u) \neq  d(v,w) = d(v.c) + d(c,w)$. This means that no vertex can be at the same distance from two vertices that are at an odd distance from each other. Consequently, no pair of vertices at an odd distance from each other can share the same identification vector. This implies that $\md_{2r}(T) \leq \md_{2r + 1}(T)$. \\

Combining (i) and (ii) proves the result.

\end{proof}

\subsubsection{Proof of Theorem~\ref{thm:k_metric_dim_tree}}
\label{sec:prf_k_metric_dim_tree}

The proof of Theorem~\ref{thm:k_metric_dim_tree} uses two key lemmas. The first lemma, proven in Appendix \ref{sec:prf_lm_lower_bound_dim}, provides a lower bound on the metric dimension of any graph~$G$.

\begin{lemma}
\label{lm:lower_bound_dim2}
Let $G$ be a connected graph and $r$ a nonnegative integer such that $\stem_r(G)$ is not a line graph. Then
\begin{align*}
 \md_{2r}(G) \wge \sigma_r(G) - \ex_r(G).
\end{align*}
\end{lemma}

The proof of Theorem~\ref{thm:k_metric_dim_tree} follows if we can show that this lower bound is tight when $G$ is a tree. 
To do so, we will use the following lemma. 

\begin{lemma}
  \label{upper_bound_dim2}  
  Let $T'$ be a tree containing a sub-tree $T$ that is not a line graph, and let $m = \md(T)$. If all vertices in $T' \setminus T$ are at most at distance $r$ from a vertex in $T$, that is, if
  \begin{equation}
      \forall v' \in T' \setminus T, \  \exists v \in T \colon d(v',v) \leq r,
  \end{equation}
  then there exists a $2r$-relaxed resolving set of size $m$ for $T'$.
\end{lemma}

Lemma~\ref{upper_bound_dim2} bridges the gap between what is known about the classic metric dimension on trees and the $2r$-relaxed metric dimension. The idea of the proof is to take a tree $T$ that can be resolved with a resolving set $R$ of size $m$. Then, we consider a second tree $T'$ such that $T$ is a sub-tree of $T'$ and every vertex in $T'$ is either in $T$ or is at most $r$ steps away from a vertex in $T$. If we look back at the notion of stemming used earlier, the vertices in $\stem_r(T')$ form a subset of the vertices in $T$. Lemma~\ref{upper_bound_dim2} shows that $T'$ can be $2r$-relaxed resolved with $R$. This is done by checking the resolvability of vertices as a function of their position in $T'$ and of the closest vertex in $T$. The proof is similar to the proof of the regular metric dimension of trees in \cite{MD_for_tree}, albeit the analysis of the different cases is more tedious and is divided in several lemmas in Appendix~\ref{sec:upper_bound_dim2}. We can now proceed with the proof of Theorem~\ref{thm:k_metric_dim_tree}.

\begin{proof} [Proof of Theorem~\ref{thm:k_metric_dim_tree}]
By Lemma \ref{lm:lower_bound_dim2}, $\md_{2r}(T) \geq \sigma_r(T) - \ex_r(T)$. Let us now establish the reverse inequality. We first construct a resolving set $R$ for $\stem_r(T)$ as follows. If $\stem_r(T)$ is a line graph, one of its end-vertices forms a resolving set by itself, and $\md_{2r}(T) = 1$. Otherwise, for each exterior major vertex $u \in \stem_r(T)$, We include all of the leaves of $u$ but one. As a result, $R$ is of size $\sigma_r(T) - \ex_r(T)$. Moreover, $R$ is a resolving set of $\stem_r(T)$ (by \cite[Theorem~5]{MD_for_tree}, which was recalled before in Equation (\ref{eq:formula_md_tree})). Observe that any vertex $v \in T \setminus \stem_r(T)$ can be at most at distance $r$ from the $r$-stem of $T$ since only $r$~stemmings were applied. Therefore we apply Lemma~\ref{upper_bound_dim2} and establish that there exists a $2r$-relaxed resolving set of size $\sigma_r(T) - \ex_r(T)$ for the tree $T$. 
\end{proof}

\subsection{Proofs for Random Trees}
\label{sec:prf_randomTrees}

\subsubsection{Subtree properties and down-stemming}

 Consider a rooted tree $T$ and a vertex $v \in T$. We denote by $T_v$ the subtree of $T$ rooted at $v$ and oriented away from the root of $T$ (\textit{i.e.,} the subtree consisting of $v$ and all its descendants). A \textit{subtree property} $\cP$ is a property of a vertex $v \in T$ that depends only on the subtree $T_v$. The number of vertices in $T$ satisfying property $\cP$ is denoted by 
 \begin{align*}
   N_{\cP}(T) \weq | \{ v \colon T_v \in \cP \} |.
 \end{align*}
 The following known result relates the frequency of a subtree property $\cP$ in a \textit{conditioned} Galton-Watson tree to the corresponding probability of $\cP$ in an \textit{unconditioned} Galton-Watson tree with the same offspring distribution.

 \begin{theorem}[Theorem 1.3 in \cite{janson_fringe_tree_2016}] \label{thm:st_n_to_uncond_prob}
  Let $(\mathcal{T}^{GW}_n)_{n \in \N}$ be a sequence of critical Galton-Watson trees conditioned on having $n$ vertices, with offspring distribution $\xi$, where $\E\left[\xi\right] = 1$ and $\E\left[\xi^2\right] < \infty$. Let $\mathcal{F}$ be an unconditioned Galton-Watson tree with the same offspring distribution $\xi$. Then, for every subtree property $\cP$, we have 
  \begin{equation*}
     \frac{N_{\mathcal{P}}(T^{GW}_n)}{n} \xrightarrow{p} \mathbb{P}(\mathcal{F} \in \mathcal{P}),
    \end{equation*}
    as $n \to \infty$. 
\end{theorem}

 Let $T^{GW}_n$ be a critical Galton-Watson tree with $n$ vertices. Because $T^{GW}_n$ is a tree, we can apply Theorem~\ref{thm:k_metric_dim_tree} to express its $2r$-relaxed metric dimension as a function of the number of its leaves and exterior major vertices in the $r$-stem of $T_n$. 
 The proof strategy follows a similar approach to that in~\cite{komjathy_Odor_GW_MD}, where the authors define two subtree properties to approximate the number of leaves and exterior major vertices in~$T^{GW}_n$, respectively. By computing the probability of these properties occurring in an unconditioned Galton-Watson tree, they establish the asymptotics of the (non-relaxed) metric dimension of Galton-Watson trees. 
 
 In our case, we adapt the subtree properties to approximate the number of leaves and exterior major vertices of $\stem_k( T^{GW}_n)$, respectively. However, additional care is required, as the iterated stemming operations may remove the root of the original tree~$T_n^{GW}$, and the subtree properties crucially depend on the presence of the root. To address this, we modify the stemming operation to ensure that the root of $T_n^{GW}$ is always preserved. We call this operation \new{down-stemming}. In Lemma~\ref{lm:st_probs_close_to_md} below, we show that this adjustment on the stemming operation affects the relaxed metric dimension by at most one vertex. 

\begin{definition}
 Let $T$ be a rooted tree. The \new{down-stem} of $T$, denoted as $\ds(T)$, is the unique subgraph of $T$ induced by all vertices of degree strictly larger than one and the root (regardless of its degree). Furthermore, let $\ds_r(T)$ denote the $r$-fold application of the down-stemming operation on $T$ and let $\ds_0(T) = T$. 
\end{definition}

\subsubsection{Subtree properties for the relaxed metric dimension}

We now define two subtree properties with which we approximate the number of leaves and exterior major vertices in the $r$-stemmed tree. Let the \textit{height} of a vertex be its distance to the root and let the height of a tree be the maximal distance to the root from any other vertex (The root has thus height zero).

\begin{definition} 
 \label{def:subtree_L_E_prop}
 Let $\mathcal{P}_r^L$ be the subtree property that the subtree is of height $r$. Furthermore, let $\mathcal{P}^E_r$ be the subtree property of a tree $T_v$ such that in $\text{Down-Stem}_r(T_v)$ the root $v$ has degree at least two and at least one of its subtrees is a line graph to a leaf (a subtree having a single vertex is considered to be a line).
\end{definition}
See Figure~\ref{fig:st_prop_def} for an example.

\begin{figure}[!ht]
    \centering
    \begin{subfigure}{0.30\textwidth}
    \centering
        \includegraphics[width=\linewidth]{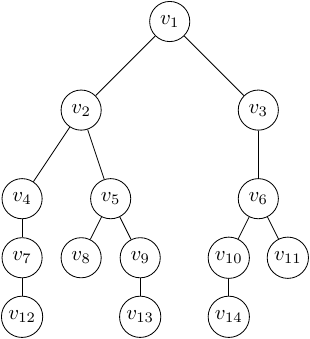}
        \caption{Tree $T$}
    \label{fig:st_prop_def2_a}
    \end{subfigure}
    \hfil
    \begin{subfigure}{0.30\textwidth}
    \centering
        \includegraphics[width=0.75\linewidth]{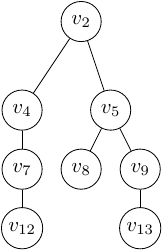}
        \caption{Subtree $T_{v_2} \in \mathcal{P}^E_2$}
    \label{fig:st_prop_def2_b}
    \end{subfigure}
    \hfil
    \begin{subfigure}{0.30\textwidth}
    \centering
        \includegraphics[width=0.5\linewidth]{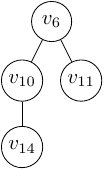}
        \caption{Subtree $T_{v_6} \in \mathcal{P}^L_2$}
    \label{fig:st_prop_def2_c}
    \end{subfigure}
    \caption{Illustration of the two properties $\cP^E_2$ and $\cP^L_2$. 
    Figure~\ref{fig:st_prop_def2_a} shows a tree $T$, rooted at $v_1$. 
    Figure~\ref{fig:st_prop_def2_b} shows $T_{v_2}$, the subtree of $T$ rooted at $v_2$. Because $\ds_2(T_{v_2})$ is composed of the vertices $v_2$, $v_4$ and $v_5$, the root $v_2$ of $\ds_2(T_{v_2})$ has degree 2. Hence $T_{v_2} \in \mathcal{P}^E_2$. 
    Figure~\ref{fig:st_prop_def2_c} shows $T_{v_6}$, the subtree of $T$ rooted at $v_6$. This subtree has height 2, and hence $T_{v_6} \in \mathcal{P}^L_2$. Note that the only subtree of $T$ verifying $\cP_2^E$ is $T_{v_2}$ while the subtrees of $T$ verifying $\cP_2^L$ are $T_{v_4}$, $T_{v_5}$ and $T_{v_6}$. 
    }
    \label{fig:st_prop_def}
    \Description[Tree with example of subtrees satisfying the properties $\mathcal{P}^E_2$ and $\mathcal{P}^L_2$]{The first figure shows an arbitrary tree $T$, the second figure shows a subtree $T_v$ of $T$ where the root of $\ds_2(T_v)$ has two children, one of which is a leaf. Hence, it satisfies $\mathcal{P}^E_2$. The third figure shows a subtree $T_u$ of $T$ where the root of $\ds_2(T_u)$ is a leaf. Hence, it satisfies $\mathcal{P}^L_2$}
\end{figure}

 We denote by $N^L_r(T)$ (resp., by $N^E_r(T)$) the number of vertices of a tree $T$ satisfying the properties $\cP_r^L$ (resp., $\cP^E_r$). 
However, the following lemma demonstrates that counting the occurrences of $\cP^L_r$ and $\cP_r^E$ provides a close approximation to the true relaxed metric dimension of $T$. For an illustration of scenarios where the subtree properties may overestimate or underestimate the count of leaves or exterior major vertices, we refer the reader to Figure~\ref{fig:over_under_counting} in Appendix~\ref{appendix:st_probs_close_to_md}.

\begin{restatable}{lemma}{StProbsVsMD} \label{lm:st_probs_close_to_md}
    For any tree $T$ and nonnegative integer $r$, 
    \begin{equation*}
        \md_{2r}(T) \weq N^L_r(T) - N^E_r(T) + \varepsilon,
    \end{equation*}
    where $\varepsilon \in \{ -1, 0, +1\}$.
\end{restatable}

Asymptotically, the additional term $\varepsilon$ appearing in Lemma~\ref{lm:st_probs_close_to_md} becomes negligible. Therefore, combined with Theorem~\ref{thm:st_n_to_uncond_prob}, Lemma~\ref{lm:st_probs_close_to_md} shows that
\begin{align*}
   \frac{\md_{2r} \left( \cT^{GW}_n \right) }{n} \xrightarrow{p} \mathbb{P}\left( \mathcal{F}  \in \mathcal{P}^L_r \right) -  \mathbb{P}\left( \cF \in \cP^E_r \right),
\end{align*}
where $\cF$ is an unconditioned Galton-Watson tree with offspring distribution $\xi$. 
To finish the proof of Theorem~\ref{thm:rlaxed_md_random_trees}, we need to compute the probabilities $\mathbb{P}\left( \cF \in \cP^L_r \right)$ and $\mathbb{P}\left( \cF \in \cP^E_r \right)$. Denote 
\begin{align}\label{def:expression_ell_k_e_k}
  \ell_r \weq \mathbb{P}\left( \mathcal{F}(\xi)  \in \mathcal{P}^L_r \right)
  \quad \text{ and } \quad 
  e_r \weq \mathbb{P}\left( \mathcal{F}(\xi)  \in \mathcal{P}^E_r \right).
\end{align}
Note that $\ell_r$ and $e_r$ depend on the offspring distribution $\xi$; however, for clarity, this dependence is omitted in the notations that follow. We also introduce the two following auxiliary variables:

\begin{equation}\label{def:expression_d_k_s_k}
  \begin{aligned}
    d_r & \weq \mathbb{P}\left(\mathcal{F}(\xi) \text{ is of height strictly less than } r\right), \\
    s_r & \weq \mathbb{P}\left( \ds_r(\mathcal{F}(\xi)) \text{ is a line graph}\right).
  \end{aligned}
\end{equation}
The following lemma relates the key quantities $\ell_r$ and $e_r$ to the auxiliary quantities $d_r$ and $s_r$. 
\begin{lemma} \label{lm:ell_k_e_k}
    Let $\cF$ be an unconditioned Galton-Watson tree with offspring distribution $\xi$ and denote $\mathbb{P}(\xi = i)=p_i$. Recall the subtree properties $\mathcal{P}^L_r$ and $ \mathcal{P}^E_r$ defined in Definition~\ref{def:subtree_L_E_prop}. Then, 
    \begin{align*}
        \ell_r \weq 
        \begin{cases}
            p_0 &\text{for } r=0, \\
            \sum_{j=0}^{\infty} \left[ p_j (d_{r-1} + \ell_{r-1})^j \right] -d_r &\text{for } r \geq 1,
        \end{cases}  
        \quad \text{ and } \quad 
        e_r \weq 1 - \sum_{j=0}^\infty p_j (1-s_r)^j - s_r +\ell_r. 
    \end{align*}
\end{lemma}
Finally, the following two lemmas express the auxiliary quantities $d_r$ and $s_r$. 

\begin{lemma} \label{lm:d_k_m_k}
 Let $\cF$ be an unconditioned Galton-Watson tree with offspring distribution $\xi$. Then
 \begin{equation*}
   d_r \weq
   \begin{cases}
     0 &\text{for } r=0, \\
     \sum_{j=0}^{\infty} p_j d^j_{r-1} &\text{for } r \geq 1,
   \end{cases}
   \quad \text{ and } \quad 
        s_r \weq \frac{\ell_r}{1- \sum_{j=1}^{\infty} p_jjd^{j-1}_r}.
    \end{equation*}
\end{lemma}
We refer to Appendix~\ref{appendix:proof_auxilliary_quantities_random_trees} for the proofs of Lemmas~\ref{lm:ell_k_e_k} and Lemma~\ref{lm:d_k_m_k}. All these results combined lead to the proof of Theorem~\ref{thm:rlaxed_md_random_trees}.

\section{Conclusion and Discussion}
\label{section:conclusion}

 We conclude this paper by discussing several related works and some future work. 

\subsection{Variants of the Metric Dimension}

  Several variations of the metric dimension have been studied in the literature, including the \textit{double metric dimension}, the \textit{$k$-truncated metric dimension}, and the \textit{$k$-metric dimension}.

  A set $S$ of sensors is a \textit{double resolving set} if and only if any two vertices can be distinguished by their relative distances to two sensors. Specifically, for any two vertices $v_1, v_2 \in V$, there exist two sensors $s_1, s_2 \in S$ such that $d(s_1, v_1) - d(s_1, v_2) \ne d(s_2, v_1) - d(s_2, v_2)$. These sets rely on relative rather than absolute distances, making them particularly useful in scenarios where only differences in distances are available, such as identifying the source of an epidemic with an unknown start time~\cite{spinelli2018many}. The \textit{$k$-truncated metric dimension} assumes that sensors can only differentiate vertices within a certain distance $k$~\cite{estrada2019k,frongillo2022truncated,bartha2023sharp}. This is particularly relevant in environments where the accuracy of distance measurements diminishes with increasing range. Lastly, the \textit{$k$-metric dimension} requires vertices to be distinguishable by at least $k$ sensors~\cite{hernando2008fault,estrada2019k}. This ensures robust identification even in scenarios where up to $k-1$ sensors may fail.

  All the aforementioned variants of the metric dimension \textit{increase} the number of required sensors by imposing stricter conditions. In contrast, our relaxation \textit{reduces} the number of sensors needed. Another strategy for achieving this reduction involves using a sequential approach. Consider the following game introduced in~\cite{seager2012locating}: an invisible and immobile target is hidden at a vertex $u$, and at each step, a single sensor can be placed, revealing its distance to $u$. The objective is to locate $u$ using the fewest possible steps. The \textit{sequential metric dimension} is defined as the minimum number of steps required to locate the target, regardless of its position, by optimally placing sensors one at a time. This concept has been studied in trees~\cite{bensmail2020sequential} and random graphs~\cite{odor2021sequential}. 

  The sequential approach allows for dynamic sensor placement, which can significantly reduce the total number of sensors required. However, it assumes that sensors can be placed adaptively, and the number of steps needed may be as large as the sequential metric dimension itself. This trade-off highlights a fundamental tension between the number of sensors deployed and the strategy used for identification. On the one hand, a fully passive strategy, where all sensors are pre-placed to completely resolve the graph, ensures immediate identification in a single step but requires in general a very large number of sensors. On the other hand, a fully active strategy, based on the sequential approach, uses in general far fewer sensors but identifies the target after a larger number of steps. The two-step approach described in Section~\ref{section:twoStepsLocaalization} finds a balance between these extremes. Indeed, this hybrid strategy consists of an initial passive phase, where a relatively small number of sensors are placed to approximately locate the target. In the second active phase, fewer additional sensors are required to uniquely identify the target, as its approximate localization is already known. 

\subsection{Metric Dimension of Random Graphs}

 Our results in Section~\ref{subsection:randomTrees} establish the relaxed metric dimension in critical Galton-Watson trees, generalizing the work of~\cite{komjathy_Odor_GW_MD}. We note, however, that~\cite{komjathy_Odor_GW_MD} also derives the asymptotics of the metric dimension for another class of random trees, namely linear preferential attachment trees. We believe that the relaxed metric dimension of this class of trees can be obtained by performing an analysis similar to ours, and leave it for future work. 

 Beyond random trees, two classes of random graphs have been considered for the study of the metric dimension: \Erdos-\Renyi~\cite{bollobas2013metric,odor2021sequential} and random geometric graphs~\cite{lichev2023localization}. 

 On the one hand, \Erdos-\Renyi graphs show two different behaviors depending on the density regime. In the dense regime, where the edge density $p$ satisfies $n p = \omega( \log^5 n)$ and $n(1-p) \ge (3n \log\log n ) / \log n$, the authors in~\cite{bollobas2013metric} showed that the metric dimension exhibits a non-monotonous ``zig-zag'' behavior as a function of the average degree, when it increases from poly-logarithmic to linear in number of vertices $n$. This behavior arises because of the \textit{expansion property} of (dense) \Erdos-\Renyi graphs: the cardinality of the set of vertices at a certain graph distance from a given vertex $v$ does not differ much for most $v$.  Moreover, the work~\cite{odor2021sequential} establishes that the ratio between the sequential metric dimension and the metric dimension converges to a constant (and this constant equals 1 in some particular cases). This highlights that the power of adaptability is very limited in \Erdos-\Renyi graphs, and hints that, asymptotically, the relaxed-metric dimension of a dense \Erdos-\Renyi graphs should converge to the (non-relaxed) metric dimension.\footnote{In fact, the expansion property of dense \Erdos-\Renyi graphs implies that most vertices are at a distance $D_{\max}$ or $D_{\max}-1$, where $D_{\max}$ is the diameter of the graph. Hence, the relaxation for $k \le D_{\max} -2$ has only a limited effect.}

 \subsection{Is Stemming the key for Understanding the Relaxed Metric Dimension?}

 Our analysis of trees clearly demonstrates that the \textit{stem} (or the \textit{down-stem} for rooted trees) is a critical structure for determining the relaxed metric dimension. This naturally raises the question: could stemming also prove useful for analyzing the relaxed metric dimension in other types of graphs? Unfortunately, the answer appears to be negative—at least in the absence of additional assumptions. 

 Consider, for example, random geometric graphs. When the connection radius is sufficiently above the connectivity threshold, no vertex in these graphs has degree one. As a result, the stem operation leaves such graphs unchanged; the stem is the graph itself. Nevertheless, as observed in the numerical simulations of Section~\ref{section:experiments}, the relaxed metric dimension of random geometric graphs is often significantly smaller than their non-relaxed metric dimension. 

 Moreover, analyzing the equivalence classes of non-resolved vertices (Figure~\ref{fig:visualization_RGG}) suggests that an analogue of the stemming procedure for random geometric graphs might involve a form of graph coarsening based on the geometric positions of the vertices. Specifically, this process would construct a reduced graph $G'$ with fewer vertices than the original graph $G$, achieved by merging vertices that are close in the geometric space (and hence may form a clique) and  indistinguishable indeed by their distance to any other vertex. Such a coarsening approach would reduce the number of vertices while preserving the essential properties relevant to their relaxed metric dimension, much like the stemming process in trees. By focusing on spatial proximity, this coarsening could potentially provide insights in the way the relaxation reduces the number of required sensors, offering a new perspective on how geometric information influences sensor placement in these graphs. 

 Finally, there does not seem to be a consistent relationship between the size of the $1$-shell and the impact of relaxation on the metric dimension in real-world graphs. For instance, both the \textit{co-authorship} and \textit{copenhagen-friends} graphs have a small $1$-shell, yet the relaxed metric dimension is dramatically smaller than the standard metric dimension for \textit{copenhagen-friends}, whereas it remains almost unchanged for \textit{co-authorship}. Because the $1$-shell consists precisely of the vertices removed during the iterative stemming process\footnote{Recall that the $2$-core of a graph $G$ is the largest subgraph in which every vertex has degree at least $2$, and the $1$-shell is the set of vertices not belonging to the $2$-core.}, these observations suggest that stemming alone cannot fully explain the reduction in the number of sensors achieved through relaxation of the metric dimension in real graphs. Further investigation into graph-specific features is necessary to fully understand this phenomenon.

%% file: appendix.tex
\clearpage
\section{Table of notations}
\label{appendix:table_notations}

Table~\ref{tab:notations} summarizes the main notations used in the paper and the proofs. 
\begin{table}[!ht]
    \centering
    \begin{tabular}{@{}ll@{}}
        \toprule
        \textbf{Symbol} & \textbf{Description} \\
        \midrule
        \( G = (V,E), \ |V|=n \) & Graph with $n$ vertices and edge set $E$ \\
        \( T \) & Tree \\
        \( \md(G) \) & Ordinary metric dimension of graph $G$ \\
        \( \md_k(G) \) & $k$-relaxed metric dimension of graph $G$ \\
        \( \sigma(T) \) & Number of leaves in tree $T$ \\
        \( \ex(T) \) & Number of exterior major vertices in $T$\\
        \( \stem_r(T) \) & $r$-stem of tree $T$ \\
        \( \ds_r(T) \) & $r$-down-stem of tree $T$ \\
        \( \xi \) & Offspring distribution of a Galton-Watson tree\\
        \( T_v \) & Subtree of a rooted tree $T$ in the vertex $v$ (subtree facing away from root) \\
        \( d(u,v) \) & shortest path distance between vertices $u$ and $v$ \\
        \( S \subseteq V, S_k \subseteq V\) & Resolving set, $k$-relaxed resolving set\\
        \( \Phi_G(u,S) \) & Identification vector of vertex $u \in V$ w.r.t. resolving set $S$ \\
        \( B_i \) & $i^{\text{th}}$ branch of a tree\\
        \( \mathbf{1}_\ell \) & All-one vector of length $\ell$\\
        \( v_T \) & Closest vertex to $v$ in tree $T$ \\
        \( v_{\ex} \) & Closest exterior major vertex to $v$ in tree $T$ \\
        \( G_L \) & Leaf path (path of degree two vertices ending in a degree one vertex) \\
        \( (\mathcal{T}^{\text{GW}}_n)_{n \in \mathbb{N}} \) & Sequence of Galton-Watson trees conditioned on having $n$ vertices \\
        \( \mathcal{P} \) & Subtree property \\
        \( N_{\mathcal{P}} \) & Number of vertices with property $\mathcal{P}$ \\
        \( \mathcal{F} \) & unconditioned Galton-Watson tree \\
        \( u \stackrel{s}{\sim} v \) & Vertices equivalent under  $\Phi_G(\cdot,S)$\\
        \( [u]_S\) & Equivalence class of vertex $v$ under $\Phi_G(\cdot,S)$\\
        \( \alpha \) & Size of the largest equivalence class $[u]_S$ in graph $G$\\
        \bottomrule
    \end{tabular}
    \caption{Table of Notations}
    \label{tab:notations}
\end{table}

\section{Additional Proofs for Arbitrary Trees (Theorem~\ref{thm:k_metric_dim_tree})}

\subsection{Additional Lemmas Related to the Stem}

\begin{lemma} 
 \label{lm:stemmed_trees}
 Let $G$ be a graph and let $\stem_r(G)$ be its $r$-stem. No cycle in $G$ contains any vertex $v \in G \setminus \stem_r(G)$. 
\end{lemma}
\begin{proof}
 Consider a general graph $G = (V,E)$ and its subgraph $\stem_r(G)$ resulting from $r$ consecutive stemming operations. Let $v \in G \setminus \stem_r(G)$ be an arbitrary vertex that got removed in the stemming process. By contradiction, assume that $v$ is an element of a cycle in $G$. Any element of the cycle must have a degree of at least $2$. Therefore, no element of the cycle can be removed in any of the stemming operations. Therefore $v$ cannot be stemmed, which contradicts our assumption, and hence $v$ cannot have been an element of a cycle. 
\end{proof}

\clearpage
\subsection{Proof of Lemma~\ref{lm:lower_bound_dim2}} 
\label{sec:prf_lm_lower_bound_dim}

\begin{proof}[Proof of Lemma \ref{lm:lower_bound_dim2}]

Let $r$ be a non-negative integer and $S$ be any $2r$-relaxed resolving set of~$G$. If $\stem_r(G)$ is an isolated vertex, then $\stem_r(G)$ is neither a leaf nor an exterior major vertex. Hence, the right hand side of the inequality is zero, and the inequality $\md_{2r}(G) \ge 0$ is satisfied by any $G$. \\
If $\stem_r(G)$ has no exterior major vertex, it cannot have a leaf vertex. Assume by contradiction that $\stem_r(G)$ has a leaf but no exterior major vertex. Any non-path graph has to have at least one major vertex (degree at least three), hence there must be a closest major vertex to the leaf. Then, the closest major vertex to the leaf must also be exterior because the path to the leaf can only contain vertices of degree two and is therefore a leaf path. This contradicts the assumption and shows that if $\stem_r(G)$ has no exterior major vertex, it cannot have a leaf vertex. Hence if $\stem_r(G)$ has no exterior major vertex, the right hand side of the inequality is again zero. 

Therefore, assume that $\stem_r(G)$ contains at least one exterior major vertex. Let~$v$ be an arbitrary exterior major vertex in $\stem_r(G)$, and denote the neighbors of $v$ in $\stem_r(G)$ that are leaves as $u_1, u_2, \ldots, u_m$.
Recall from Lemma~\ref{lm:stemmed_trees} that to be a leaf in $\stem_r(G)$, $u_i$ has to be the root of an induced subtree $B_i$ of depth $r$ in $G$. We call this a branch. This subtree can only have been connected to $G$ via the edge $(v,u_i)$ because otherwise the vertex $u_i$ would be part of a cycle.

We first prove by contradiction that $S$ contains at least one vertex from each of the branches~$B_i$ with at most one exception, \textit{i.e.,} $ \left| \{ i \in [m] \colon B_i \cap S = \emptyset \} \right| \le 1 $. Indeed, suppose that two distinct branches, for example, $B_1$ and $B_2$, do not include any vertex from $S$. Let $w_1$ and $w_2$ be vertices in $G$ at distance $r+1$ from $v$ and on the branches $B_1$ and $B_2$, respectively. The existence of $w_1$ and $w_2$ is guaranteed since we have removed subtrees of height $r$ (rooted at $u_1$ and $u_2$) during the stem operation(s).

Since neither $B_1$ nor $B_2$ contains a vertex of $S$, it follows that the vertices ${w_1}$ and ${w_2}$ are indistinguishable with respect to $S$. In other words, for any vertex $x \in S$, the distances $d(x, w_1)$ and $d(x, w_2)$ are equal. Consequently, $\Phi_G(w_1, S) = \Phi_G(w_2, S)$ and $d(w_1,w_2) = 2r + 2$, which contradicts the assumption that $S$ is a $2r$-relaxed resolving set.

Therefore we have shown that $S$ contains at least one vertex from each of the branches $B_i$ $(1 \leq i \leq m)$, with at most one exception. 
This implies $\md_{2r}(G) \wge \sigma_r(G) - \ex_r(G)$. 
\end{proof}

\subsection{Proof of Lemma~\ref{upper_bound_dim2}} 
\label{sec:upper_bound_dim2}

\begin{proof}[Proof of Lemma~\ref{upper_bound_dim2}]
    Consider a tree $T'$ with a sub-tree $T$ as defined in Lemma \ref{upper_bound_dim2}. We begin by constructing a set $S$ of vertices from~$T$ as follows: for each exterior major vertex of $T$, we choose all its leaves except one and add them to the set $S$. It is known from the proof of Theorem~5 in~\cite{slater_leaves_1975} that $S$ is a resolving set for tree $T$. Let $|S| = m$. 

    Note that all vertices in $T$ are also distinguished in $T'$ as the shortest distances in $T$ are invariant between $T$ and $T'$. It remains to be shown that all vertices in $T' \setminus T$ are $2r$-distinguished from all other vertices in $T'$.
    
    Given a vertex $v\in T'$, we denote by $v_T \in T$ the closest vertex from $v$ belonging to $T$. In particular, $v=v_T$ if $v \in T$. We furthermore denote the all-one vector of length $|S|$ by $\mathbf{1}_{|S|}$.

    Consider an arbitrary vertex $v \in T'$ with $v_T$ its closest vertex in $T$. We argue by cases depending on the position of $v_T$ (see Fig. \ref{fig:lemma_vis} for an Example).
    
    \paragraph{Case 1} Suppose $v_T$ lies on a path between two vertices in $S$ and is not a exterior major vertex. Then, by Lemmas~\ref{lm:condition_for_easy_distinction} and~\ref{lm:phi_distinction}, $v$ is $2r$-distinguished from all other vertices in $T'$. 

    \paragraph{Case 2} Suppose $v_T$ lies on a path between a exterior major vertex $w$ and a leaf in $T$ that is not in~$S$.  Consider a second arbitrary vertex $u \in T'$ with its closest vertex in $T$ being $u_T$. We show that~$v$ can still be distinguished from any $u$.
    
    - If $u_T$ lies on a path between two vertices in $S$ and is not a exterior major vertex then by Case 1, $v$ and $u$ are $2r$-distinguished by $S$.

    - If $v_T$ and $u_T$ occur on paths from different exterior major vertices $w$ and $w'$ to leaves that are not in $S$, then by Lemma \ref{lm:distinguish_separate_me_leaves}, $v$ and $u$ are $2r$-distinguished by $S$.

    - Finally, if $v_T$ and $u_T$ occur on the same path from exterior major vertex $w$ to a leaf that are not in $S$, then by Lemma \ref{lm:distinguish_same_me_leaf}, if they have the same identification vector, the distance $d(v,u) \leq 2r$. Hence by contrapositive, $v$ and $u$ are $2r$-distinguished by $S$.

This concludes the proof of the Lemma.
\end{proof}

\begin{figure}[H]
    \centering
    \includegraphics[width=0.75\linewidth]{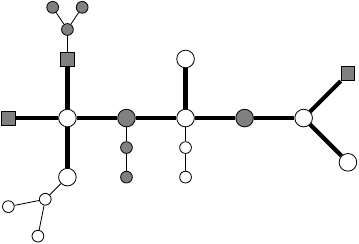}
    \caption{Example tree $T'$ for the cases in Lemma \ref{upper_bound_dim2}. In \textbf{bold} subtree $T$, square vertices are in $S$, Case $1$ vertices in gray and Case $2$ vertices in white.}
    \label{fig:lemma_vis}
    \Description[An example tree $T'$ that contains a subtree $T$, where vertices satisfying case 1 and 2 in the proof of Lemma \ref{upper_bound_dim2}.]{The tree $T'$ contains a subtree $T$. Sensors are placed in the leaves of $T$. Vertices belonging to case 1 or 2 in the proof are marked gray and white.}
\end{figure}

\begin{lemma}
 \label{lm:phi_distinction}
 Let $v_1 \in T$ and $u, v \in T'$ such that $d(v,u) >2r$ and $v_T = v_1$ and $u \in T'$. Suppose that~$v_1$ verifies the following condition:
 \begin{equation}
  \label{eq:v1v2_distinction}
  \forall v_2 \in T \setminus \{v_1\}, \nexists \alpha \in \mathbb{Z} \colon \quad \Phi_{T'}(v_1,S) \weq \Phi_{T'}(v_2,S) + \alpha \mathbf{1}_{|S|}.
 \end{equation}
 Then $u$ and $v$ are distinguished by the resolving set $S$.
\end{lemma}

\begin{proof}
 Consider two vertices $v, u \in T'$ with $d(v,u)>2r$ and $u_T = v_1$. If $v_T = u_T$, then by definition of $T'$, $d(v,u) <2r$. Thus, $v_T \neq u_T$.

 We prove the result by contradiction. Let us assume $\Phi_{T'}(v,S) = \Phi_{T'}(u,S)$. 
 Notice firstly that for all $w \in T$ we have $d(u,w) = d(u,u_T) + d(u_T,w)$, and a similar relationship holds for $d(v,w)$. Hence, because all vertices in $S$ belong to $T$, we obtain
 \begin{align*}
  \Phi_{T'}(v,S) &\weq \Phi_{T'}(v_T, S) + d(v,v_T) \mathbf{1}_{|S|}, \\
  \Phi_{T'}(u,S) &\weq \Phi_{T'}(u_T, S) + d(u,u_T) \mathbf{1}_{|S|}.
 \end{align*}
 Using the assumption $\Phi_{T'}(v,S) = \Phi_{T'}(u,S)$, we have
 \begin{align*}
  \Phi_{T'}(u_T, S) \weq \Phi_{T'}(v_T, S) + \left( d(v,v_T) - d(u,u_T) \right) \mathbf{1}_{|S|}, 
 \end{align*}
which contradicts~\eqref{eq:v1v2_distinction}. Hence, $v$ and $u$ must have different identification vectors and therefore be distinguishable with respect to $S$.
\end{proof}

\begin{lemma} \label{lm:condition_for_easy_distinction}
  A vertex $v_1$ satisfy Condition~\eqref{eq:v1v2_distinction} if the following two conditions are verified:
    \begin{itemize}
        \item $v_1$ is an element of the shortest path between two vertices in $S$;
        \item $v_1$ is not an exterior major vertex. 
    \end{itemize}
\end{lemma}

\begin{proof}
    Consider a vertex $v_1 \in T$ such that $v_1$ is not an exterior major vertex and is on the shortest path between two vertices in $S$.
    Let $\{B_i\}_{i \leq \text{deg}(v_1)}$ be the subtrees of $v_1$ induced by the connected components in $T \setminus \{v_1\}$, i.e. the subtrees or branches created by removing $v_1$ from the tree. 
    
    Note that each of these subtrees has to contain at least one vertex in $S$ because either $v_1$ has a degree less than three, and is on the path between two vertices in $S$ which means all $B_i$ contain vertices of $S$. Or else $v_1$ has a degree greater or equal to three but is not a exterior major vertex. Hence, each $B_i$ contains at least one exterior major vertex and with it some vertex in $S$.
    
    Consider an arbitrary vertex $v_2 \in T \setminus \{v_1\}$ and the path between $v_1$ and $v_2$ of length $d(v_1,v_2)$. For any two neighboring vertices the identification vector $\Phi_{T'}$ has to change in every component by $\pm 1$ due to the tree structure of $T'$. Let $B_{v_2}$ be the subtree of $v_1$ that $v_2$ belongs to. Then, the distance of any sensor $q \in S \setminus B_{v_2}$ has to be $d(q, v_2) = d(q, v_1) +d(v_1, v_2)$. We observe how the identification vector $\Phi_{T'}(\cdot,S)$ changes on the path between $v_1$ and $v_2$. As the distance to $q$ only increases on a path between $v_1$ and $v_2$, the identification vector component $[\Phi_{T'}(\cdot,S)]_q$ belonging to sensor $q$ has to increase.

    Consider $q' \in S \cap B_{v_2}$, a sensor in the branch containing $v_2$. On the path between $v_1$ and $v_2$ the distance to $q'$ has to decrease at least on the first step. Hence, $d(q', v_2) < d(q', v_1) + d(v_1, v_2)$. This implies that for all $\alpha \in \mathbb{Z}$
    \begin{align*}
        \Phi_{T'}(v_1,S) \wne \Phi_{T'}(v_2,S) + \alpha \mathbf{1}_{|S|}.
    \end{align*}
\end{proof}

\begin{lemma} \label{lm:distinguish_separate_me_leaves}
    If $u, v \in T'$ and $d(u,v) >2r$ such that $u_T, v_T$ occur on paths from different exterior major vertices to leaves, then $u$ and $v$ are distinguishable by $S$.
\end{lemma}

\begin{proof}
    First note that for any two exterior major vertices $w \neq w' \in T$ we have
    \begin{equation} \label{eq:distinguish_maj_ext}
        \nexists \alpha \in \mathbb{Z} \colon \Phi_{T'}(w, S) = \Phi_{T'}(w', S) + \alpha \mathbf{1}_{|S|}.
    \end{equation}
    This is because both $w$ and $w'$ lie on some shortest path between two vertices $q, q' \in S$. Without loss of generality assume the path is ordered as $q, w, w', q'$. Then $d(w,q) < d(w',q)$ but also $d(w,q')> d(w',q')$, which shows Statement \ref{eq:distinguish_maj_ext}.

    Now consider two vertices $u,v \in T'$ as stipulated in Lemma \ref{lm:distinguish_separate_me_leaves}. We denote the closest vertices in $T$ as $u_T$ and $v_T$ and the closest exterior major vertices in $T$ as $u_{\text{ex}}$ and $v_{\text{ex}}$ respectively (potentially $u_T = u_{\text{ex}}$ or $v_T = v_{\text{ex}}$). We write their identification vectors as
    \begin{align}
        \Phi_{T'}(u, S) &= \Phi_{T'}(u_T, S) +d(u_T, u)\mathbf{1}_{|S|} = \Phi_{T'}(u_{\text{ex}}, S) +(d(u_{\text{ex}},u_T) +d(u_T, u))\mathbf{1}_{|S|} \\
        \Phi_{T'}(v, S) &= \Phi_{T'}(v_T, S) +d(v_T, v)\mathbf{1}_{|S|} = \Phi_{T'}(v_{\text{ex}}, S) +(d(v_{\text{ex}},v_T) +d(v_T, v))\mathbf{1}_{|S|}.
    \end{align}
    By contradiction, assume $\Phi_{T'}(u,S) = \Phi_{T'}(v,S)$, then
    \begin{align}
        \Phi_{T'}(u_{\text{ex}}, S) +(d(u_{\text{ex}},u_T) +d(u_T, u))\mathbf{1}_{|S|} &= \Phi_{T'}(v_{\text{ex}}, S) +(d(v_{\text{ex}},v_T) +d(v_T, v))\mathbf{1}_{|S|} \\
        (\Leftrightarrow) \ \ \ \ \Phi_{T'}(u_{\text{ex}}, S) &= \Phi_{T'}(v_{\text{ex}}, S) +\alpha \mathbf{1}_{|S|}.
    \end{align}
    Where $\alpha = d(v_{\text{ex}},v_T) +d(v_T, v) -d(u_{\text{ex}},u_T) -d(u_T, u)$ which clearly contradicts Statement \ref{eq:distinguish_maj_ext}. Hence $u$ and $v$ have to be distinguishable.
\end{proof}

\begin{lemma} \label{lm:distinguish_same_me_leaf}
    If $u, v \in T'$ such that $u_T, v_T$ occur on the same path between an exterior major vertex and a leaf not in $S$ and they have the same identification vector, then their distances is less than or equal to $2r$.
\end{lemma}
\begin{proof}
    Consider $u,v \in T'$ such that $u_T$ and $v_T$ both lie on a path between a leaf not in $S$ and an exterior major vertex $w \in T$. Assume that they have the same identification vector $\Phi_{T'}(u, S) = \Phi_{T'}(v, S)$. Hence $d(w,u) = d(w,v)$. Without loss of generality assume $d(u_T, w) \leq d(v_T, w)$. We write

    \begin{alignat}{2}
         \quad && d(w, u) &= d(w,v)\\
         (\Leftrightarrow) \quad &&  d(w,u_T) + d(u_T, u) &= d(w, u_T) + d(u_T, v_T) + d(v_T, v) \\
        (\Leftrightarrow) \quad &&   2 d(u_T, u) &= \underbrace{d(u, u_T) + d(u_T, v_T) + d(v_T, v)}_{d(u,v)}
    \end{alignat}
And as the distance $d(u_T, u)$ is by definition less than or equal to $r$, we have shown that $d(u,v)$ must be less than or equal $2r$.  
\end{proof}

\section{Additional Proofs for Random Trees}
In this section, we prove all the lemmas related to the proof of the metric dimension in Galton-Watson trees. For reader's convenience, we restate the various lemmas. 

\subsection{Difference between the down-stem and the stem}
\label{appendix:sec_proof_lm_ds_leaf_path}

The following lemma shows that the vertices belonging to the down-stem but not to the stem form a line graph. We denote the graph difference between two graphs $G = (V,E)$ and $G' = (V',E')$ as $G -G'$, the subgraph of $G$ induced by $V \setminus V'$.

\begin{restatable}{lemma}{StDsLeafPath} \label{lm:ds_leaf_path}
  Consider an arbitrary rooted tree $T$ and a positive integer $r$ such that $2r < \text{diam}(T)$. Then, $\ds_r(T) - \stem_r(T)$ is a line graph (potentially of length zero, if the root has a degree larger than one at every stemming iteration).
\end{restatable}

 \begin{proof}  Let $T$ be an arbitrary rooted tree and we denote $L_r(T) = \ds_r(T) - \stem_r(T)$. We prove by induction over $r$ the following hypothesis $\cH(r)$:
 \begin{align*}
  \cH(r) \ \colon \ L_r(T) \text{ is a line graph (potentially empty)}.
 \end{align*}

 \paragraph{Base Case $(r=0)$.} By the definition of the stemming and down-stemming operations $\stem_0(T) = \ds_0(T) = T$. Hence, $L_0 = \emptyset$ and thus $\cH(0)$ holds. 

 \paragraph{Inductive Step.} Suppose that $\cH(r)$ holds for some $r \ge 0$, and let us prove $\cH(r+1)$. 
 Consider the set $S_r$ of vertices in $\stem_r(T)$ having degree~1. By the definition of the stemming operation, these vertices are removed from $\stem_r(T)$ to obtain $\stem_{r+1}(T)$. 
 Similarly, denote $S'_r$ the set of non-root vertices in $\ds_r(T)$ having degree~1; these are the vertices removed from $\ds_r(T)$ to obtain $\ds_{r+1}(T)$. 
 Finally, denote by $L_r$ the graph $\stem_r(T) - \ds_r(T)$. By the induction hypothesis, $L_r$ is a (potentially empty) line graph. 

 If $L_r$ is empty, then $\stem_{r}(T) = \ds_{r}(T)$ and thus the difference $L_{r+1}$ between $\stem_{r+1}(T)$ and $\ds_{r+1}(T)$ is either empty (if the root has degree greater or equal than $2$) or the singleton $\{ \rm{root} \}$ composed of the root (if the root has degree 1). In any cases, $L_{r+1}$ is a line graph, and thus $\cH(r+1)$ holds. 
  
 Suppose now that $L_r$ is not empty (notice that it implies that the root of $T$ does not belong to $\stem_r(T)$). We consider $\ds_r(T)$ as a graph union of $\stem_r(T)$ and $L_r$. Let $v \in \stem_r(T) \cap L_r$, the vertex where the leaf path $L_r$ is attached to the stem. We discuss two cases depending on whether $v$ has degree $1$ in $\stem_r(T)$ and would be removed in the next iteration or not:
 \begin{enumerate}
  \item[(i)] Suppose $v \not\in S_r$. Then, we obtain $\ds_{r+1}(T)$ from $\ds_{r}(T)$ by deleting the set $S_r'$ of non-root vertices in $\ds_{r+1}(T)$ of degree $1$. Because $v \not\in S_r$ and the root does not belong to $\stem_r(T)$, we have $S'_r = S_r$. Hence, $L_{r+1} = L_r$ and thus the difference between $\ds_{r+1}(T)$ and $\stem_{r+1}(T)$ remains a line graph. 
  \item[(ii)] Suppose $v \in S_r$. Then, vertex $v$ has degree~2 in $\ds_r(T)$ and hence $v \not\in S_r'$ (or equivalently, $v \in \ds_{r+1}(T)$). Because $v \not\in \stem_{r+1}(T)$, the vertex~$v$ belongs to $L_{r+1}$. More precisely, we have $S_{r}' = S_r \setminus \{v\}$ and thus $L_{r+1}-L_r = \{v\}$. Because $v$ is attached to $L_r$, $L_{r+1}$ remains a line graph. 
 \end{enumerate}
 In any cases, $\cH(r+1)$ holds. 
\end{proof}

\subsection{Proof of Lemma~\ref{lm:st_probs_close_to_md}}
\label{appendix:st_probs_close_to_md}

We recall the statement of Lemma~\ref{lm:st_probs_close_to_md}. 

\StProbsVsMD*

The proof of this lemma consists of the careful analysis of how the true number of leaves and exterior major vertices in the stem can differ from the numbers estimated by the subtree properties  $\mathcal{P}_r^L$ and $\mathcal{P}^E_r$ on the $\ds_r(T)$ graph (see Figure~\ref{fig:over_under_counting}). We show that it is possible to miss at most one leaf and that we can either miss or over-count an exterior major vertex in two different ways. However, these over and under-counting events are mutually exclusive which means that we only over- or underestimate the relaxed metric dimension by at most one. 
Still, this is quite a tedious proof because we need to bridge the gap both between the stem and down-stem as well as the true graph structure and what we can count via the subtree properties.

\begin{proof}
Consider an arbitrary tree $T$ and a nonnegative integer $r$.

    \underline{Claim 1:} $\sigma_r(T)- N^L_r(T) \in \{0,+1\}$
    
    Recall that $\sigma_r(T)$ denotes the number of leaves in the $\stem_r(T)$ and for ease of notation assume that if the root gets removed in one of the stemming iterations, the root assignment moves to its only descendant. 
    
    Consider a vertex $v \in T$ such that $T_v \in \mathcal{P}_r^L$. Then, $v$ is a leaf in $\stem_r(T)$ as well as in  $\text{Down-}\stem_r(T)$. Hence, $\sigma_r(T) \geq N^L_r$.
    
    Now consider a leaf vertex in $\stem_r(T)$ or $\text{Down-}\stem_r(T)$ which is not the root. It must have had a subtree of depth $r$ in $T$ which gets removed through the stemming procedure. Hence, $T_v \in \mathcal{P}_r^L$. However, if the root is a leaf in $\stem_r(T)$ the subtree may not satisfy the property $\mathcal{P}_r^L$ (see Figure \ref{fig:ou_ctr_cl1}). Similarly, if the root in the $\text{Down-}\stem_r(T)$ is a leaf, the subtree may also not satisfy $\mathcal{P}_r^L$. Therefore, $\sigma_r(T) \leq N^L_r +1$.

    Next, we need to bound the error on our estimate of the exterior major vertices in the true stem $\stem_r(T)$. There are two sources of error, the first one (Claim 2) originates from the fact that we are counting occurrences of subtree properties instead of the actual number of exterior major vertices. The second error comes from using a the $\ds_r(T)$ instead of the true $\stem_r(T)$ to count these properties (Claim 3). Finally, we combine these sources of error and consider that some cases cannot occur together (Claim 4). 
    
    \underline{Claim 2:} $\ex(\text{Down-}\stem_r(T)) - N^E_r(T) \in \{-1, 0, +1\}$

    Consider a vertex $v \in V$ that satisfies $T_v \in \mathcal{P}_r^E$, if $v$ is not the root of $T$, this is a vertex with degree greater or equal to three and with a line graph to a leaf. Hence, this vertex is an exterior major vertex in $\ds_r(T)$ and $\ex(\text{Down-}\stem_r(T)) - N^E_r(T) \geq -1$.

    Next, consider $v \in \ds_r(T)$ and $v$ is an exterior major vertex, hence it has three or more neighbors and one of them is a path to a leaf. This means that $T_v \in \mathcal{P}_r^E$ unless the leaf path goes through the parent of vertex $v$ (see Figure \ref{fig:ou_ctr_cl2b}). If $v$ has a leaf path only through its parent, its subtree does not contain a leaf path and $\mathcal{P}_r^E$ is not satisfied. We will now show that only one vertex can have a leaf path through its parent. By contradiction, assume both a vertex $u$ and $v$ are exterior major vertices and have a leaf path through their parent. Then both vertices must be connected to the root via a leaf path, hence $u$ is on the leaf path to $v$ and vice versa. But both $u$ and $v$ have a degree greater than two and can therefore not be part of a leaf path (except for as the starting vertex). Hence only one vertex that is an exterior major vertex but does not satisfy $\mathcal{P}_r^E$ can exist and we summarize $\ex(\text{Down-}\stem_r(T)) - N^E_r(T) \leq +1$.

    \underline{Claim 3:} $ \ex(\stem_r(T)) - \ex(\ds_r(T)) \in \{0, -1\}$

    By Lemma~\ref{lm:ds_leaf_path}, $\ds_r(T)$ and $\stem_r(T)$ only differ by a leaf path $G_L$. 
    
    Consider an exterior major vertex $v$ in $\stem_r(T)$. This will also be an exterior major vertex in $\ds_r(T)$ except for if the leaf path $G_L$ is attached at the only leaf path of $v$ such that it is no longer a leaf path (i.e. one of the vertices on the path becomes a degree three vertex). This can only occur once in $\stem_r(T)$ because $G_L$ can only be part of one leaf path. But in this case, the vertex with degree three becomes a new exterior major vertex and the overall number of exterior major vertices is preserved. Hence $\ex(\ds_r(T)) - \ex(\stem_r(T)) \geq 0$.

    Now consider an exterior major vertex $w$ in $\ds_r(T)$, this is also an exterior major vertex in $\stem_r(T)$, except for if $w$ has degree 3 and is the vertex where the leaf path $G_L$ is attached to $\stem_r(T)$ (see for example Figure \ref{fig:ou_ctr_cl3}). This can either result in another exterior major vertex not being such in $\ds_r(T)$ as outlined in the paragraph above, or it results in an additional exterior major vertex in $\ds_r(T)$. Hence $\ex(\ds_r(T)) - \ex(\stem_r(T)) \leq 1$

    \underline{Claim 4:} $\ex(\stem_r(T))-N^E_r(T) \in \{-1, 0, +1\}$

    Based on Claim 2 and 3, this could only be violated if $\ex(\text{Down-}\stem_r(T)) - N^E_r(T) = -1$ and $ \ex(\stem_r(T)) - \ex(\ds_r(T)) = -1$ which would lead to $\ex(\stem_r(T))-N^E_r(T) = -2$. 
    As detailed in Claim 2, the first overestimate $N^E_r = \ds_r(T) +1$ can only occur if the root satisfies property $\mathcal{P}_r^E$, i.e. has exactly two children one of them a leaf path but no parent, which makes it not an exterior major vertex. Note the root having degree two implies that the $\ds_r(T) = \stem_r(T)$. Hence this precludes the second overestimate, which proves the statement.

    Finally, we combine the estimates of the number of leaves with the number of exterior major vertices. Again, we check whether we can underestimate the number of leaf vertices while overestimating the number of exterior major vertices:
    
    \underline{Claim 5:} $\sigma_r(T)- N^L_r(T) - (\ex(\stem_r(T))-N^E_r(T)) \in \{-1, 0, +1\}$
    
    Again we check whether $\sigma_r(T)- N^L_r(T) = +1$ and $\ex(\stem_r(T))-N^E_r(T) = -1$ can occur together. Note that we under estimate the number of leaves only if the root in $\stem_r(T)$ has degree one, which immediately precludes the case of overestimating the number of exterior major vertices due to the root being incorrectly identified as an exterior major vertex.

    The second overestimate of the number of exterior major vertices occurs if, as explained in Claim~3, the additional leaf path $G_L$ creates an exterior major vertex not present in the $\stem_r(T)$. The exterior major vertex $w$ is created exactly where the leaf path $G_L$ joins the $\stem_r(T)$. However in this case, the root in $\stem_r(T)$ must have been stemmed and by definition, moved to exactly where the leaf path meets the $\stem_r(T)$. Hence, $w$ is the root of the $\stem_r(T)$ and of degree two, which means no leaf is underestimated.
    
    As shown above, these are the only two cases in which the number of exterior major vertices can be overestimated. As they contradict the underestimation of a leaf, we verify Claim~5.

    Finally, note that by Theorem \ref{thm:k_metric_dim_tree}, $\md_{2r}(T) = \sigma_r(T) - \ex(\stem_r(T))$. Hence by Claim 5,
    \begin{equation}
        \md_{2r}(T) = N^L_r(T) - N^E_r(T) + \varepsilon
    \end{equation}
    where $\epsilon \in \{-1, 0, +1\}$. 
    \end{proof}

\begin{figure}[ht]
\centering
    \begin{subfigure}[b]{0.23\textwidth}
        \centering
        \includegraphics[width=\textwidth]{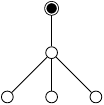}
        \caption{Black vertex is a leaf vertex but does not satisfy $\mathcal{P}^L_r$. \newline}
        \label{fig:ou_ctr_cl1}
    \end{subfigure}
    \hfill
    \begin{subfigure}[b]{0.23\textwidth}
        \centering
        \includegraphics[width=\textwidth]{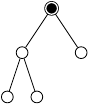}
        \caption{Black vertex satisfies $\mathcal{P}^E_r$ but is not an exterior major vertex. \newline}
        \label{fig:ou_ctr_cl2a}
    \end{subfigure}
    \hfill
    \begin{subfigure}[b]{0.23\textwidth}
        \centering
        \includegraphics[width=\textwidth]{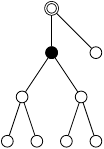}
        \caption{Black vertex is an exterior vertex but does not satisfy $\mathcal{P}^E_r$ \newline}
        \label{fig:ou_ctr_cl2b}
    \end{subfigure}
    \hfill
    \begin{subfigure}[b]{0.23\textwidth}
        \centering
        \includegraphics[width=\textwidth]{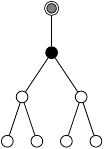}
        \caption{Black vertex is an exterior major vertex in $\ds_r(T)$, but not in $\stem_r(T)$}
        \label{fig:ou_ctr_cl3}
    \end{subfigure}
    \caption{Examples of $\ds_r(T)$ where the subtree properties $\mathcal{P}^L_r$ and $\mathcal{P}^E_r$ over- or under-count leaves or exterior major vertices. Root as a double lined vertex, the offending vertex as black, and a vertex in the down-stem but not in the stem in gray.}
    \label{fig:over_under_counting}
    \Description[Examples where the defined subtree properties over- or under-count the true number of leaves and exterior major vertices]{Examples of $\ds_r(T)$ where the subtree properties $\mathcal{P}^L_r$ and $\mathcal{P}^E_r$ over- or under-count leaves or exterior major vertices (as outlined in proof of Lemma \ref{lm:st_probs_close_to_md}). Root as a double lined vertex, the offending vertex as black, and a vertex in the down-stem but not in the stem in gray.}
\end{figure}

\subsection{Proof of Lemmas~\ref{lm:ell_k_e_k} and~\ref{lm:d_k_m_k}}
\label{appendix:proof_auxilliary_quantities_random_trees}

 In all the following, $\cF$ denotes an unconditioned Galton-Watson tree with offspring distribution $\xi$. We also let $p_i = \mathbb{P}(\xi = i)$. We recall the quantities $d_r$, $\ell_r$, $s_r$ and $e_r$ defined in~\eqref{def:expression_ell_k_e_k} and~\eqref{def:expression_d_k_s_k}. We prove Lemmas~\ref{lm:ell_k_e_k} and~\ref{lm:d_k_m_k} by establishing the expressions of the quantities $d_r$, $\ell_r$, $s_r$ and $e_r$. 

 \subsubsection{Expression of $d_r$}
 Recall that 
 \begin{align*}
   d_r & \weq \mathbb{P}\left(\mathcal{F} \text{ is of height strictly less than } r\right).
 \end{align*}
 We will establish that 
    \begin{equation} \label{eq:expression_d_r}
        d_r \weq
        \begin{cases}
            0 &\text{for } r=0, \\
            \sum_{j=0}^{\infty} p_j d^j_{r-1} &\text{for } r \geq 1.
        \end{cases}
    \end{equation}
    Assume $r=0$. Because, by definition, the height of a subtree cannot be strictly less than zero, we indeed have $d_0 = 0$. Assume now that $r \ge 1$. Using the law of total probability over the possible number of children of the root, we have 
    \begin{equation*}
        d_r = \sum_{j=0}^{\infty} p_j \ \mathbb{P}(\text{each child is the root of a subtree of height less than } r-1 \cond \text{the root has }j \text{ children}).
    \end{equation*}
    Because each child subtree is identically distributed and independent of the other subtrees, we have 
    \begin{align*}
      \mathbb{P}(\text{each child is the root of a subtree of height less than } r-1 \cond \text{the root has }j \text{ children}) \weq d_{r-1}^j,
    \end{align*}
    and therefore $d_r = \sum_{j=0}^{\infty} p_j d_{r-1}^j$. This establishes~\eqref{eq:expression_d_r}. 

 \subsubsection{Expression of $\ell_r$} Recall that 
 \begin{align*}
  \ell_r \weq \mathbb{P}\left( \cF \text{ is of height } r \right),
 \end{align*}
 and we establish that 
    \begin{equation*}\label{eq:expression_e_r}
        \ell_r \weq
        \begin{cases}
            p_0 &\text{for } r=0, \\
            \sum_{j=0}^{\infty} \left[ p_j (d_{r-1} + \ell_{r-1})^j \right] -d_r &\text{for } r \geq 1.
        \end{cases}
    \end{equation*}
    Suppose $r=0$. Observe that $\ell_0$ is the probability that $\cF$ has height 0, \textit{i.e.,} $\cF$ is simply the root. This arise only if the root has no children, and thus has $\ell_0 = p_0$. 

    Suppose now $r \ge 1$. If $\cF$ is a leaf after $r$ down-stemming operations, then it must have height exactly $r$. Applying the law of total probability over the offspring distribution, we obtain 
    \begin{equation*}
        \ell_r \weq \sum_{j=0}^{\infty} p_j \mathbb{P}(\mathcal{L}_{r-1,j}| \text{ the root has } j \text{ children}),
    \end{equation*}
    where $\cL_{r-1,j}$ is the event that each child-subtree of the root has height less than or equal to $r-1$ and one of the child-subtrees has height exactly $r-1$. The probability of having a subtree of height less than $r-1$ is $d_{r-1}$, and the probability of having a subtree of height exactly $r-1$ is $\ell_{r-1}$. Hence, 
    \begin{equation*}
      \mathbb{P}(\mathcal{L}_{r-1,j}|\text{the root has $j$ children}) \weq (d_{r-1} +\ell_{r-1})^j -d_{r-1}^j.
    \end{equation*}
    Hence, 

    \begin{align*}
     \ell_r 
     & \weq \sum_{j=0}^{\infty} p_j\left[ (d_{r-1} +\ell_{r-1})^j -d_{r-1}^j \right] \\ 
     & \weq \sum_{j=0}^{\infty} \left[p_j (d_{r-1} +\ell_{r-1})^j \right] -d_r.
    \end{align*}

 \subsubsection{Expression of $s_r$} Recall that 
   \begin{align*}
    s_r \weq \mathbb{P}\left( \ds_r(\cF) \text{ is a line graph} \right),
  \end{align*}
  and we establish that 
    \begin{equation*}
        s_r \weq \frac{\ell_r}{1- \sum_{j=1}^{\infty} p_jjd^{j-1}_r}.
    \end{equation*}
    The $\ds_r(\mathcal{F})$ is a line graph if it is a single vertex or the stemming removes all but one child which has a path underneath it. The leaf case is captured by the probability $\ell_r$ while the second probability can once more be expressed as a sum over the possible number of children. We obtain 
    \begin{equation*}
        s_r = \ell_r + \sum_{j=1}^\infty p_j \mathbb{P}\left(\mathcal{H}_{r,j}| \text{the root has $j$ children} \right)
    \end{equation*}
    where $\mathcal{H}_{r,j}$ is the event that after stemming, one of the $j$ child-subtrees is a path and all other child-subtrees have been removed. A child subtree is removed by $r$ times down-stemming if the subtree is of depth less than $r$. The probability of this event we denoted above by $d_r$. Hence, with the choice of one of the $j$ subtrees to be the path we write:
    \begin{equation} \label{eq:prop_line_and_no_sibs}
        \mathbb{P}\left(\mathcal{H}_{r,j}| \text{the root has $j$ children} \right) = j s_rd^{j-1}_r.
    \end{equation}
    Hence,
    \begin{equation*}
        s_r = \ell_r + s_r\sum^\infty_{j=1}p_j jd_r^{j-1} =\frac{\ell_r}{1- \sum^\infty_{j=1}p_j jd_r^{j-1}}.
    \end{equation*}

 \subsubsection{Expression of $e_r$} Finally, recall that \begin{align*}
  e_r \weq \mathbb{P}\left( \ds_r(\mathcal{F}) \text{ has at least two children and one child subtree is a leaf path} \right).
\end{align*}
We establish that 
    \begin{equation*}
        e_r = 1 - \sum_{j=0}^\infty p_j (1-s_r)^j - s_r +\ell_r.
    \end{equation*}
    We express the probability $e_r = \mathbb{P}( \mathcal{F} \in \mathcal{P}_r^E) = 1- \mathbb{P}(\mathcal{F} \notin \mathcal{P}_r^E)$. Recall that the property $\mathcal{P}^E_r$ requires the root of the $\ds_r(\mathcal{F})$ to have at least two children and one of the child-subtrees has to be a line graph. Hence, the negation of this is for the down-stem to have no line graph sub-tree or to have a line graph subtree but no other subtree (event $\mathcal{H}_{r,j}$). This occurs either if the original $\mathcal{F}$ has less than two children (probability $p_0 +p_1$) or the root of $\mathcal{F}$ has more than two children but when they get stemmed they loose at least one of the two properties above. 
    
    Conditioned on the root of $\mathcal{F}$ having $j\geq 2$ children, we write the probability of having no line graph subtrees in the children as:
    \begin{equation*}
        \mathbb{P}(\text{no child subtree of } \ds_r(\mathcal{F})\text{ is a line graph}| \text{ the root has $j$ children}) \weq (1-s_r)^j.
    \end{equation*}
    The second condition is disjoint from the first and can be written as in Equation \eqref{eq:prop_line_and_no_sibs}.
    Hence, we obtain:
    \begin{align*}
        e_r &=  1- \mathbb{P}(\mathcal{F} \notin \mathcal{P}_r^E) \\
        &= 1- p_0 -p_1 - \sum_{j=2}^\infty p_j \left((1-s_r)^j +s_rjd^{j-1}_r\right) \\
        &= 1- p_0 -p_1 - \sum_{j=0}^\infty p_j (1-s_r)^j +p_0 +p_1(1-s_r) -\sum_{j=1}^\infty p_js_rjd^{j-1}_r + p_1s_r \\
        &= 1- \sum_{j=0}^\infty p_j (1-s_r)^j - (s_r -\ell_r).
    \end{align*}

\section{\texorpdfstring{Computing the $k$-relaxed metric dimension of a graph}{Computing the k-relaxed metric dimension of a graph}}
\label{appendix:algo}
In this section, we present an algorithm that can approximate the $k$-relaxed metric dimension of a graph with $n$ vertices in polynomial time within a factor of $\cO(\log n)$. First, we show that there is an approximation preserving reduction from the problem of finding $\md_k(G)$ to the set cover problem. This reduction draws significant inspiration from the one depicted in~\cite{khuller1996landmarks}, which addresses the $\md_0(G)$ problem. Building upon this reduction, we can then use the $\mathcal{O}(\log{}n)$ factor approximation algorithm for the set cover problem~\cite{set_cover} to obtain an approximation algorithm for the $k$-relaxed metric dimension problem. 

\begin{theorem}
\label{thm:validity_algo}
Given an arbitrary graph $G = (V, E)$ with $n$ vertices, then $\md_k(G)$ can be approximated within a factor of $\mathcal{O}(\log{}n)$ in $\mathcal{O}(n^3)$ time-complexity.
\end{theorem}

\begin{proof}
We construct an instance of the set cover problem from $G$. The intuition is that every pair of distinct vertices separated by a distance greater than $k$ must be distinguished by a sensor. We can easily compute all the pairs of vertices that are distinguished by placing a sensor on a given vertex. The $k$-relaxed metric dimension problem is that of finding a set of vertices of minimum cardinality such that every pair of vertices (at a distance greater than $k$) is distinguished by some vertex in this set. The elements of the universe (in the set cover problem) correspond to pairs of vertices $\{u, v\}$ of $G$ such that $u \neq v$ and $d(u,v) > k$. For each vertex $v \in V$, we place the set of all pairs of vertices which are distinguished by placing a sensor at $v$ into a single subset $S_v$. Therefore there are a maximum of $\binom{n}{2}$ elements and $n$ subsets in the set cover problem. Moreover, there is a set cover of size $m$ if and only if there exists a minimal cardinality $k$-relaxed resolving set of size $m$ in $G$. Finding a set cover within a factor of $\mathcal{O}(\log{}n)$ therefore yields the same approximation for the $k$-relaxed metric dimension problem. 
\end{proof} 

Hence, the following simple greedy heuristic is able to approximate the $k$-relaxed metric dimension of a graph with $n$ vertices in polynomial time within a factor of $\mathcal{O}(\log{}n)$.

\RestyleAlgo{ruled}

\SetKwComment{Comment}{/* }{ */}

\begin{algorithm}[hbt!]
\caption{An algorithm to estimate the $k$-relaxed metric dimension of a graph}\label{algo}
\KwData{$G = (V,E)$ and $k$ an integer}
\KwResult{$C$, a $k$-relaxed resolving set for $G$}
\tcp{$U$ is the universe of all pairs that should be distinguished}\label{cmt}
$U \gets \emptyset$ \\
\For{$\{u,v\}$ in $G.V \times G.V$}{
\If{$d(u, v) > k$}{
$U \gets U \cup \{u,v\}$
}
}
\tcp{$S_i$ is the set of pairs distinguishable by vertex $i$}

\For{$i$ in $G.V$}{
$S_i \gets \emptyset$ \\
\For{$\{u,v\}$ in $G.V \times G.V$}{
\If{$d(i, u) \neq d(i,v)$}{
$S_i \gets S_i \cup \{u,v\}$
}
}
}
$C \gets \emptyset$ \\
\While{$U \neq \emptyset$}{
select $S_i$ that maximizes $|S_i \cap U|$ \\
\Indp
$U \gets U - S_i$ \\
$C \gets C \cup \{i\}$
}
\Return $C$
\end{algorithm}

\section{Additional Numerical Experiments}
\label{appendix:additional_numerical_results}

\subsection{Additional Numerical Results on the Equivalent Classes}

\subsubsection{Number of non-resolved vertices and largest equivalent class}
 Finally, we finish our exposition of synthetic graphs by comparing the size $\alpha$ of the largest equivalent class with the number of non-resolved vertices, as a function of the relaxation parameter $k$. Results for different graphs are shown in Figure~\ref{fig:non_resolved_vs_equiv_class}. We observe that, while the number of non-resolved vertices increases very fast (and then saturates when almost every vertices is non-resolved), the size of the largest equivalent class always remains small and only increases dramatically when $k$ becomes close to the graph's diameter. 

 \begin{figure}[!ht]
    \centering
    \begin{subfigure}{0.32\textwidth}
        \includegraphics[width=\linewidth]{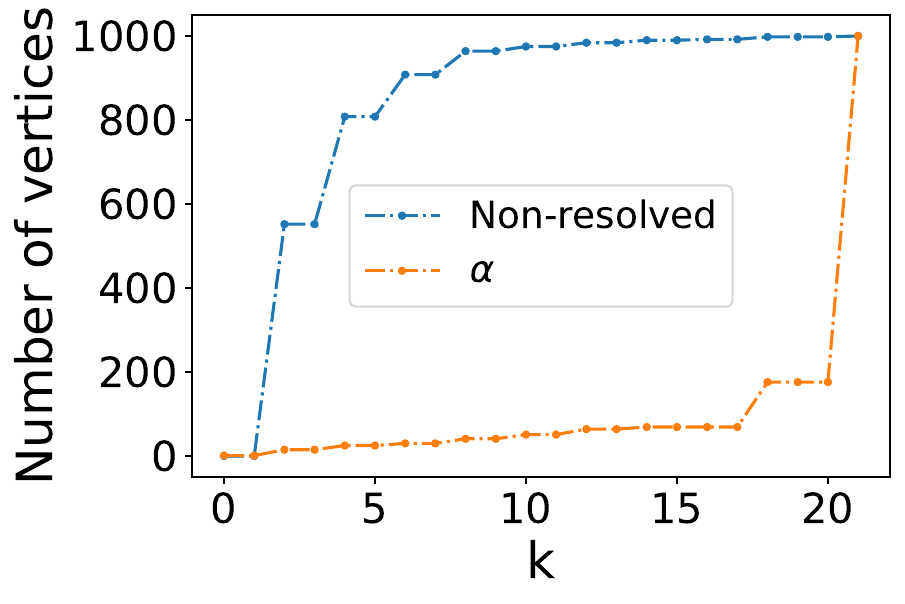}
        \caption{\Barabasi-Albert}
    \end{subfigure}
    \begin{subfigure}{0.32\textwidth}
        \includegraphics[width=\linewidth]{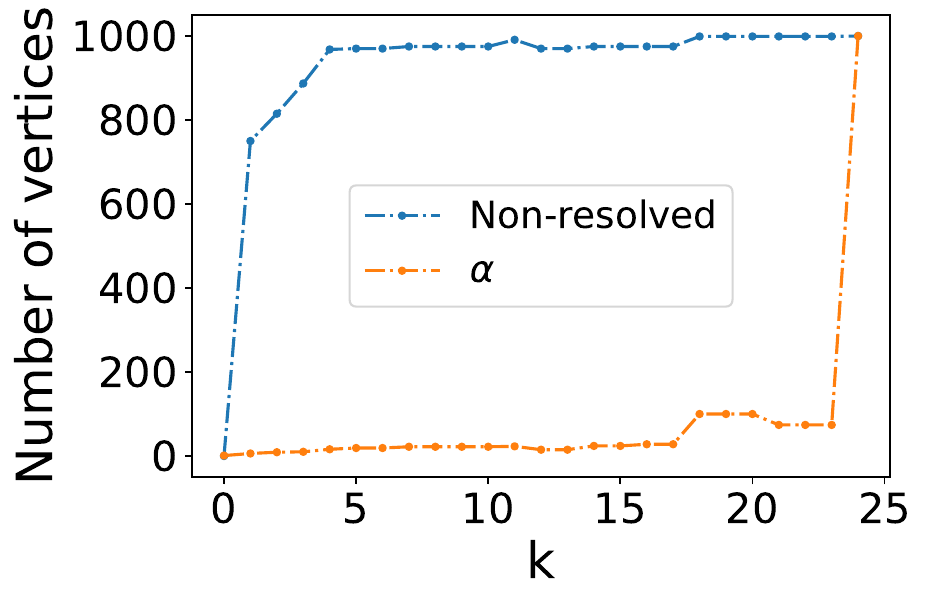}
        \caption{random geometric graph}
    \end{subfigure}
    \begin{subfigure}{0.32\textwidth}
        \includegraphics[width=\linewidth]{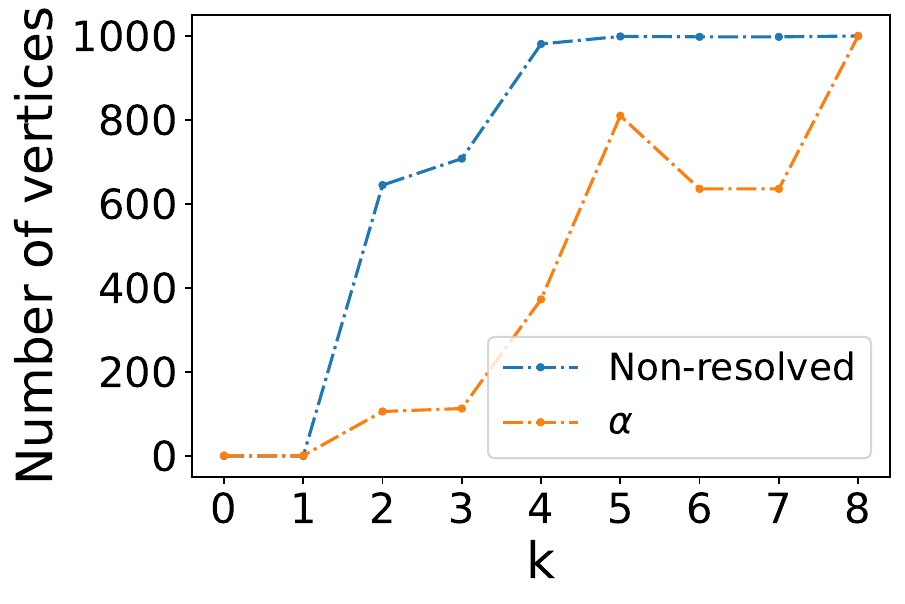}
        \caption{configuration model}
    \end{subfigure}
    \caption{Size of the largest equivalent class and number of non-resolved vertices as a function of the relaxation parameter, for one instance of various random graph models. For each graph, we vary $k$ from $0$ to the graph's diameter. Note that the functions may non-decreasing with $k$ as the resolving sets are obtained by Algorithm~\ref{algo} and hence may be sub-optimal.}
    \label{fig:non_resolved_vs_equiv_class}
\end{figure}

\subsubsection{Distribution of the size of the non-resolved equivalent classes}

We plot in Figures~\ref{fig:sizes_non_resolved_classes_BA} and~ \ref{fig:sizes_non_resolved_classes_RGG} the cardinality of the non-resolved equivalent classes. In most cases, the non-resolved classes have a small cardinality.

\begin{figure}[!ht]
 \centering
 \begin{subfigure}{0.32\textwidth}
  \includegraphics[width=\linewidth]{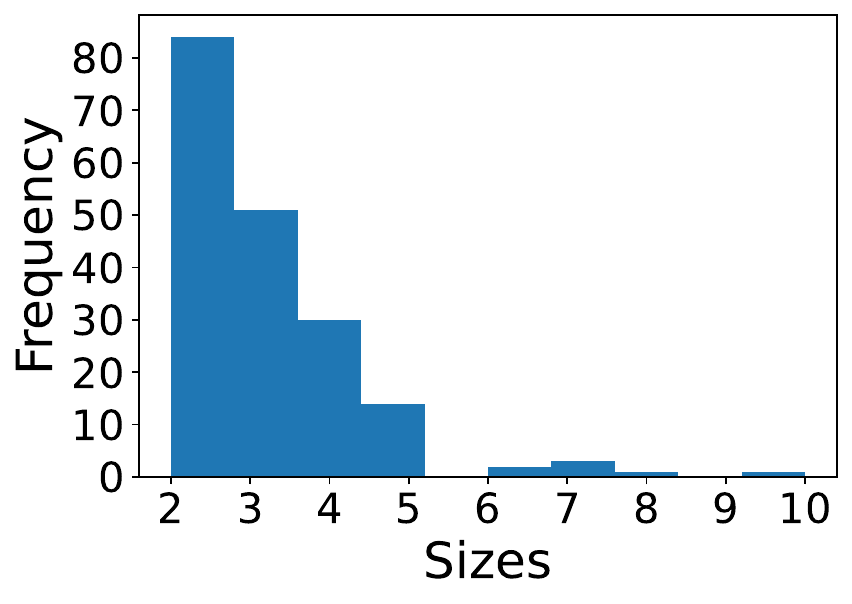}
  \caption{$k=2$}
 \end{subfigure}
 \begin{subfigure}{0.32\textwidth}
  \includegraphics[width=\linewidth]{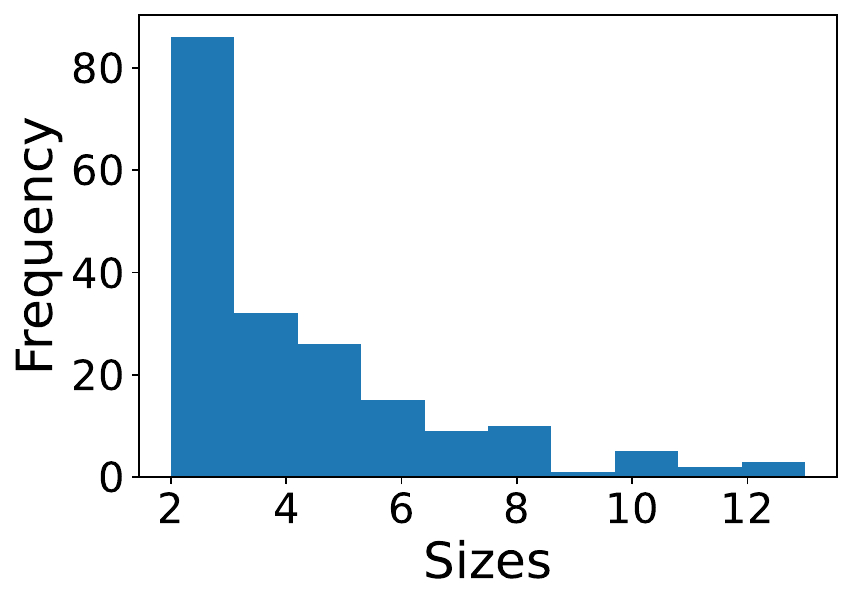}
  \caption{$k=4$}
 \end{subfigure}
 \begin{subfigure}{0.32\textwidth}
  \includegraphics[width=\linewidth]{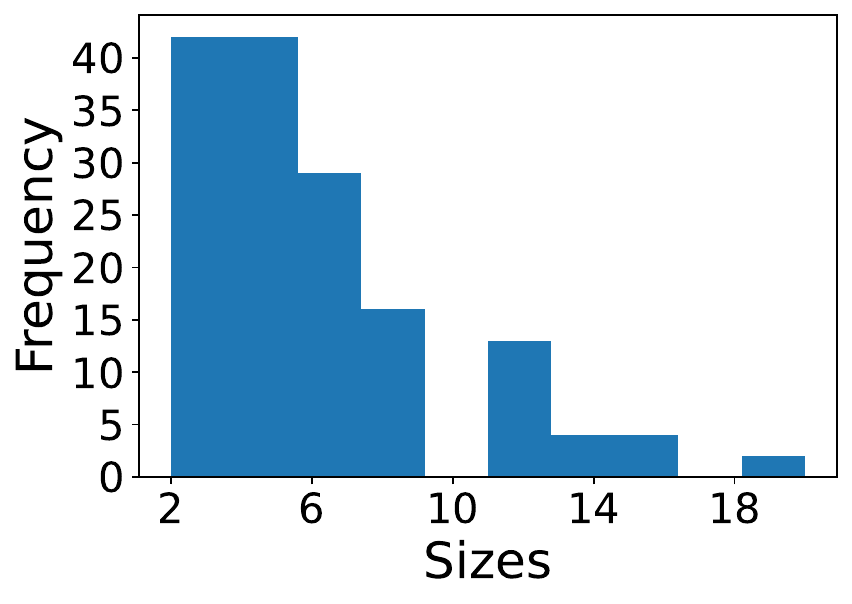}
  \caption{$k=6$ }
 \end{subfigure}
 \caption{Histogram of the sizes of the non-resolved equivalent classes on a \Barabasi-Albert random tree with $n=1000$ vertices for various values of the relaxation parameter $k$. Note that the histogram does not show the resolved vertices (who belong to equivalent classes of size 1).}
 \label{fig:sizes_non_resolved_classes_BA}
\end{figure}

\begin{figure}[!ht]
 \centering
 \begin{subfigure}{0.32\textwidth}
  \includegraphics[width=\linewidth]{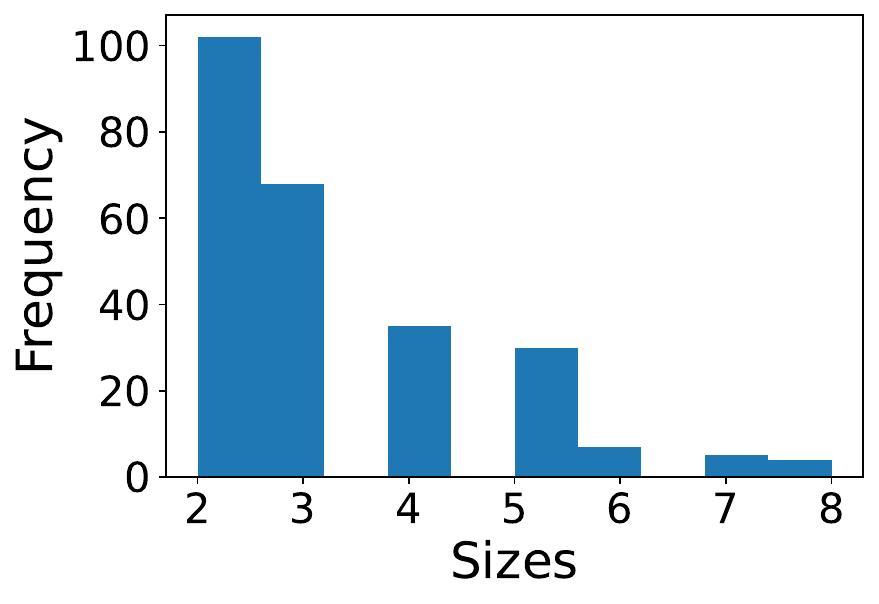}
  \caption{$k=2$}
 \end{subfigure}
 \begin{subfigure}{0.32\textwidth}
  \includegraphics[width=\linewidth]{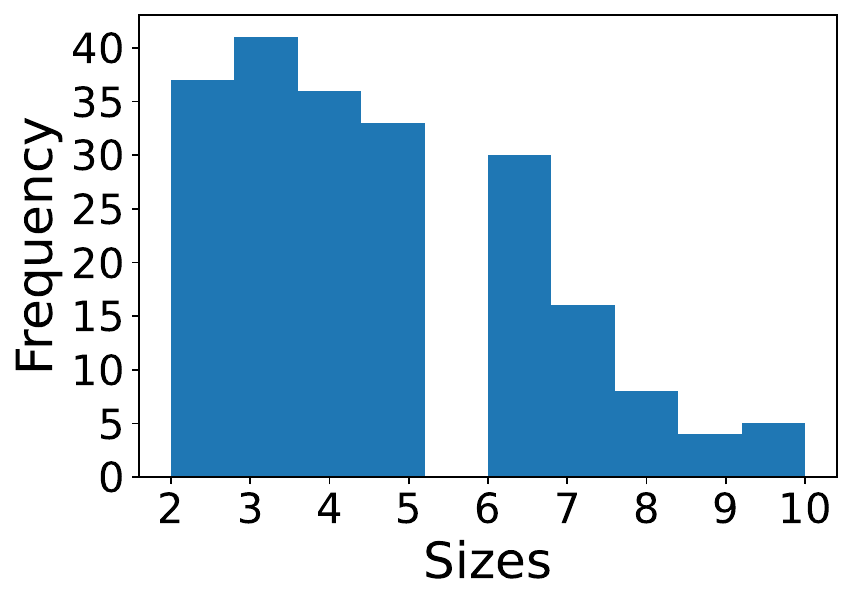}
  \caption{$k=4$}
 \end{subfigure}
 \begin{subfigure}{0.32\textwidth}
  \includegraphics[width=\linewidth]{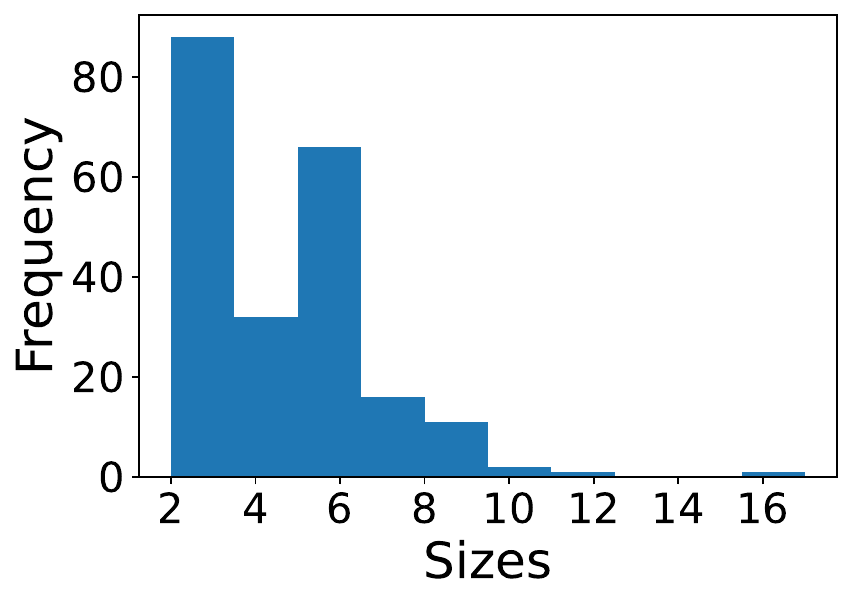}
  \caption{$k=6$ }
 \end{subfigure}
 \caption{Histogram of the sizes of the non-resolved equivalent classes on a random geometric graph with $n=1000$ vertices for various values of the relaxation parameter $k$. Note that the histogram does not show the resolved vertices (who belong to equivalent classes of size 1).}
 \label{fig:sizes_non_resolved_classes_RGG}
\end{figure}

\subsection{Metric Dimension for other Random Trees and Random Graphs}

\subsubsection{Galton-Watson trees}

In this section, we show that the observations made for \Barabasi-Albert random trees also apply to Galton-Watson trees. Figure~\ref{fig:evolution_md_GW} shows the evolution of the relaxed metric dimension and the number of non-resolved vertices, and Figure~\ref{fig:visualization_GW} shows a detail example on a Galton-Watson tree with $n=100$ vertices. 

\begin{figure}[!ht]
 \centering
 \begin{subfigure}{0.32\textwidth}
  \includegraphics[width=\linewidth]{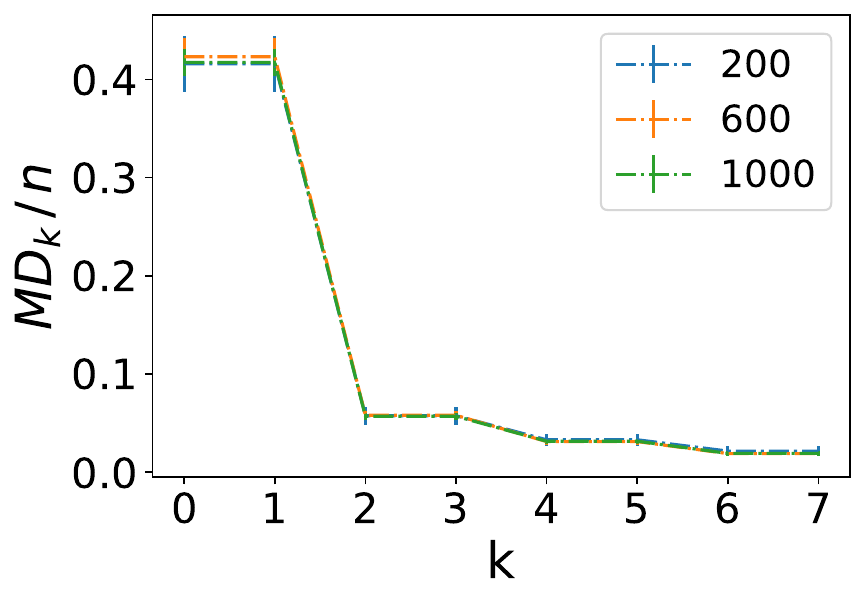}
  \caption{$\md_k / n$}
 \end{subfigure}
 \begin{subfigure}{0.32\textwidth}
  \includegraphics[width=\linewidth]{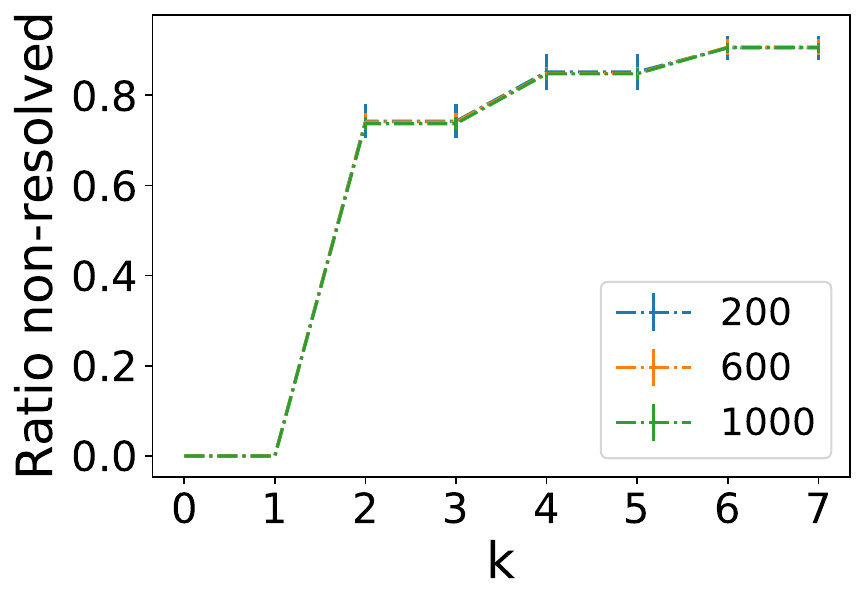}
  \caption{Ratio of non-resolved vertices}
 \end{subfigure}
 \begin{subfigure}{0.32\textwidth}
   \includegraphics[width=\linewidth]{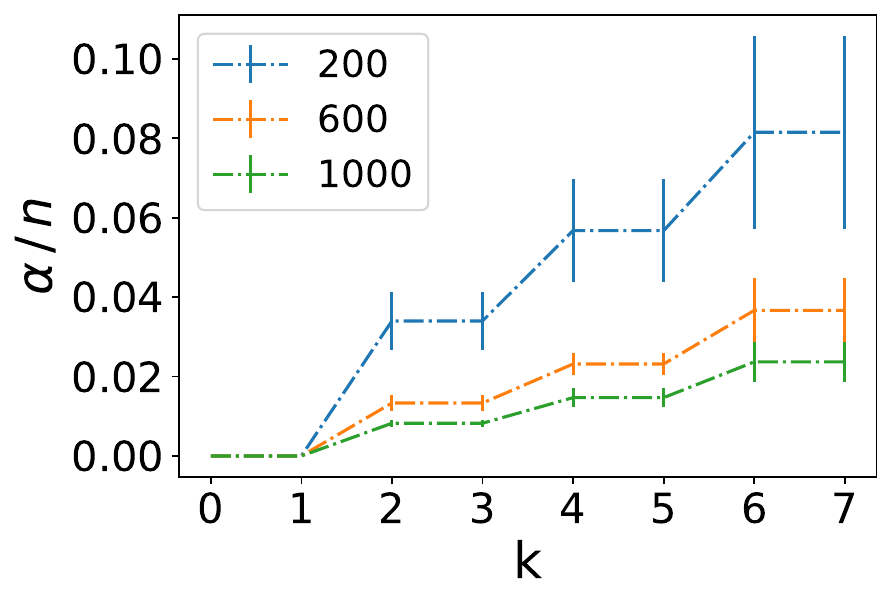}
  \caption{ $\alpha / n$ }
 \end{subfigure}
 \caption{Galton-Watson random trees conditioned on having $n$ vertices and with a Poisson offspring distribution of mean 3. Results are averaged over $20$ realizations, and error bars show the standard deviation.}
 \label{fig:evolution_md_GW}
\end{figure}

\begin{figure}[!ht]
    \centering
    \begin{subfigure}{0.32\textwidth}
        \includegraphics[width=\linewidth]{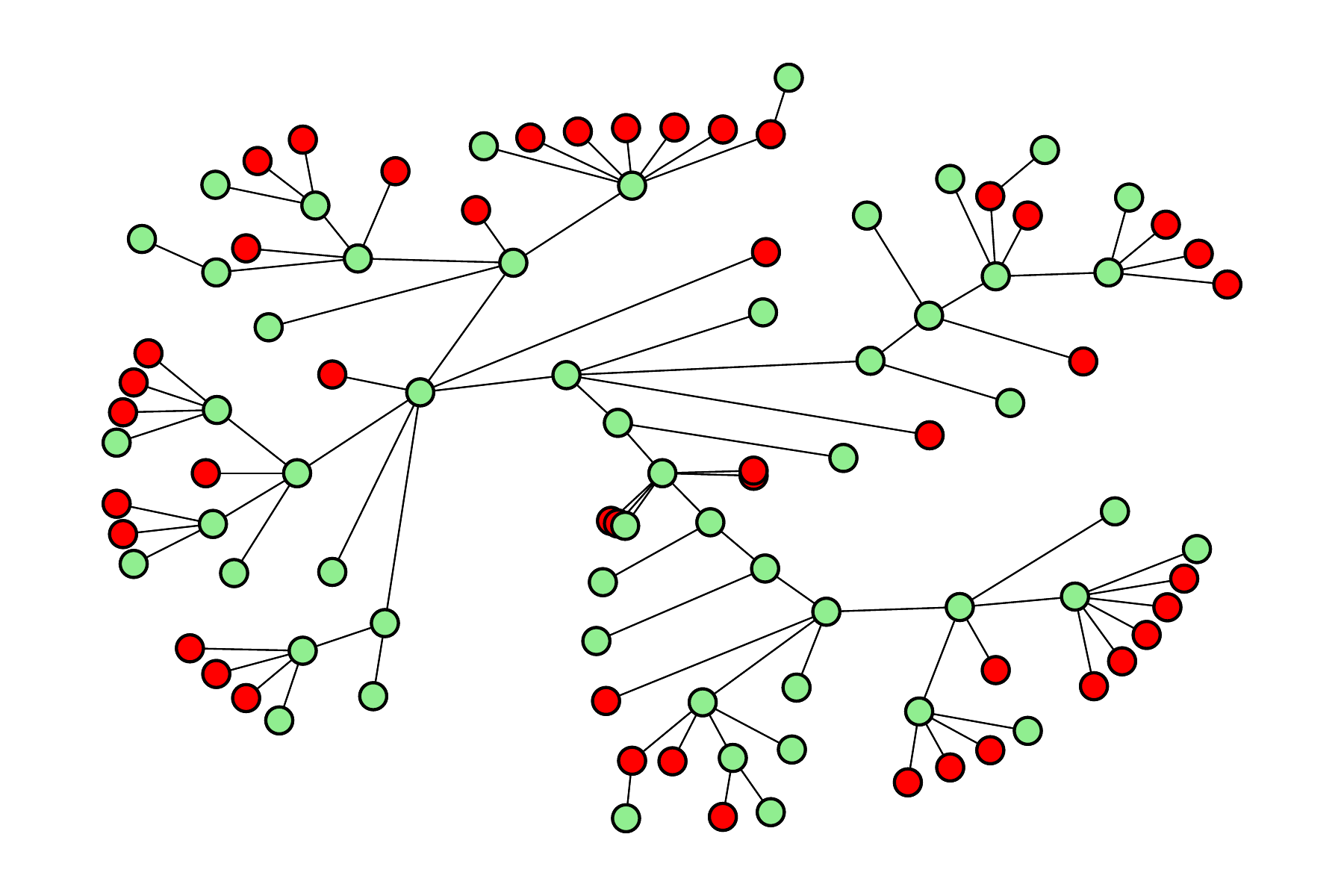}
        \caption{0-relaxed resolving set}
    \end{subfigure}
    \begin{subfigure}{0.32\textwidth}
        \includegraphics[width=\linewidth]{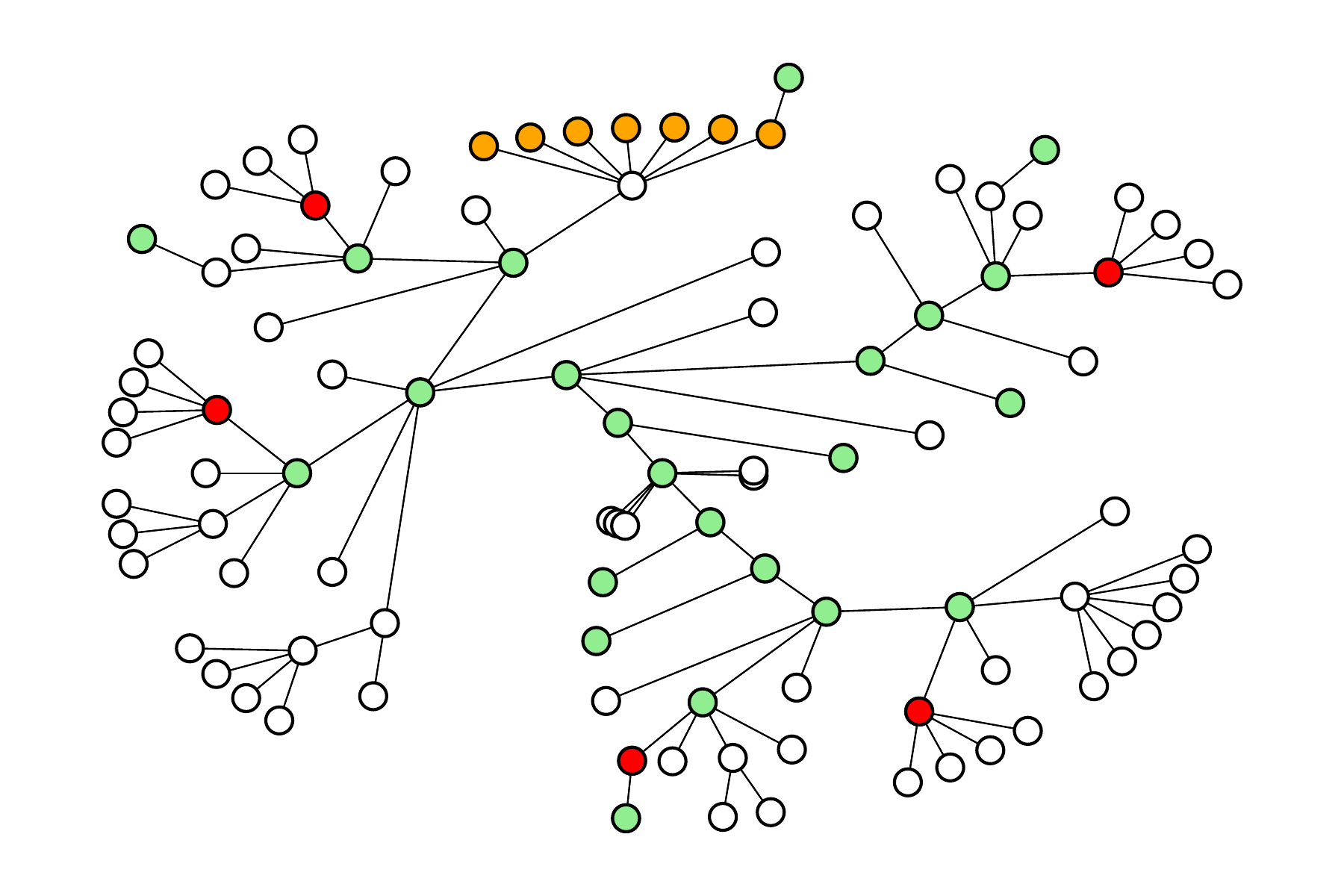}
        \caption{2-relaxed resolving set}
    \end{subfigure}
    \begin{subfigure}{0.32\textwidth}
        \includegraphics[width=\linewidth]{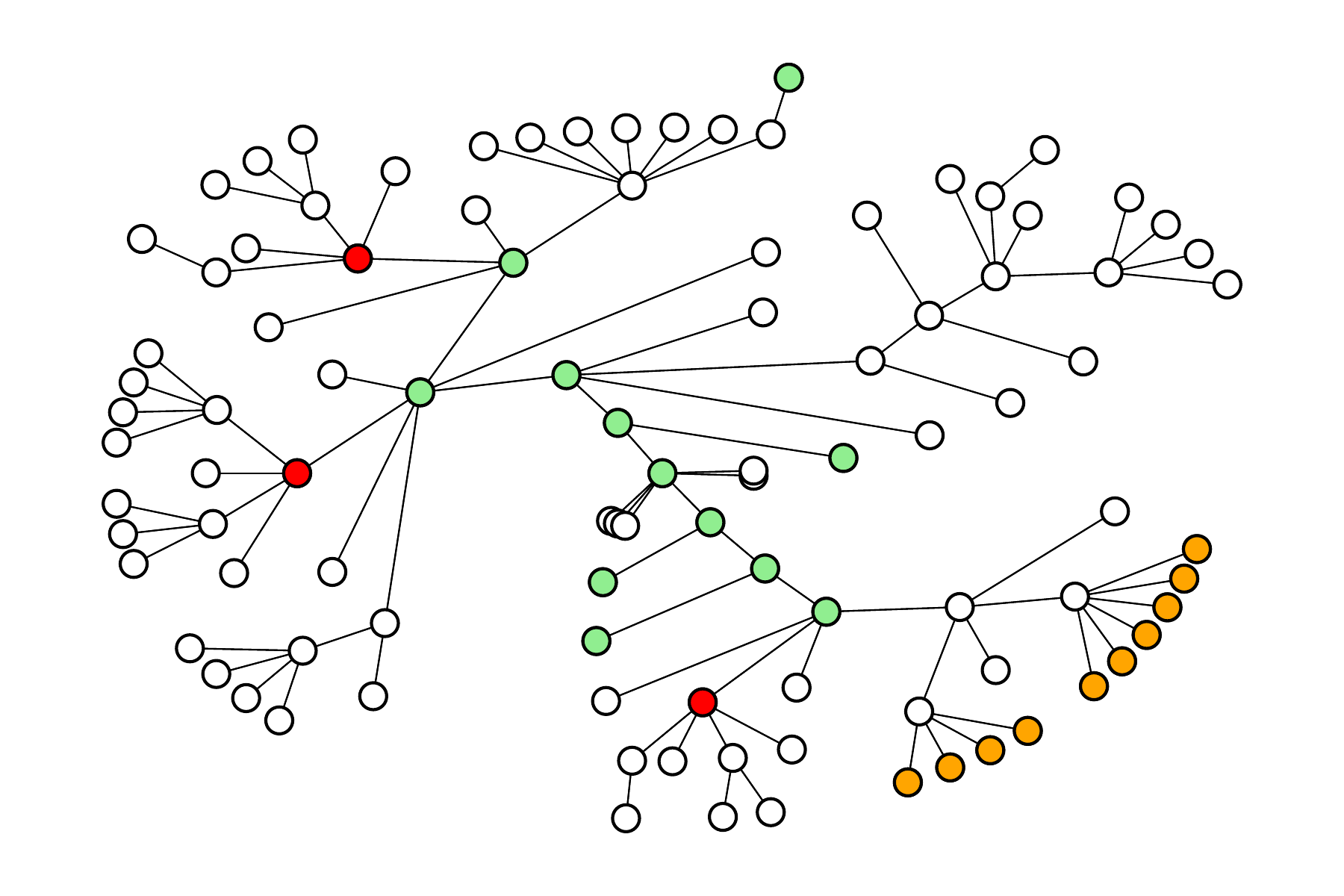}
        \caption{4-relaxed resolving set}
    \end{subfigure}
    \caption{Resolving sets obtained by Algorithm~\ref{algo} on a Galton-Watson random tree conditioned to have $100$ vertices. {\color{red}Red}: vertices belonging to the relaxed resolving set found by Algorithm~\ref{algo}. {\color{green}Green}: vertices with a unique identification vector. {\color{orange}Orange}: vertices belonging to the largest equivalent class of non-resolved vertices.}
    \label{fig:visualization_GW}
\end{figure}

\subsubsection{Configuration model}
We plot in Figure~\ref{fig:evolution_md_CM} the evolution of the three metrics as a function of the relaxation parameter $k$ for the configuration model. We observe that the metric dimension is initially small, and thus the relaxation has a smaller effect, albeit relaxing from $k=0$ to $k=2$ still diminishing by a factor 2 the number of sensors needed. Moreover, the number of non-resolved vertices increases less dramatically than in random trees or random geometric graph. Finally, the quantity $\alpha$ always remains small. 

\begin{figure}[!ht]
 \centering
 \begin{subfigure}{0.32\textwidth}
  \includegraphics[width=\linewidth]{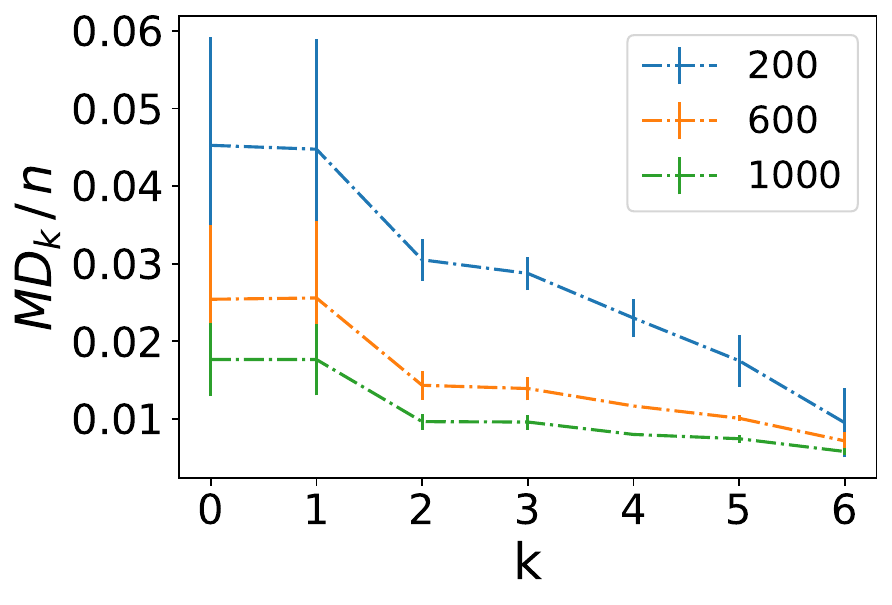}
  \caption{$\md_k / n $}
 \end{subfigure}
 \begin{subfigure}{0.32\textwidth}
  \includegraphics[width=\linewidth]{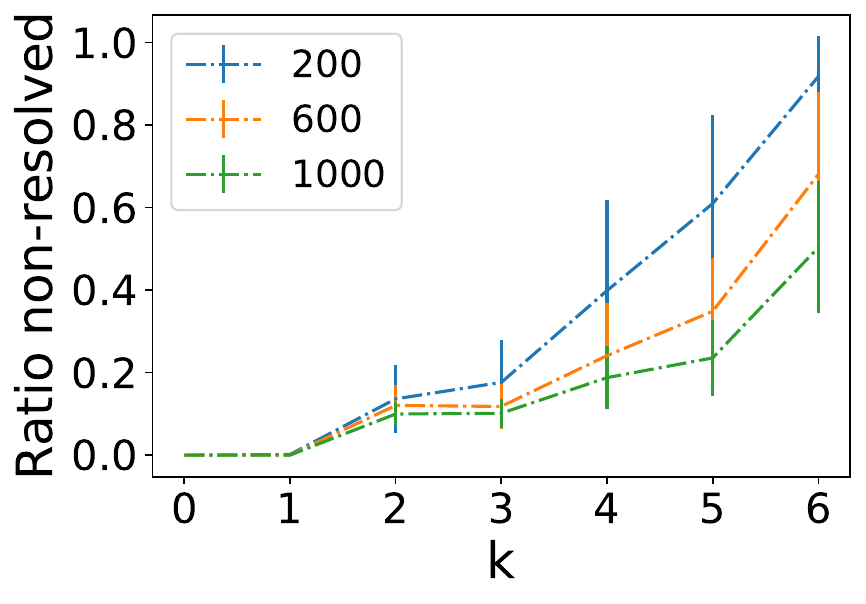}
  \caption{Ratio of non-resolved vertices}
 \end{subfigure}
 \begin{subfigure}{0.32\textwidth}
  \includegraphics[width=\linewidth]{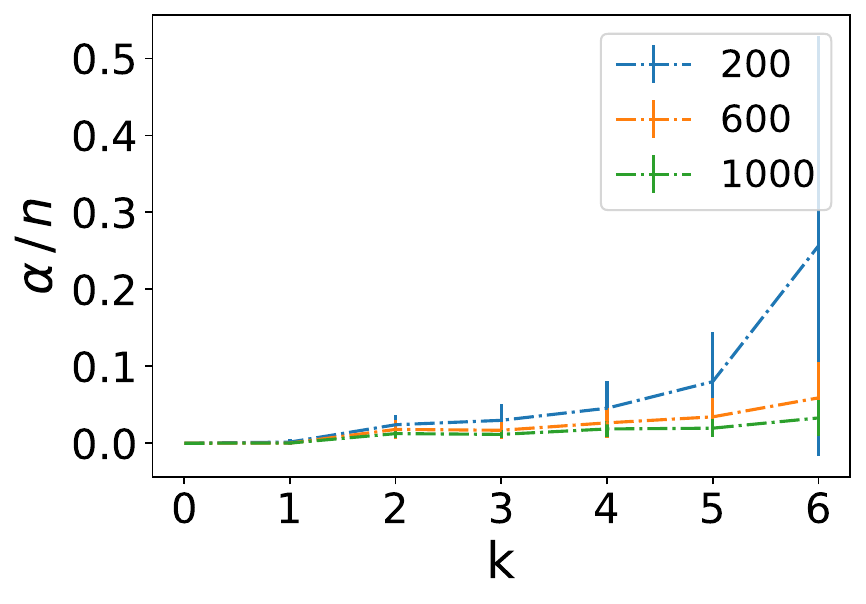}
  \caption{ $\alpha / n $ }
 \end{subfigure}
 \caption{Effect of the relaxation of the metric dimension on graphs sampled from the configuration model. Results are averaged over $20$ realizations, and error bars show the standard deviation.}
 \label{fig:evolution_md_CM}
\end{figure}

\subsection{Visualization of Relaxed Resolving Sets of Real Graphs}
Figure~\ref{fig:visualization_copencalls} shows the relaxed resolving set obtained by Algorithm~\ref{algo} on the \textit{copenhagen-calls} dataset. 

\begin{figure}[!ht]
    \centering
    \begin{subfigure}{0.32\textwidth}
        \includegraphics[width=\linewidth]{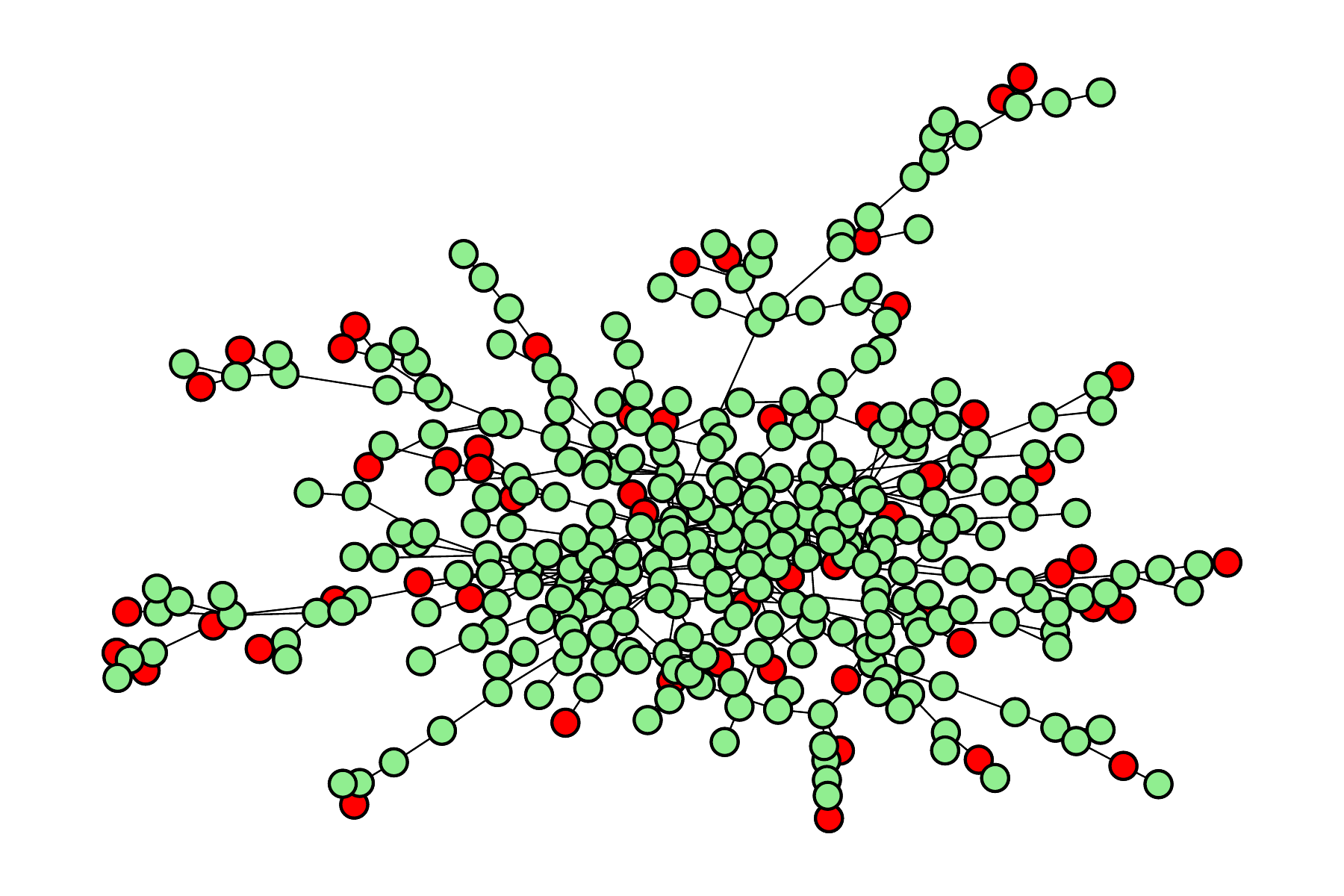}
        \caption{0-relaxed resolving set}
    \end{subfigure}
    \begin{subfigure}{0.32\textwidth}
        \includegraphics[width=\linewidth]{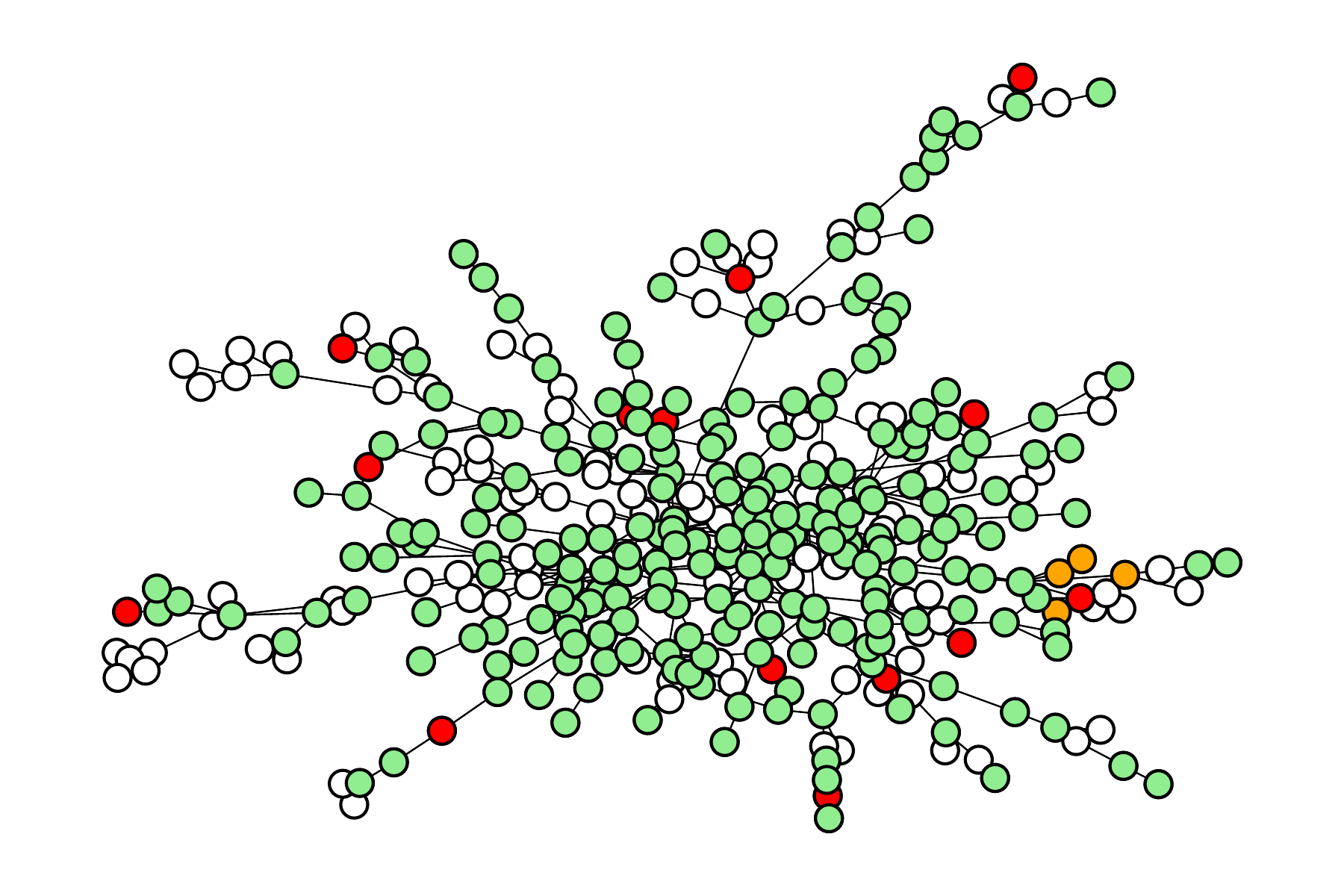}
        \caption{2-relaxed resolving set}
    \end{subfigure}
    \begin{subfigure}{0.32\textwidth}
        \includegraphics[width=\linewidth]{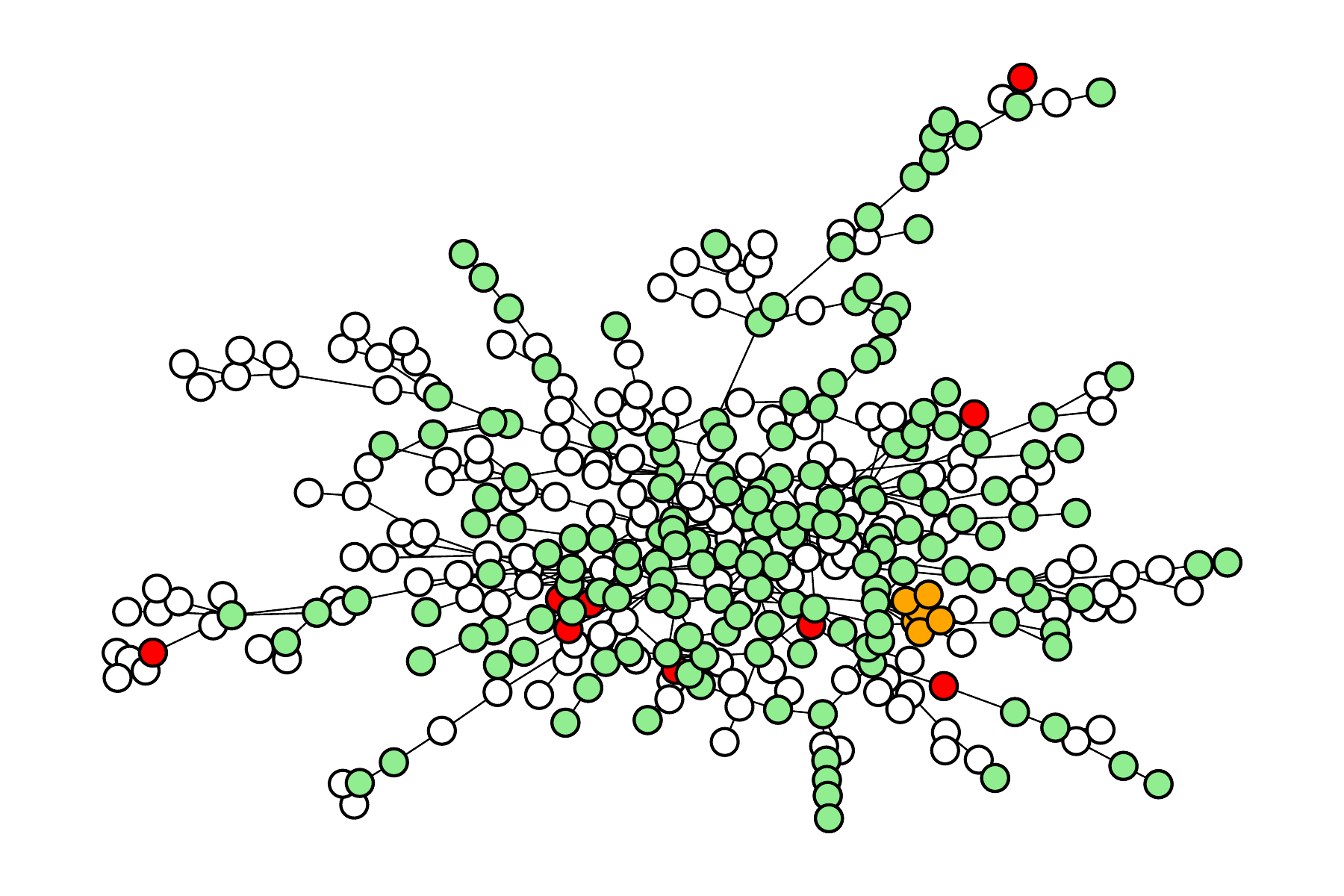}
        \caption{4-relaxed resolving set}
    \end{subfigure}
    \caption{Resolving sets obtained by Algorithm~\ref{algo} on \textit{copenhagen-calls} graph. {\color{red}Red}: vertices belonging to the relaxed resolving set found by Algorithm~\ref{algo}. {\color{green}Green}: vertices with a unique identification vector. {\color{orange}Orange}: vertices belonging to the largest equivalent class of non-resolved vertices.}
    \label{fig:visualization_copencalls}
\end{figure}